\documentclass{lmcs}
\usepackage{versions}\excludeversion{ignore}\excludeversion{exclude}
\usepackage{soul}\setstcolor{red}
\usepackage{todonotes}
\begin{ignore} 
<ccs2012>
   <concept>
       <concept_id>10003752.10003753.10010622</concept_id>
       <concept_desc>Theory of computation~Abstract machines</concept_desc>
       <concept_significance>500</concept_significance>
       </concept>
   <concept>
       <concept_id>10003752.10003790.10003798</concept_id>
       <concept_desc>Theory of computation~Equational logic and rewriting</concept_desc>
       <concept_significance>500</concept_significance>
       </concept>
 </ccs2012>
\end{ignore} 
\usepackage{hyperref}
\usepackage{latexsym,MnSymbol,stmaryrd} 
\usepackage{savesym}
\savesymbol{Cross}
\usepackage{bbding}
\usepackage{mathrsfs}
\usepackage{tabularx}
\usepackage{csquotes}\MakeOuterQuote{"}  
\usepackage{xspace}
\usepackage{graphicx}
\usepackage[pdf]{pstricks}
\psset{dotsep=1.2pt,dash=3pt 2pt}
\usepackage{pst-node,pst-plot}
\usepackage[arrow,matrix,tips,curve]{xypic}
\newtheorem{claim}[thm]{Claim}
\newcommand{\textbi}[1]{\textbf{\textit{#1}}}
\usepackage{mathtools}
\newcommand{\hide}[1]{}
\newcommand{\cA}{\mathcal{A}}

\renewcommand{\cD}{\mathcal{D}}
\newcommand{\cE}{\mathcal{E}}
\newcommand{\cF}{\Sigma} 

\newcommand{\cI}{\mathcal{I}}
\renewcommand{\cL}{\mathcal{L}}

\renewcommand{\cR}{\mathcal{R}}
\newcommand{\cS}{\mathcal{S}}

\newcommand{\cV}{\mathcal{V}}
\newcommand{\cX}{\Xi} 
\newcommand{\Nat}{\mathbb{N}}

\newcommand{\Var}[1]{\cV\!\textit{ar}\/({#1})}
\newcommand{\Vertex}[1]{\cV({#1})}

\newcommand{\Lab}[1]{\cL({#1})}  

\newcommand{\Dom}[1]{\cD\textit{om}({#1})}

\newcommand{\Ima}[1]{\cI\!\textit{ma}\/({#1})}
\newcommand{\bare}[1]{\dot{#1}}
\newcommand{\la}{\leftarrow}
\newcommand{\ra}{\rightarrow}
\newcommand{\da}{\downarrow}
\newcommand{\ua}{\uparrow}

\newcommand{\lrps}[2]{\mathop{\smash{{\longrightarrow}^{#1}_{#2}}}}
\newcommand{\rlps}[1]{\mathop{\smash{{_{#1}\!\!\longleftarrow}}}}

\newcommand{\rlpstr}[1]{\mathop{\smash{{_{#1}\!\!\la\!\!\!\!\longleftarrow}}}}
\newcommand{\drlpstr}[1]{\displaystyle{\rlpstr{#1}}}

\newcommand{\dlrps}[2]{\displaystyle{\lrps{#1}{#2}}}
\newcommand{\drlps}[1]{\displaystyle{\rlps{#1}}}

\newcommand{\lrpstr}[1]{\mathop{\smash{\longrightarrow\!\!\!\!\!\!\!\!\longrightarrow}_{#1}}}
\newcommand{\dlrpstr}[1]{\displaystyle{\lrpstr{#1}}}

\newcommand{\nf}[1]{#1{\downarrow}}

\newcommand{\she}{\cong_{\textit{sh}}}

\newcommand{\Acc}[1]{\cA\textit{cc}(#1)}

\newcommand{\roots}[1]{\cR({#1})} 
\newcommand{\spr}[1]{\cS(#1)}

\newcommand{\rest}[2]{#1\mathop{|}_{#2}}
\newcommand{\subd}[2]{#1\mathop{\downharpoonright}_{{#2}}} 
\newcommand{\cont}[2]{#1\mathop{\upharpoonright}_{{#2}}}
\usepackage{setspace}
\renewcommand{\red}[1]{\textcolor{red}{#1}}

\renewcommand{\blue}[1]{\textcolor{blue}{#1}}
\definecolor{cadmiumgreen}{rgb}{0.0, 0.42, 0.24}
\renewcommand{\green}[1]{\textcolor{cadmiumgreen}{#1}}

\newcommand{\wir}[2]{{#1\rightsquigarrow #2}}
\newcommand{\gwiring}[2]{#1_{#2}}
\newcommand{\wiring}[2]{{#1}_{#2}}
\newcommand{\targetw}[3]{#3(#1)}
\newcommand{\drag}[1]{\llbracket #1\rrbracket}
\newcommand{\targets}[3]{#3_0(#1)}
\def\edge(#1,#2){#1\longrightarrow#2}
\def\edgen(#1,#2,#3){#1\longrightarrow^{#2} #3}

\newcommand{\npred}{\textit{in}} 
\newcommand{\In}{\textit{in}} 

\newcommand{\X}[1]{X({#1})}  

\newcommand{\pred}{\textit{pred}}
\newcommand{\dpo}{\textsc{DPO}\xspace}
\newcommand{\omicron}{o}
\newcommand{\gtnat}{>_{\Nat}}
\newcommand{\target}[1]{\xi^{\raisebox{-1pt}{\scriptsize\rm !}}(#1)}
\def\ext(#1,#2){\langle{#1},{#2}\rangle}   
\def\int(#1,#2){({#1},{#2})}

\renewcommand{\st}{\mathrel:}
\newcommand{\IVertex}[1]{\cI nt{(#1)}} 

\title{Drag Rewriting}
\author[N. Dershowitz]{Nachum Dershowitz}
\address{School of Computer Science,   Tel Aviv University, Ramat Aviv, Israel}
\author[J.-P. Jouannaud]{Jean-Pierre Jouannaud}
\address{\'Ecole Normale Sup\'erieure de Paris-Saclay, France}
\author[F. Orejas]{Fernando Orejas}
\address{  Universitat Politècnica de Catalunya, Barcelona, Spain}

\begin{document}
\date{Draft of \today}

\begin{abstract}
We present a new and powerful algebraic framework for graph
rewriting, based on \emph{drags},
a class of graphs
enjoying a novel composition operator.
Graphs are embellished with roots and sprouts, which can be wired together to form edges.
Drags enjoy a rich algebraic structure with sums and products.
Drag rewriting naturally extends graph rewriting, dag rewriting, and term rewriting models.

\end{abstract}

\maketitle
\begin{center}
\begin{tabular}{lr}
\raisebox{-5cm}{\includegraphics[angle=0,width=0.45\textwidth]{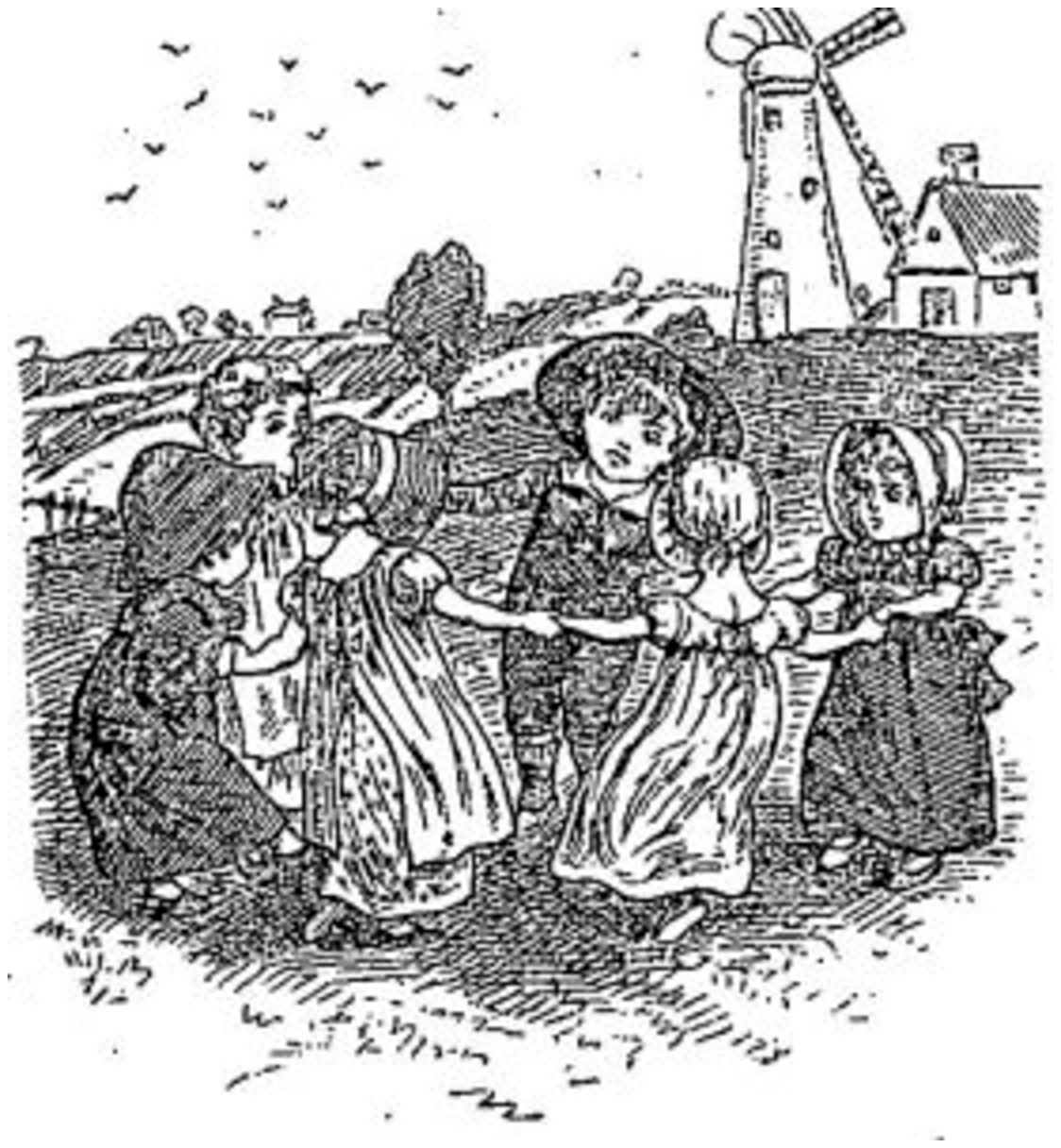}}
&
\begin{minipage}{0.5\textwidth}
\begin{verse}\it\large
Ring-a-ring o' roses,\\
A pocket full of posies,\\
A-tishoo! A-tishoo!\\
We all fall down.
\end{verse}
\bigskip
\raggedleft
---From: Kate Greenaway,\\
 \emph{Mother Goose or \\ The Old Nursery Rhymes} (1881);\\ illustration from\\ \textit{Harper's Young People} (1881)
\end{minipage}
\end{tabular}
\end{center}
\newpage
\tableofcontents
\newpage

\section{Introduction}
Rewriting with graphs has a long history in computer science, graphs
being used to represent data structures, as well as program structures
and even concurrent and distributed computational models. They  therefore play a key r\^ole in
program evaluation, transformation, and optimization, and more
generally in program analysis; see, for
example,~\cite{DBLP:journals/corr/abs-1802-05862}.

As rewriting graphs is very similar to rewriting algebraic terms, the same questions arise for both: What rewriting relation do we need?  
Is there an efficient pattern-matching algorithm? 
How can one determine if a
particular rewriting system is confluent and/or terminating? 
Does graph rewriting encompass term rewriting?

These questions have---for these reasons---been addressed by the rewriting community since the mid-seventies with two different research trends: 
(1) in research initiated by Hartmut Ehrig and his collaborators, either in the context of general graphs equipped with pushouts~\cite{DBLP:conf/focs/EhrigPS73} or for particular classes of graphs---specifically those that do not contain cycles~\cite{DBLP:conf/tagt/PlumpH94}; 
and (2) in research initiated by Henk Barendregt and his collaborators~\cite{Barendregtetal} in the context of term graphs, a generic terminology capturing various generalizations of terms.
Over the years, some effort was devoted to uniting these two trends (e.g.\@ \cite{Plump,CorradiniGadducci99}), 
aiming at answering the question to what extent term rewriting is covered. 
As part of these efforts, termination and confluence techniques have been elaborated for various generalizations of trees, such as rational trees, directed acyclic graphs, jungles, term-graphs, lambda-terms, and lambda-graphs, as well as for graphs in general.
See~\cite{DBLP:books/el/leeuwen90/Courcelle90,DBLP:conf/ac/Courcelle93}
for detailed accounts of these
techniques and~\cite{DBLP:conf/birthday/2018ehrig} for a survey of
implementations of various forms of graph rewriting and of available
analysis tools.

One important practical objective of many of these works is to make modeling of sharing within graph structures possible so as to enable formal proofs of sharing techniques used in programming languages, whether by means of terms or graphs. 
In this context, it becomes crucial to faithfully encode term rewriting as a form of graph rewriting. 
This question was first addressed in \cite{Barendregtetal}, and solved in the particular case of term graphs (directed graphs whose all vertices are accessible from one of them, and equipped with a function symbol whose arity dictates the number of their successors) and left-linear rewrite rules whose all critical pairs are trivial. 
This question has returned again and again in the literature on term-graphs, by extending the term-graph framework to more general graphs and rewrite rules \cite{CorradiniGadducci99}, by restricting the general graph framework to acyclic graphs, or by accepting cyclic graphs while restricting one to acyclic left-hand sides of rules to better capture term rewrites \cite{PlumpJungle}. 
One reason, among several, for the multiplicity of efforts is the difficulty of integrating non-linear rewrite rules within graph frameworks, a problem pointed out already in \cite{Barendregtetal}, but---as they argue---one that cannot be solved positively without substantial modifications to their framework. 
Nor has it been solved in the categorical framework, where its importance has also been stressed \cite{DBLP:conf/gg/Parisi-PresicceEM86}.

Drags, the alternative model considered here, are arbitrary directed graphs equipped with roots and sprouts that facilitate composition. Internal vertices, as in term graphs, are labeled by function symbols having a fixed arity. Sprouts are non-internal vertices without successors, labeled by variables.
Our framework is able to faithfully encode term rewriting without restriction, hence finally solving this old question positively. 
An earlier version of drags paved the way for the present work \cite{DBLP:journals/tcs/DershowitzJ19}. The current, substantially revised version has the potential to provide an interesting alternative to the double pushout (\dpo) framework for graph rewriting. 
In addition to its conceptual simplicity, one advantage of drags over DPO is the absence of dangling edges caused by rewriting.

Intuitively, drags can be viewed as networks of processing units that accept a given (finite) number of data as inputs and deliver data as outputs  that can be sent over an arbitrary (finite) number of one-way channels.
Channels can of course be duplicated,
allowing thereby for arbitrary sharing. There is an order among
the input channels of a processing unit so as to appropriately
discriminate among the inputs. 
Inputs enter a drag at its sprouts, which are vertices without successors, labeled by variables.
Outputs exit the drag at its roots, which are incoming edges with no source. Duplication is modeled by multiplicity of a given root vertex and of a given variable labeling several sprouts.
Thus, drags differ from ordinary directed labeled graphs in their distinguishing of roots and sprouts.

The generality of this graph model requires that these one-way channels connect the processing units in an arbitrary way via their respective roots and sprouts. Connecting two drags defines a composition operator so that matching a left-hand side of rule $L$ in a drag $D$ amounts to writing $D$ as the composition of a context graph $C$ with $L$, and rewriting $D$ with the rule $L\ra R$ amounts to replacing $L$ with $R$ in that composition. In the case of drags, composition plays therefore the r\^ole of both context grafting and substitution in the case of trees. 
The same applies to equations, allowing us to define the notion of a congruence generated by a set of drag equations in the usual way,
with rewriting providing, as for terms, a way to decide whether two drags are congruent.
In sharp contrast with term rewriting, on the other hand, drag rewriting takes advantage of the underlying graph model: subdrags shared by $L$ and $R$ need not be removed before  being (re-) generated. This holds for sprouts as well. 
This is similar to the \dpo approach to rewriting, which specifies which vertices are to be removed, which are to be (re-) generated, and for which variables are absent.
The \dpo formalism applies to many different categories of graphs. 
In our model, such specifications follow the determination of the user that left-hand and right-hand sides of a rule do or do not share specific subgraphs.

\section{The Drag Model} 
\label{s:drags}

Drags were introduced in~\cite{DBLP:journals/tcs/DershowitzJ19,DJlpar} and developed further in~\cite{JO22a}.
Retaining the moniker, we introduce here a completely new version, which is much better behaved and allows for easier generalization, while at the same time is more simply defined.

Drags are finite \textbi{d}irected \textbi{r}ooted l\textbi{a}beled ordered
multi-\textbi{g}raphs.  
Some vertices with no outgoing edges are designated \emph{sprouts}.
Other vertices are \emph{internal}.  
We presuppose the following: a set of function symbols $\cF$, whose elements---equipped with a fixed arity---serve as labels for internal vertices; and a set of nullary variable symbols $\cX$, disjoint from $\cF$, used to label sprouts. 

We make considerable use of multisets.
A \emph{finite multiset} is a function from a finite base set $E$ to the set $\Nat$ of natural numbers. Finite multisets are often enumerated as in $\{a,a,b,a,b,c\}$. 
\emph{Finite sets} are finite multisets all of whose elements occur precisely once as in $\{a,b,c\}$.
A \emph{sub-multiset} $N$ of a multiset $M$, written $N\subseteq M$, gives a number of elements $N(e)\leq M(e)$ for each element $e$ of the base $E$.
Given two finite multisets $M$ over $E$ and $N$ over $F$, a (total) \emph{multi-map} $f:M\to N$ is some (dependent) function $f^+:(a,m)\mapsto (b,n)$, $0<m\leq M(a)$, $0<n\leq N(b)$, so that $f(e)$ will denote the multiset of elements $\{f^+(e,1),\ldots, f^+(e,M(a))\}$.  
(The second arguments $m,n$ of the function description $f^+$ serve as indices of the multiple $M(a),N(b)$ occurrences of the first arguments $a,b$).
The multiset $M$ is the \emph{domain} of $f$, also denoted $\Dom{f}$, while the multiset $N$ is its \emph{codomain} or \emph{image}, denoted $\Ima{f}$. 
A multi-map $f:M\to N$ is \emph{partial} if the associated map $f^+$ is partial, and \emph{multi-injective} (\emph{multi-equijective}) if the associated map $f^+$ is injective (bijective, respectively). 
Note that a partial multi-injective multi-map becomes multi-equijective as a multi-map from its domain to its image.
If $M,N$ are finite sets, then $f$ is a classical injective (bijective) map from $M$ to $N$. Finally, a map $f:E\to F$ extends to a finite multiset $\{a_i\}_i$ over $E$ as the multi-map returning the finite multiset $\{f(a_i)\}_i$ over $F$.

To ameliorate notational burden, we  use vertical bars $|\_|$
to denote various quantities, such as length of lists, size of expressions, of sets or multisets, and even the arity of function symbols. We use
$\varnothing$ for an empty list, set, or multiset, $\cup$ for both set and multiset union (which takes the maximum of multiplicities for multisets), $\cap$ for set and multiset intersection (minimum for multisets),
$\setminus$ for set and multiset difference  (natural subtraction for multisets), $\uplus$ for disjoint union (which adds multiplicities), and $\in$ for membership ($a\in M$ iff $M(a)>0$ for multiset $M$).
We will identify a singleton set or multiset with its single element to avoid unnecessary clutter.
So, for example, $a_0\cup \{a_i\}_{i=1}^{i=n}=\{a_i\}_{i=0}^{n}$. 

\begin{defi}[Drag]
\label{d:drag}
A \emph{drag} $D$ is a tuple $\langle V, R, L, X, S\rangle$,
where
\begin{enumerate}
\item
$V$ is a finite set of \emph{vertices} (vertices have a \emph{name});
\item
$R:V\to \Nat$ is a finite, possibly empty, multiset of vertices, called \emph{roots}; when $R(v)>0$, vertex $v$ is   \emph{rooted}; when $R(v)=0$ it's \emph{rootless}; sprouts can be rooted;
\item
$S\subseteq V$ is a set of \emph{sprouts}, leaving $I=V\setminus S$ to be the \emph{internal} vertices;
\item
$L: V \to \cF\cup\cX$ is the \emph{labeling} function, mapping
  internal vertices $I$ to labels from  vocabulary $\cF$ and sprouts $S$ to labels from  vocabulary $\cX$ (vertices have a \emph{label});
\item $X: V\ra V^*$ is the \emph{successor function}, mapping each vertex $v\in V$ to a list of vertices in $V$ whose length equals the arity of its label, that is, $|X(v)|=|L(v)|$. Sprouts have no successors.
\end{enumerate}
  The pair $\int(R,S)$ of roots and sprouts is  the \emph{interface} of drag $D$.
\end{defi}
\noindent
Drags are accordingly based on ordered multigraphs with roots. 

\begin{defi}[Linear; bare; ground; empty; disjoint]
A drag $D$ is \emph{linear} if no two sprouts in $\spr{D}$ have the same label, \emph{closed} if it has no sprout ($\spr{D}=\varnothing$), \emph{bare} if it has no roots, \emph{ground} if it has neither root ($\roots{D}=\varnothing$) nor sprout, and \emph{empty} (denoted by $\varnothing$) if it has no vertices at all ($\Vertex{D}=\varnothing$). We denote by $\bare{D}$ the bare drag obtained from $D$ by removing all its roots.
Two drags are \emph{disjoint} if they share neither vertex nor variable.
\end{defi}

Here is an example of a linear drag with three internal vertices (blue and yellow), one (red) root and two (green) sprouts: 
\[
\psset{unit=3mm}
\begin{pspicture}(1,-1.4)(9,9.8)
\cnodeput[linecolor=white,fillcolor=white]{0}(-1,0){Ac}{}
\cnodeput[fillstyle=solid,fillcolor=green]{0}(4.5,8.7){Bc}{$x$}
\cnodeput[fillstyle=solid,fillcolor=green]{0}(10,0){Cc}{$y$}
\cnodeput[fillstyle=solid,fillcolor=cyan]{0}(1.5,1.5){Ao}{$a$}
\cnodeput[fillstyle=solid,fillcolor=yellow]{0}(4.5,6.0){Bo}{$o$}
\cnodeput[fillstyle=solid,fillcolor=yellow]{0}(7.5,1.5){Co}{$o$}
\nccurve[linecolor=black,linestyle=solid,angleA=90,angleB=180]{->}{Ao}{Bo}
\nccurve[linecolor=black,linestyle=solid,angleA=0,angleB=90]{->}{Bo}{Co}
\nccurve[linecolor=black,linestyle=solid,angleA=270,angleB=270]{->}{Co}{Ao}
\ncline[linecolor=red,linestyle=solid]{|->}{Ac}{Ao}
\ncline[linecolor=black,linestyle=solid]{<-}{Bc}{Bo}
\ncline[linecolor=black,linestyle=solid]{<-}{Cc}{Co}
\end{pspicture}
\]


\begin{defi}[Accessibility]
Drags are directed:
If $b$ is the $k$th vertex in the list $X(a)$ of successors of vertex $a$ of drag $D=\langle V, R, L, X, S\rangle$, then $\edgen(a,k,b)$ is a directed \emph{edge} with \emph{tail} $a$ and \emph{head} $b$, $k$ being usually omitted.
The reflexive-transitive closure $X^*$ of the successor relation $X$ is called \emph{accessibility}. 
We also write $a X b$, $a\rightarrow b \in X$, or just $\edge(a,b)$, as well as $a X^* b$, $a\rightarrow b\in X^*$, or $a\longrightarrow^* b$.
\begin{enumerate}
\item A vertex $v$ is said to be \emph{accessible  from} vertex $u$, and likewise that $u$ \emph{accesses} $v$, if $u X^* v$. 
\item Two vertices are \emph{unrelated} if neither is accessible from the other.
\item
Accessibility extends to sets as expected, denoting  the set of vertices of $D$ that are accessible from any vertex in $W\subseteq V$ by $X^*(W)$. 
\item A vertex $v$ is \emph{accessible} (without qualification) if it is accessible from some root, that is if $v\in X^*(r)$ for some rooted vertex $r$.
\item A \emph{path} of \emph{length} $n$ is a sequence $u_0,\ldots, u_n$ of vertices such that $\forall i\in[0..n-1]\st \edge(u_i,u_{i+1})\in X$. 
The path is \emph{trivial} if $n=0$. 
\item A \emph{cycle} is a (non-trivial) path such that $s_n=s_0$. 
A \emph{loop} is a cycle of length one.
\end{enumerate}
\end{defi}

\begin{defi}[Connected component]
Given a drag $D$ with successor relation $X$, a \emph{connected component} is a subdrag of $D$ generated by a set $W$ of vertices closed under predecessor, ancestor, and equal labeling of sprouts:
$\forall u\forall v \st  uXv$ we have $ u\in W $ iff $ v\in W$, and
$\forall s:x \forall t:x$ we have $s\in W $ iff $ t\in W$.
\end{defi}

\begin{defi}[Predecessor; indegree]
We denote by $\pred(v,D)$, or simply $\pred(v)$, the number of incoming edges to $v$ in drag $D$, and by $\npred(v,D)$, or simply
$\npred(v)$, called the \emph{indegree} of $v$,
the number of its incoming edges plus the number of times $v$ is a root of $D$: $\npred(v,D)= \pred(v,D) + R(v)$.  

A vertex is a \emph{source} of a drag if it has no predecessor, is a \emph{sink} if it has no successor, and is \emph{isolated} if it has neither predecessor nor successor, hence is both a source and a sink.

A \emph{tree} is a ground drag  all vertices of which have one predecessor, except for a lone vertex with none.
A \emph{forest} is a ground drag with no cycle, all vertices of which have zero or one predecessor. A \emph{dag} is a ground drag sans cycles.
\end{defi}

Here is a ground drag that is not a dag:
\[
\psset{unit=5mm}
\begin{pspicture}(2.,1.5)(8,5.5)
\cnodeput[fillstyle=solid,fillcolor=cyan]{0}(3,2.4){Aa}{$a$}
\cnodeput[fillstyle=solid,fillcolor=cyan]{0}(4.5,4.5){Ba}{$a$}
\cnodeput[fillstyle=solid,fillcolor=cyan]{0}(6,2.4){Ca}{$a$}
\ncline[linecolor=red,linestyle=solid,linewidth=1pt]{->}{Aa}{Ba}
\ncline[linecolor=red,linestyle=solid,linewidth=1pt]{->}{Ba}{Ca}
\ncline[linecolor=red,linestyle=solid,linewidth=1pt]{->}{Ca}{Aa}
\end{pspicture}
\]

\begin{rem}\label{rem:roots}
Two other natural data structures are possible for the roots of drags: lists with repetitions and sets. 
Sets imply that arbitrarily many output channels can be given access to a given root.
Multisets put a precise bound on that number and have  slightly better algebraic behavior compared to lists with repetitions, which were used in~\cite{DBLP:journals/tcs/DershowitzJ19}. 
Completing the set of natural numbers with an infinite value allows one to easily encode set-like behavior by a multiset, another reason for our choice here.
\end{rem}

\begin{rem}
Another important difference vis-\`a-vis~\cite{DBLP:journals/tcs/DershowitzJ19} is that we now also consider drags with inaccessible vertices, such as ground drags \emph{all} of whose vertices are inaccessible.
\end{rem}

\paragraph*{\bf Notations:}

\begin{itemize}
\item
When convenient, a drag $D=\langle V, R, L, X, S\rangle$ will be denoted $$\langle \Vertex{D}, \roots{D},\spr{D}, \Lab{D}, \X{D} \rangle$$ with $\IVertex{D}$ being its internal vertices, $\Acc{D}$ its set of accessible vertices, $\Var{D}$ the set of variables labeling its sprouts, and $\Dom{R}$ the domain of the partial function $R$.
\item
We write $r^{[n]}\in R$ to indicate that there are $n$ copies of the rooted vertex $r$ in the multiset $R$, that is, $R(r)=n$, but also $r\in R$, considering $R$ as the set of rooted vertices. A root may  also be seen as an edge without designated tail; so we will sometimes, in particular in figures, use the notation $\edge(,r)$ 
to indicate that vertex $r$ is rooted.

\item
We write $u:f$ when the vertex $u$ (possibly a sprout) has label $f$ (possibly a variable). The labeling function extends to lists, sets, and multisets of vertices as expected. 

\item
In examples, we will often name vertices by their label, when the intention is clear. In case of ambiguity, we will index the label by a positive number, so that $f(x,x)$ has vertices $f$, $x_1$, and $x_2$. Combining these notations, 
$f^{[2]}(x,x^{[1]})$ has now two roots at vertex $f$ and one root at vertex $x_2$. In that case, we will alternatively write $f^{[2]}(x_1,x_2^{[1]})$.

\item
In pictures, roots will be illustrated by incoming arrows, possibly indexed by a natural number indicating their multiplicity.
\end{itemize}

\section{Contexts and Subdrags}

\begin{defi}[Subdrag; context]
\label{d:subd}
Let $D$ be the drag $\langle V,R,L,X,S\rangle$, and let $W\subseteq V$.
We define the following notions:
\begin{enumerate}
    \item The \emph{restriction} $\rest{D}{W}$ of drag $D$ to vertices $W$ is the drag 
    \[D'=\langle W\cup S',R',L',X',(W\cap S)\cup S'\rangle\]
    where 
    \begin{enumerate}
         \item  
        $S'=\{s_v \st v\in V\setminus W,\, \exists w\in W \st w X v\}$ are new sprouts, 
       the $s_v$ being  new vertices;
               \item  $R'(w)=R(w) + \Sigma_{\,\edgen(v,k,w)\in X, v\in V\setminus W} k$, for each $w\in W$
   (hence $\In(w,D')= \In(w,D)$);
        \item $L'$ coincides with  $L$ on $W$,
         while $L'(s_v)= x_v$ for each $s_v\in S'$, where
        $x_v$ is a fresh variable;
    \item $X'$ coincides with $X$ on $W$, and  for each $s_v\in S'$,
        $ \edgen(u,k,s_v)\in X'$ iff $\edgen(u,k,v)\in X$.
            \end{enumerate}
    \item The \emph{subdrag} $\subd{D}{W}$
   of $D$ \emph{generated} by $W$ is the restriction of $D$ to the set of all vertices accessible from $W$. That is, $\subd{D}{W}=\rest{D}{X^*(W)}$. The subdrag is \emph{void} when $W=\varnothing$, \emph{trivial} when $X^*(W)=V$
   (as for $\subd{D}{R}$ because all vertices of $D$ are accessible from $R$), and \emph{strict} when not trivial. 
    \item The \emph{context} $\cont{D}{W}$ of $W$ in $D$ 
is the restriction of $D$ to the set of vertices that are inaccessible from $W$, viz.\@ $\rest{D}{V\setminus X^*(W)}$.
\end{enumerate}
\end{defi}

Let $D$ a drag reduced to a single internal vertex $v$ and edge $\edge(v,v)$. Then, the restriction of $D$ to $v$ is $D$ itself.
Examples of drags, subdrags, and context drags are shown in Figure~\ref{f:subd}.

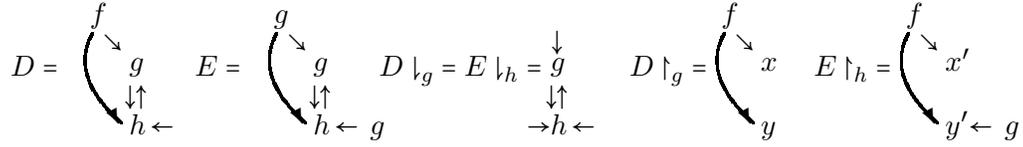
\begin{figure}[!t]
\setlength{\unitlength}{0.7cm} 
\hspace*{5mm}
\begin{picture}(17,4)(0,0)   
\thicklines
\put(-1,2.5){
\begin{picture}(0,3.25)(0,0)
\put(-0.5,-0.5){$D=$}
\put(1,0.5){\textcolor{black}{$f$}}
\put(1.2,0){\textcolor{black}{$\searrow$}}
\put(1.75,-0.4){\textcolor{black}{$g$}}
\put(1.65,-1.05){\textcolor{black}{$\da$}}
\put(1.85,-1.05){\textcolor{black}{$\ua$}}
\put(1.75,-1.6){\textcolor{black}{$h$}}
\put(2.15,-1.6){$\la$}
\qbezier(1.05,0.3)(0.6,-0.5)(1.6,-1.4)
\put(1.55,-1.43){\vector(1,-2){0}}
\end{picture}
}

\put(2.5,2.5){
\begin{picture}(0,3.25)(0,0)
\put(-0.5,-0.5){$E=$}
\put(1,0.5){\textcolor{black}{$g$}}
\put(1.2,0){\textcolor{black}{$\searrow$}}
\put(1.75,-0.4){\textcolor{black}{$g$}}
\put(1.65,-1.05){\textcolor{black}{$\da$}}
\put(1.85,-1.05){\textcolor{black}{$\ua$}}
\put(1.75,-1.6){\textcolor{black}{$h$}}
\put(2.15,-1.6){$\la\,g$}
\qbezier(1.05,0.3)(0.6,-0.5)(1.6,-1.4)
\put(1.55,-1.43){\vector(1,-2){0}}
\end{picture}
}

\put(7,2.5){
\begin{picture}(0,3.25)(0,0)
\put(-1.5,-0.5){$\subd{D}{g}=\subd{E}{h}=$}
\put(1.75,0){\textcolor{black}{$\da$}}
\put(1.75,-0.4){\textcolor{black}{$g$}}
\put(1.65,-1.05){\textcolor{black}{$\da$}}
\put(1.85,-1.05){\textcolor{black}{$\ua$}}
\put(1.75,-1.6){\textcolor{black}{$h$}}
\put(1.3,-1.6){$\ra$}
\put(2.15,-1.6){$\la$}
\end{picture}
}

\put(11,2.5){
\begin{picture}(0,3.25)(0,0)
\put(-0.75,-0.5){$\cont{D}{g}=$}
\put(1,0.5){\textcolor{black}{$f$}}
\put(1.2,0){\textcolor{black}{$\searrow$}}
\put(1.75,-0.4){\textcolor{black}{$x$}}
\put(1.75,-1.6){\textcolor{black}{$y$}}
\qbezier(1.05,0.3)(0.6,-0.5)(1.6,-1.4)
\put(1.55,-1.43){\vector(1,-2){0}}
\end{picture}
}

\put(14.5,2.5){
\begin{picture}(0,3.25)(0,0)
\put(-0.75,-0.5){$\cont{E}{h}=$}
\put(1,0.5){\textcolor{black}{$f$}}
\put(1.2,0){\textcolor{black}{$\searrow$}}
\put(1.75,-0.4){\textcolor{black}{$x'$}}
\put(1.75,-1.6){\textcolor{black}{$y'$}}
\put(2.2,-1.6){$\la\,g$}
\qbezier(1.05,0.3)(0.6,-0.5)(1.6,-1.4)
\put(1.55,-1.43){\vector(1,-2){0}}
\end{picture}
}
\end{picture}

\caption{Two drags with the same subdrag but different context drags.}\label{f:subd}
\end{figure}

Subdrags need no new sprouts since their vertices are closed under succession, but context drags do. On the other hand, a (nontrivial, non-void) subdrag always has  new roots. 
In particular, a nontrivial subdrag of a term at some position in the term has a root at its head, while the subterm has none. These new roots play an important r\^ole in the reconstruction of $D$ from $\rest{D}{W}$ and $\cont{D}{W}$.

\begin{lem}
The strict subdrag relation is a well-founded order.
\end{lem} 

\begin{proof}\label{l:subd}
Because strict subdrags have fewer vertices.
\end{proof}

\section{An Example}\label{s:example}

To motivate the development of drag rewriting, consider an example.
The goal is to take a ring of blue vertices (a ground drag) like this:
\setlength{\tabcolsep}{5mm}
\[
\psset{unit=5mm}
\begin{pspicture}(2,1.5)(8,5)
\cnodeput[fillstyle=solid,fillcolor=cyan]{0}(3,2.4){Aa}{$a$}
\cnodeput[fillstyle=solid,fillcolor=cyan]{0}(4.5,4.5){Ba}{$a$}
\cnodeput[fillstyle=solid,fillcolor=cyan]{0}(6,2.4){Ca}{$a$}
\ncline[linecolor=red,linestyle=solid,linewidth=1pt]{->}{Aa}{Ba}
\ncline[linecolor=red,linestyle=solid,linewidth=1pt]{->}{Ba}{Ca}
\ncline[linecolor=red,linestyle=solid,linewidth=1pt]{->}{Ca}{Aa}
\end{pspicture}
\]
and create instead a ring consisting of red vertices, in the same quantity as the blue but going in the opposite direction, like this:
\[
\psset{unit=2mm}
\begin{pspicture}(2,-3.2)(9,10)
\cnodeput[fillstyle=solid,fillcolor=red]{0}(0,0){Ac}{$c$}
\cnodeput[fillstyle=solid,fillcolor=red]{0}(4.5,8){Bc}{$c$}
\cnodeput[fillstyle=solid,fillcolor=red]{0}(9,0){Cc}{$c$}
\nccurve[linecolor=blue,linestyle=solid,linewidth=1pt,angleA=90,angleB=0]{->}{Cc}{Bc}
\nccurve[linecolor=blue,linestyle=solid,linewidth=1pt,angleA=180,angleB=90]{->}{Bc}{Ac}
\nccurve[linecolor=blue,linestyle=solid,linewidth=1pt,angleA=290,angleB=250]{->}{Ac}{Cc}
\end{pspicture}
\]
To make what's happening clearer, we will be using red for original edges, blue for new, and green for temporary ones.

We seek a rewriting algorithm that can apply at the same time---in parallel---to many vertices along the ring.


The left-hand and right-hand sides are drags; see~\cite{DBLP:journals/tcs/DershowitzJ19}.
Left roots map to right roots, and the two share variables.

There are two rules.  The first creates the red vertex and introduces a gray vertex to keep track of connections.
\[
\psset{unit=4mm}
\begin{tabular}{c@{\raisebox{15mm}{\Huge\qquad$\Longrightarrow$\qquad}}c}
\begin{pspicture}(2,-1)(5,6)
\cnodeput[fillstyle=solid,fillcolor=cyan]{0}(3,2.4){Aa}{$a$}
\cnodeput[]{0}(4.5,4.5){BFx}{$x$}
\pnode(6,2.4){CF_}{}
\ncline[linecolor=black,linestyle=solid]{|->}{CF_}{Aa}
\ncline[linecolor=red,linestyle=solid]{->}{Aa}{BFx}
\end{pspicture}
&
\begin{pspicture}(1,-1)(5,6)
\cnodeput[fillstyle=solid,fillcolor=red]{0}(0,0){Ac}{$c$}
\cnodeput[fillstyle=solid,fillcolor=lightgray]{0}(1.5,1.5){Ao}{$o$}
\cnodeput[]{0}(4.5,4.5){BFx}{$x$}
\pnode(6,2.){CF_}{}
\ncline[linecolor=black,linestyle=solid]{|->}{CF_}{Ao}
\nccurve[linecolor=cadmiumgreen,linestyle=solid,angleA=90,angleB=180]{->}{Ao}{BFx}
\ncline[linecolor=cadmiumgreen,linestyle=solid]{<->}{Ac}{Ao}
\end{pspicture}
\end{tabular}
\]
The two new vertices, red $c$ and gray $o$, are added with edges connecting them in both directions.
The edge from $a$ to whatever $x$ may be is replaced by an edge from $o$.
Instead of the root $a$ on the left, it is $o$ that becomes the new root.
Vertex $a$ becomes detached as its incident edges are removed by the rule.

Applying this rule thrice to
\[
\psset{unit=5mm}
\begin{pspicture}(2,1.5)(8,5,5)
\cnodeput[fillstyle=solid,fillcolor=cyan]{0}(3,2.4){Aa}{$a$}
\cnodeput[fillstyle=solid,fillcolor=cyan]{0}(4.5,4.5){Ba}{$a$}
\cnodeput[fillstyle=solid,fillcolor=cyan]{0}(6,2.4){Ca}{$a$}
\ncline[linecolor=red,linestyle=solid,linewidth=1pt]{->}{Aa}{Ba}
\ncline[linecolor=red,linestyle=solid,linewidth=1pt]{->}{Ba}{Ca}
\ncline[linecolor=red,linestyle=solid,linewidth=1pt]{->}{Ca}{Aa}
\end{pspicture}
\]
starting at the lower left and proceeding counterclockwise,
we get the following sequence of drags:
\[
\begin{tabular}{ccc}
\psset{unit=5mm}
\psset{unit=4mm}
\begin{pspicture}(0,-2)(9,5)
\cnodeput[fillstyle=solid,fillcolor=red,linestyle=dashed,dash=6pt 2pt]{0}(0,0){Ac}{$c$}
\cnodeput[fillstyle=solid,fillcolor=lightgray,linestyle=dashed,dash=6pt 2pt]{0}(1.5,1.5){Ao}{$o$}
\cnodeput[fillstyle=solid,fillcolor=cyan]{0}(4.5,4.5){Ba}{$a$}
\cnodeput[fillstyle=solid,fillcolor=cyan]{0}(6,2.4){Ca}{$a$}
\cnodeput[linecolor=black,hatchwidth=0.5pt,hatchsep=3pt,fillstyle=crosshatch*,fillcolor=white,hatchcolor=cyan]{0}(3,2.4){Aa}{$a$}
\ncline[linecolor=red,linestyle=solid]{->}{Ba}{Ca}
\nccurve[linecolor=red,linestyle=solid,angleA=220,angleB=0]{->}{Ca}{Ao}
\nccurve[linecolor=cadmiumgreen,linestyle=solid,angleA=90,angleB=180]{->}{Ao}{Ba}
\ncline[linecolor=cadmiumgreen,linestyle=solid]{<->}{Ac}{Ao}
\end{pspicture}
&
\psset{unit=4mm}
\begin{pspicture}(2,-2)(9,9)
\cnodeput[fillstyle=solid,fillcolor=red]{0}(0,0){Ac}{$c$}
\cnodeput[fillstyle=solid,fillcolor=red]{0}(9,0){Cc}{$c$}
\cnodeput[fillstyle=solid,fillcolor=lightgray]{0}(1.5,1.5){Ao}{$o$}
\cnodeput[fillstyle=solid,fillcolor=lightgray]{0}(7.5,1.5){Co}{$o$}
\cnodeput[fillstyle=solid,fillcolor=cyan]{0}(4.5,4.5){Ba}{$a$}
\cnodeput[linecolor=black,hatchwidth=0.5pt,hatchsep=3pt,fillstyle=crosshatch*,fillcolor=white,hatchcolor=cyan]{0}(3,2.4){Aa}{$a$}
\cnodeput[linecolor=black,hatchwidth=0.5pt,hatchsep=3pt,fillstyle=crosshatch*,fillcolor=white,hatchcolor=cyan]{0}(6,2.4){Ca}{$a$}
\nccurve[linecolor=red,linestyle=solid,angleA=0,angleB=110]{->}{Ba}{Co}
\nccurve[linecolor=cadmiumgreen,linestyle=solid,angleA=90,angleB=180]{->}{Ao}{Ba}
\nccurve[linecolor=cadmiumgreen,linestyle=solid,angleA=270,angleB=270]{->}{Co}{Ao}
\ncline[linecolor=cadmiumgreen,linestyle=solid]{<->}{Ac}{Ao}
\ncline[linecolor=cadmiumgreen,linestyle=solid]{<->}{Cc}{Co}
\end{pspicture}
&
\psset{unit=4mm}
\begin{pspicture}(0,-2)(9,9)
\cnodeput[fillstyle=solid,fillcolor=red]{0}(0,0){Ac}{$c$}
\cnodeput[fillstyle=solid,fillcolor=red]{0}(4.5,8){Bc}{$c$}
\cnodeput[fillstyle=solid,fillcolor=red]{0}(9,0){Cc}{$c$}
\cnodeput[fillstyle=solid,linestyle=dotted,fillcolor=lightgray]{0}(1.5,1.5){Ao}{$o$}
\cnodeput[fillstyle=solid,fillcolor=lightgray]{0}(4.5,5.9){Bo}{$o$}
\cnodeput[fillstyle=solid,fillcolor=lightgray]{0}(7.5,1.5){Co}{$o$}
\cnodeput[linecolor=black,hatchwidth=0.5pt,hatchsep=3pt,fillstyle=crosshatch*,fillcolor=white,hatchcolor=cyan]{0}(3,2.4){Aa}{$a$}
\cnodeput[linecolor=black,hatchwidth=0.5pt,hatchsep=3pt,fillstyle=crosshatch*,fillcolor=white,hatchcolor=cyan]{0}(4.5,4.3){Ba}{$a$}
\cnodeput[linecolor=black,hatchwidth=0.5pt,hatchsep=3pt,fillstyle=crosshatch*,fillcolor=white,hatchcolor=cyan]{0}(6,2.4){Ca}{$a$}
\nccurve[linecolor=cadmiumgreen,linestyle=solid,angleA=90,angleB=180]{->}{Ao}{Bo}
\nccurve[linecolor=cadmiumgreen,linestyle=solid,angleA=0,angleB=90]{->}{Bo}{Co}
\nccurve[linecolor=cadmiumgreen,linestyle=solid,angleA=270,angleB=270]{->}{Co}{Ao}
\ncline[linecolor=cadmiumgreen,linestyle=solid]{<->}{Ac}{Ao}
\ncline[linecolor=cadmiumgreen,linestyle=solid]{<->}{Bc}{Bo}
\ncline[linecolor=cadmiumgreen,linestyle=solid]{<->}{Cc}{Co}
\end{pspicture}
\end{tabular}
\]
The orphaned $a$ vertices are left shaded.

The next rule connects the added red vertices in the opposite direction:
\[
\psset{unit=4mm,linecolor=black}
\begin{tabular}{c@{\raisebox{20mm}{\Huge\qquad$\Longrightarrow$\qquad}}c}
\begin{pspicture}(-3,-2)(5,10)
\pnode(2.6,3.6){Ac}{}
\cnodeput{0}(0,0){Acy}{$y$}
\pnode(7,7.7){Bc_}{}
\cnodeput[linestyle=dotted,fillstyle=solid,fillcolor=red]{0}(4.5,8){Bc}{$c$}
\cnodeput[fillstyle=solid,fillcolor=lightgray,linestyle=dotted]{0}(1.5,1.5){Ao}{$o$}
\cnodeput{0}(4.5,6.0){Bo}{$x$}
\cnodeput*[]{0}(3.5,0.5){Co}{}
\pnode(6,2.4){Ca}{}
\nccurve[linecolor=black,linestyle=solid,angleA=220,angleB=270]{|->}{Co}{Ao}
\ncline[linecolor=black,linestyle=solid]{|->}{Ac}{Ao}
\ncline[linecolor=black,linestyle=solid]{|->}{Bc_}{Bc}
\ncline[linecolor=cadmiumgreen,linestyle=solid]{<-}{Acy}{Ao}
\nccurve[linecolor=cadmiumgreen,linestyle=solid,angleA=90,angleB=180]{->}{Ao}{Bo}
\ncline[linecolor=cadmiumgreen,linestyle=solid]{->}{Bc}{Bo}
\end{pspicture}
&
\begin{pspicture}(-0,-2)(9,10)
\pnode(2.6,3.6){Ac}{}
\psset{linecolor=black}
\cnodeput{0}(0,0){Acy}{$y$}
\cnodeput{0}(4.5,6.0){Bo}{$x$}
\pnode(3.5,5.0){B1}{}
\pnode(5.5,5.0){B2}{}
\pnode(7,7.7){Bc_}{}
\cnodeput[fillstyle=solid,fillcolor=red,linecolor=black,linestyle=dashed,dash=6pt 2pt]{0}(4.5,8){Bc}{$c$}
\cnodeput[fillstyle=solid,fillcolor=yellow,linestyle=dashed]{0}(1.5,1.5){Ao}{$e$}
\cnodeput*[]{0}(3.5,0.5){Co}{}
\pnode(6,2.4){Ca}{}
\nccurve[linecolor=black,linestyle=solid,angleA=220,angleB=270]{|->}{Co}{Ao}
\ncline[linecolor=black,linestyle=solid]{|->}{Bc_}{Bc}
\ncline[linecolor=black,linestyle=solid]{|->}{B1}{Bo}
\ncline[linecolor=black,linestyle=solid]{|->}{B2}{Bo}
\ncline[linecolor=black,linestyle=solid]{|->}{Ac}{Ao}
\nccurve[linecolor=blue,linestyle=solid,linewidth=1pt,angleA=180,angleB=90]{->}{Bc}{Acy}
\end{pspicture}
\end{tabular}
\]
The red $c$ vertex on the left is dotted and the one on the right is dashed  to indicate that they are actually distinct vertices.
The edge from it is redirected (in blue) to the other vertex pointed to by the gray $o$ that points to the original head of the edge from $c$.
A new, yellow double-rooted vertex $e$ is created to replace the deleted, double-rooted $o$, and serves to preserve incoming edges.
The two edges to the vertex signified by variable $x$ have been cut, so are now just roots.

Applying this rule to
\[
\psset{unit=4mm}
\begin{pspicture}(0,-1)(9,9)
\cnodeput[fillstyle=solid,fillcolor=red]{0}(0,0){Ac}{$c$}
\cnodeput[fillstyle=solid,fillcolor=red]{0}(4.5,8){Bc}{$c$}
\cnodeput[fillstyle=solid,fillcolor=red]{0}(9,0){Cc}{$c$}
\cnodeput[fillstyle=solid,fillcolor=lightgray]{0}(1.5,1.5){Ao}{$o$}
\cnodeput[fillstyle=solid,fillcolor=lightgray]{0}(4.5,5.9){Bo}{$o$}
\cnodeput[fillstyle=solid,fillcolor=lightgray]{0}(7.5,1.5){Co}{$o$}
\nccurve[linecolor=cadmiumgreen,linestyle=solid,angleA=90,angleB=180]{->}{Ao}{Bo}
\nccurve[linecolor=cadmiumgreen,linestyle=solid,angleA=0,angleB=90]{->}{Bo}{Co}
\nccurve[linecolor=cadmiumgreen,linestyle=solid,angleA=270,angleB=270]{->}{Co}{Ao}
\ncline[linecolor=cadmiumgreen,linestyle=solid]{<->}{Ac}{Ao}
\ncline[linecolor=cadmiumgreen,linestyle=solid]{<->}{Bc}{Bo}
\ncline[linecolor=cadmiumgreen,linestyle=solid]{<->}{Cc}{Co}
\end{pspicture}
\]
three times, we get
\[
\psset{unit=4mm}
\begin{tabular}{ccc}
\begin{pspicture}(0,-3)(9,9)
\cnodeput[fillstyle=solid,fillcolor=red]{0}(0,0){Ac}{$c$}
\cnodeput[fillstyle=solid,fillcolor=red]{0}(4.5,8){Bc}{$c$}
\cnodeput[fillstyle=solid,fillcolor=red]{0}(9,0){Cc}{$c$}
\pnode(3.5,5.0){B1}{}
\pnode(5.5,5.0){B2}{}
\cnodeput[fillstyle=solid,fillcolor=yellow]{0}(1.5,1.5){Ao}{$e$}
\cnodeput[fillstyle=solid,fillcolor=lightgray]{0}(4.5,6.0){Bo}{$o$}
\cnodeput[fillstyle=solid,fillcolor=lightgray]{0}(7.5,1.5){Co}{$o$}
\nccurve[linecolor=cadmiumgreen,linestyle=solid,angleA=0,angleB=90]{->}{Bo}{Co}
\nccurve[linecolor=cadmiumgreen,linestyle=solid,angleA=270,angleB=270]{->}{Co}{Ao}
\ncline[linecolor=cadmiumgreen,linestyle=solid]{<-}{Bc}{Bo}
\ncline[linecolor=cadmiumgreen,linestyle=solid]{<->}{Cc}{Co}
\ncline[linecolor=cadmiumgreen,linestyle=solid]{->}{Ac}{Ao}
\ncline[linecolor=black,linestyle=solid]{|->}{B1}{Bo}
\ncline[linecolor=black,linestyle=solid]{|->}{B2}{Bo}
\nccurve[linecolor=blue,linestyle=solid,linewidth=1pt,angleA=180,angleB=90]{->}{Bc}{Ac}
\end{pspicture}
&
\begin{pspicture}(0,-3)(9,9)
\cnodeput[fillstyle=solid,fillcolor=red]{0}(0,0){Ac}{$c$}
\cnodeput[fillstyle=solid,fillcolor=red]{0}(4.5,8){Bc}{$c$}
\cnodeput[fillstyle=solid,fillcolor=red]{0}(9,0){Cc}{$c$}
\pnode(2.5,0.5){A1}{}
\pnode(2.5,2.5){A2}{}
\pnode(3.5,5.0){B1}{}
\pnode(5.5,5.0){B2}{}
\cnodeput[fillstyle=solid,fillcolor=yellow]{0}(1.5,1.5){Ao}{$e$}
\cnodeput[fillstyle=solid,fillcolor=lightgray]{0}(4.5,6.0){Bo}{$o$}
\cnodeput[fillstyle=solid,fillcolor=yellow]{0}(7.5,1.5){Co}{$e$}
\nccurve[linecolor=cadmiumgreen,linestyle=solid,angleA=0,angleB=90]{->}{Bo}{Co}
\ncline[linecolor=cadmiumgreen,linestyle=solid]{<-}{Bc}{Bo}
\ncline[linecolor=cadmiumgreen,linestyle=solid]{->}{Cc}{Co}
\ncline[linecolor=black,linestyle=solid]{|->}{A1}{Ao}
\ncline[linecolor=black,linestyle=solid]{|->}{A2}{Ao}
\ncline[linecolor=black,linestyle=solid]{|->}{B1}{Bo}
\ncline[linecolor=black,linestyle=solid]{|->}{B2}{Bo}
\nccurve[linecolor=blue,linestyle=solid,linewidth=1pt,angleA=180,angleB=90]{->}{Bc}{Ac}
\nccurve[linecolor=blue,linestyle=solid,linewidth=1pt,angleA=290,angleB=250]{->}{Ac}{Cc}
\end{pspicture}
&
\begin{pspicture}(0,-3)(9,9)
\cnodeput[fillstyle=solid,fillcolor=red]{0}(0,0){Ac}{$c$}
\cnodeput[fillstyle=solid,fillcolor=red]{0}(4.5,8){Bc}{$c$}
\cnodeput[fillstyle=solid,fillcolor=red]{0}(9,0){Cc}{$c$}
\pnode(2.5,0.5){A1}{}
\pnode(2.5,2.5){A2}{}
\pnode(3.5,5.0){B1}{}
\pnode(5.5,5.0){B2}{}
\pnode(6.5,0.5){C1}{}
\pnode(6.5,2.5){C2}{}
\cnodeput[fillstyle=solid,fillcolor=yellow]{0}(1.5,1.5){Ao}{$e$}
\cnodeput[fillstyle=solid,fillcolor=yellow]{0}(4.5,6.0){Bo}{$e$}
\cnodeput[fillstyle=solid,fillcolor=yellow]{0}(7.5,1.5){Co}{$e$}
\ncline[linecolor=black,linestyle=solid]{|->}{A1}{Ao}
\ncline[linecolor=black,linestyle=solid]{|->}{A2}{Ao}
\ncline[linecolor=black,linestyle=solid]{|->}{B1}{Bo}
\ncline[linecolor=black,linestyle=solid]{|->}{B2}{Bo}
\ncline[linecolor=black,linestyle=solid]{|->}{C1}{Co}
\ncline[linecolor=black,linestyle=solid]{|->}{C2}{Co}
\nccurve[linecolor=blue,linestyle=solid,linewidth=1pt,angleA=90,angleB=0]{->}{Cc}{Bc}
\nccurve[linecolor=blue,linestyle=solid,linewidth=1pt,angleA=180,angleB=90]{->}{Bc}{Ac}
\nccurve[linecolor=blue,linestyle=solid,linewidth=1pt,angleA=290,angleB=250]{->}{Ac}{Cc}
\end{pspicture}
\end{tabular}
\]
This rule cannot be applied before a red vertex is created by the previous rule.
Later on, we will consider garbage collection for vertices like $e$ that have served out their purpose.
Actually, the old $c$ vertex also sticks around and can be recycled.

An alternative is to allow the left-hand and right-hand drags to share internal vertices---not just sprouts,
but to have separate sets of edges for the two sides.
The two rules are the same as before, except that now the red vertex $c$ in the second rule
is shared by both sides:
\[
\psset{unit=4mm}
\begin{tabular}{c@{\raisebox{20mm}{\Huge\qquad$\Longrightarrow$\qquad}}c}
\begin{pspicture}(-2,-2)(5,10)
\pnode(2.6,3.6){Ac}{}
\cnodeput{0}(0,0){Acy}{$y$}
\pnode(7,7.7){Bc_}{}
\cnodeput[fillstyle=solid,fillcolor=red]{0}(4.5,8){Bc}{$c$}
\cnodeput[fillstyle=solid,fillcolor=lightgray,linestyle=dotted]{0}(1.5,1.5){Ao}{$o$}
\cnodeput{0}(4.5,6.0){Bo}{$x$}
\cnodeput*[]{0}(3.5,0.5){Co}{}
\pnode(6,2.4){Ca}{}
\ncline[linecolor=black,linestyle=solid]{|->}{Bc_}{Bc}
\nccurve[linecolor=black,linestyle=solid,angleA=220,angleB=270]{|->}{Co}{Ao}
\ncline[linecolor=black,linestyle=solid]{|->}{Ac}{Ao}
\nccurve[linecolor=cadmiumgreen,linestyle=solid,angleA=90,angleB=180]{->}{Ao}{Bo}
\ncline[linecolor=cadmiumgreen,linestyle=solid]{<-}{Acy}{Ao}
\ncline[linecolor=cadmiumgreen,linestyle=solid]{->}{Bc}{Bo}
\end{pspicture}
&
\begin{pspicture}(-0,-2)(9,10)
\pnode(2.6,3.6){Ac}{}
\cnodeput{0}(0,0){Acy}{$y$}
\cnodeput{0}(4.5,6.0){Bo}{$x$}
\pnode(3.5,5.0){B1}{}
\pnode(5.5,5.0){B2}{}
\pnode(7,7.7){Bc_}{}
\cnodeput[fillstyle=solid,fillcolor=red]{0}(4.5,8){Bc}{$c$}
\cnodeput[fillstyle=solid,fillcolor=yellow,linestyle=dashed,dash=6pt 2pt]{0}(1.5,1.5){Ao}{$e$}
\cnodeput*[]{0}(3.5,0.5){Co}{}
\pnode(6,2.4){Ca}{}
\nccurve[linecolor=black,linestyle=solid,angleA=220,angleB=270]{|->}{Co}{Ao}
\ncline[linecolor=black,linestyle=solid]{|->}{Bc_}{Bc}
\ncline[linecolor=black,linestyle=solid]{|->}{B1}{Bo}
\ncline[linecolor=black,linestyle=solid]{|->}{B2}{Bo}
\ncline[linecolor=black,linestyle=solid]{|->}{Ac}{Ao}
\nccurve[linecolor=blue,linestyle=solid,linewidth=1pt,angleA=180,angleB=90]{->}{Bc}{Acy}
\end{pspicture}
\end{tabular}
\]
The red $c$ vertex on the two sides has a solid border to make it clear that it is the same vertex on the two sides.
The edge from $c$ to $x$ is redirected to $y$, as before.
As in the previous version, the gray $o$ vertex is only in the left drag, while the yellow $e$ is only in the right one.
The application of this rule would look the same as the above, except that the red vertices are preserved by the rule. 
Later, in Example~\ref{ex:share},
we will recast this example by using a more general rewrite rule format, with sharing of subdrags between left- and right-hand sides.

\section{Drag Morphisms}

A vertex of a drag has both a name (an element in $V$) and a label, taken from $\Sigma$ or $\Xi$. In particular, the sprouts of a drag are vertices labeled by variables from $\Xi$. The vertices of ordinary terms, on the other hand, are usually left nameless, with their positions (in a Dewey-decimal-like notation) standing in for names. 
In the term tradition, sprouts \emph{are} variables: the difference between a sprout and its label does not matter because the term framework does not distinguish between two terms that correspond to distinct isomorphic  graphs. 
Drags being graphs, two drags may be identical or they may be isomorphic as graphs. 
This distinction becomes crucial when it comes to sharing, which is why it is not relevant for terms, where there is no sharing.
Of course, terms can be seen as a particular kind of drag, but, as just stressed, it is important to understand that term equality for trees corresponds to isomorphism of drags, not identity.

To define precisely the kinds of equalities on drags in which we are interested, notions of drag morphism, drag monomorphism, and drag isomorphism are required. 
These must of course reduce to the corresponding notions on graphs for ground drags. 
The possibility that drags share vertices will require special care.
Another important matter is categoricity. 

\subsection{Morphisms, monomorphisms, isomorphisms and equimorphisms}
As usual, morphisms will be maps from the set of edges of the input drag to the set of edges of the target drag that are the identity on shared vertices. 
What differs from the usual graph morphisms is that we need to take care of variables, regarding which we encounter three problems. 
First, internal vertices need be mapped to internal vertices, as is already the case with terms.
Second, the drag may be nonlinear, two or more sprouts being labeled the same; such sprouts cannot be mapped independently of each other. 
Third, two vertices of the input drag may be mapped to the same vertex of the target drag in case the mapping is not injective. 
This is a standard situation for morphisms, but this has to happen for monomorphisms as well, if  only for two sprouts sharing the same label. 
But monomorphisms may also have to map a sprout and an internal vertex of the input drag to the same vertex of the target drag: 
Monomorphisms will be injective on internal vertices only. 
We address these questions in turn in the sequel.

\begin{defi}[Equimorphism]
\label{d:emo}
Given two drags $D=\langle V, R, L, X, S\rangle$ and $D'=\langle V', R', L', X', S'\rangle$, an \emph{equimorphism} is a bijective map $\omicron: V \ra V'$ that
preserves labels \green{and root quantity} at all vertices [$\forall u\in V \st L'(\omicron(s))=L(u)$ \green{and $R'(\omicron(u))=R(u)$}]
as well as the successor function [$\edgen(u,i,v)\in X \mbox{ iff } \edgen(\omicron(u),i,\omicron(v))\in X'$].
\end{defi}

Equimorphic drags are just copies of one another, possibly sharing some vertices, and therefore, the inverse of an equimorphism is itself an equimorphism.

\begin{defi}[Morphism]
\label{d:pmo}
Given two drags $D=\langle V, R, L, X, S\rangle$ and $D'=\langle V', R', L', X', S'\rangle$, $I$ and $I'$ their respective sets of internal vertices, a \emph{morphism} from $D$ to $D'$ is a map $\omicron: V\ra V'$, called a \emph{vertex map}, which 
\begin{enumerate}
\item
restricts to mappings from internal vertices $I$ to $I'$ and preserves their labels;
\item
restricts to mappings from isolated sprouts of $D$ to isolated sprouts of $D'$;
\item preserves indegrees at all vertices of $\omicron(I)$ [$\forall v \in \omicron(I) \st \Sigma_{u\in I\st \omicron(u)=v}\npred(u,D)=\npred(v, D')$];
\item maps edges of $D$ to corresponding edges of $D'$ 
[$\forall \edgen(u,i,v)\in X \st \edgen(\omicron(u),i,\omicron(v))\in X'$];
\item 
preserves equimorphic subdrags, that is,  maps equimorphic subdrags of the source drag to equimorphic subdrags of the target drag.
\end{enumerate}
We denote by $\omicron_I$ the restriction  of $\omicron$ to the internal vertices of $D$, and define the \emph{edge map} of the morphism $\omicron$ as $\omicron_X(\edgen(u,i,v))=\edgen(\omicron(u),i,\omicron(v))$.
\end{defi}

If $D$ and $D'$ are ground drags, morphisms are just ordinary graph morphisms. 

An important aspect of our definition is that morphisms may map a non-isolated sprout to any vertex of the same indegree, while isolated sprout must be mapped to an isolated sprout of the same indegree. In the literature of term graphs, isolated sprouts are rarely considered when defining morphisms. Likewise, condition (5) is either omitted in the absence of non-linear sprouts, or strengthened by mapping non-linear sprouts to the same vertex.

\begin{defi}[Monomorphism, isomorphism]
\label{d:mmo}
A \emph{monomorphism} is a morphism whose vertex map restricts injectively to internal vertices and isolated sprouts. An \emph{injection} is a monomorphism whose vertex map is the identity on internal vertices and isolated sprouts.
A \emph{isomorphism} is a morphism whose vertex map is bijective.
\end{defi}

We will see later that finding a monomorphism mapping a drag $D$ to a drag $D'$ is nothing but matching $D$ against $D'$, an operation that will rely on indegree preservation. 
Note that indegree preservation becomes indeed a true preservation property for monomorphisms.

Since an edge \mbox{$\edgen(\omicron(u),i,\omicron(v))$} is characterized by the pair $(\omicron(u),i)$, which is unique when $\omicron$ is injective, the edge map $\omicron_X$ of a monomorphism $\omicron$ is injective on all edges. That's why we do not need to state it, even though $\omicron$ is not injective on all vertices. This would not be the case were an edge simply a pair $\edge(u,v)$ instead of a triple $\edgen(u,i,v)$.

Because isomorphisms are bijective and preserve equimorphic subdrags, they must map two sprouts of the same label $x$ to two sprouts of the same label $y$ possibly different from $x$, implying that the inverse of an isomorphism is an isomorphism. It therefore appears that equimorphisms are those isomorphisms that preserve labels at all sprouts.

Given a morphism $\omicron: D\ra D'$, edges of $D'$ may be mapped from edges between internal vertices of $D$, they will appear in black on the figures. Some others may be mapped for an edge of $D'$ whose target is a sprout, they will appear in red. Other edges of $D'$ are \emph{context edges}. They appear in green if their target is not a vertex in the image of $\omicron$, otherwise in blue. This categorization of edges will be made formal later and play an important r\^ole. For all examples, checking the requirements for being morphisms is left to the reader.

\begin{exa}[Morphism]
\label{ex:mo}
Consider drags $D=f^{[1]}(f^{[2]}(x,x),x)$, with vertices $f_1$ (above) and $f_2$ (below) labeled by the binary symbol $f$ and three other vertices $x_1$, $x_2$,  and $x_3$ (in depth-first order), and $D'= f(\textsc{self}, \textsc{self})$, a drag with a single vertex $f_3$ labeled $f$, two edges $\edgen(f_3,1,f_3)$ and $\edgen(f_3,2,f_3)$, and no root.
Observe that the mapping $\omicron:D\to D'$ such that 
$\omicron(f_1)=\omicron(f_2)=\omicron(x_1)=\omicron(x_2)=\omicron(x_3)=f_3$ is a morphism. In particular, the indegree of $f_3$ ($=4$) is the sum of the indegrees of $f_1$ ($=3$) and $f_2$ ($=1$).
This example is represented on the right of Figure~\ref{f:monos}.
\qed\end{exa}

\begin{exa}[Monomorphism]
\label{ex:mono}
The same Figure~\ref{f:monos} shows on its left a monomorphism that maps the internal vertex $f_1$ to vertex $f_2$ of the same indegree 3, and both sprouts $x_1$ and $x_2$ to the internal vertex $f_2$. See how the two (red) edges of $D$ become edges of $D'$ between internal vertices.
\end{exa}

\begin{exa}
\label{e:monos}
Figure~\ref{f:monos1} (left) is an example of a monomorphism [$\omicron(f_1)=f_2$, $\omicron(a_1)=a_2$, $\omicron(x)=a_2$] with various kinds of edges. Note that the left vertex labeled $a$ has lost its roots. One has actually been "used" by the edges $\edgen(f_1,1,x)$ to create the edge $\edgen(f_2,1,a_2)$; The other has been "absorbed" by the edge $\edgen(h,1,a_2)$ of the target drag.
\end{exa}

\begin{figure}[t]
\resizebox*{1.\linewidth}{!}{
\begin{tikzpicture}[
Droundnode/.style={circle, draw=blue!50, fill=blue!5, very thick, minimum size=4mm},
namenode/.style={},
Dproundnode/.style={circle, fill=red!5, very thick, minimum size=4mm},
namenode/.style={}
]
%
\node[Dproundnode]    (f1)    []  {$\,f_1^{[3]}$};
\node (flf) [left=3mm of f1] {};
\node (frf) [right=3mm of f1] {};
\node (x1) [below=19mm of flf, very thick] {$x_1$};
\node (x2) [below=19mm of frf, very thick] {$x_2$};
\draw[->, very thick, red] (f1) to node [above] {1} (x1);
\draw[->, very thick, red] (f1) to node [above] {2} (x2);

\node (femb) [right=15mm of f1] {};
\node (emb) [below=8mm of femb] {${\huge \hookrightarrow}$};

\node[Dproundnode]    (g)    [right=30mm of f1, very thick]    {$g$}; 
\node[Dproundnode]    (h)    [right=10mm of g, very thick]    {$h$}; 

\node[Dproundnode]    (f2)    [below=15mm of g, very thick] {$f_2$};
\draw[->, very thick, red] (f2) .. controls + (-1,-1) and + (-0.5,1) .. (f2);
\draw[->, very thick, red] (f2) .. controls + (1,-1) and + (0.5,1) .. (f2);

\draw[->, very thick, blue] (h) to node [above] {1} (g);
\draw[->, very thick, blue] (g) to node [left] {1} (f2);
\node    (n1)    [left=0mm of f2, thick, red] {$1\;$};
\node    (n2)    [right=0mm of f2, thick, red] {$\;2$};
\node   (nn2)   [below=5mm of x1] {};
\node (n12) [right=-5mm of nn2, thick] {$\omicron(x_1)=\omicron(x_2)=\omicron(f_1)= f_2$};

\node[Dproundnode]    (nf1)    [right=7cm of f1]  {$\,f_1^{[1]}$};
\node (rnf1) [left=0mm of nf1, very thick] {$\ra$};
\node (nflf) [left=3mm of nf1] {};
\node (nfrf) [right=3mm of nf1] {};
\node (nf2) [below=15mm of nflf, very thick] {$\;f_2^{[2]}$};
\node (nflf2) [left=3mm of nf2] {};
\node (nfrf2) [right=3mm of nf2] {};
\node (nx1) [below=15mm of nflf2, very thick] {$x_1$};
\node (nx2) [below=15mm of nfrf2, very thick] {$x_2$};
\node (nx3) [below=15mm of nfrf, very thick] {$x_3$};
\draw[->, very thick, black] (nf1) to node [above] {1} (nf2);
\draw[->, very thick, black] (nf1) to node [above] {2} (nx3);
\draw[->, very thick, red] (nf2) to node [above] {1} (nx1);
\draw[->, very thick, red] (nf2) to node [above] {2} (nx2);

\node (femb2) [right=15mm of nf1] {};
\node (nemb2) [below=8mm of femb2] {${\huge \ra}$};

\node[Dproundnode]    (nf3)    [right=10mm of nemb2, very thick] {$f_3$};
\draw[->, very thick, red] (nf3) .. controls + (-1,-1) and + (-1,1) .. (nf3);
\node    (faf3)    [above=1mm of nf3, thick] {};
\node    (fbf3)    [below=1mm of nf3, thick, red] {};
\node    (n3)    [left=0mm of faf3, thick] {$1\quad$};
\node    (n4)    [left=0mm of fbf3, thick, red] {$1\quad$};
\draw[->, very thick] (nf3) .. controls + (-0.9,-0.85) and + (-0.85,0.9) .. (nf3);
\node    (n5)    [right=0mm of faf3, thick] {$\quad2$};
\node    (n6)    [right=0mm of fbf3, thick, red] {$\quad2$};
\draw[->, very thick, red] (nf3) .. controls + (1,-1) and + (1,1) .. (nf3);
\draw[->, very thick] (nf3) .. controls + (0.9,-0.9) and + (0.9,0.9) .. (nf3);

\node    (n10)    [below=.5cm of nf3, thick] {$\omicron(f_1)=\omicron(f_2)=f_3$};
\node (nn11) [below=4mm of n10, thick] {} ;
\node (n11) [right=-3.1cm of nn11, thick] {$\omicron(x_1)=\omicron(x_2)=\omicron(x_3)= f_3$};

\end{tikzpicture}
}
\caption{A monomorphism ($\hookrightarrow$) on the left and a morphism ($\ra$) on the right. Incoming arrows at some vertex annotated with numbers stand for multiple roots.}\label{f:monos}
\end{figure}
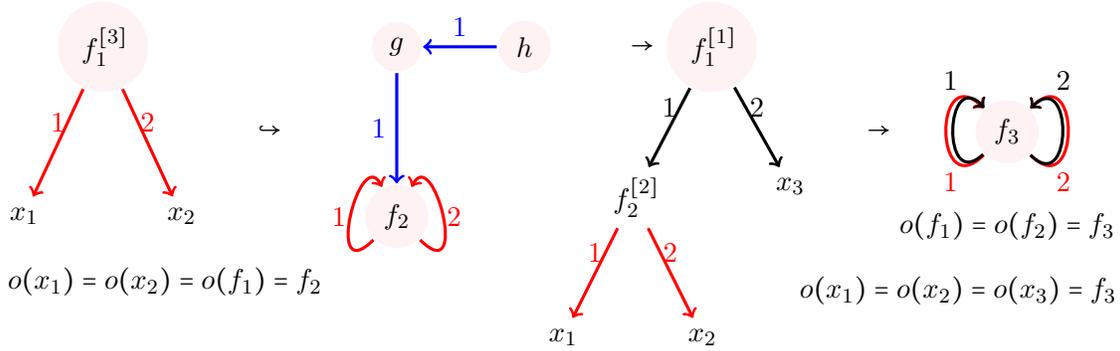

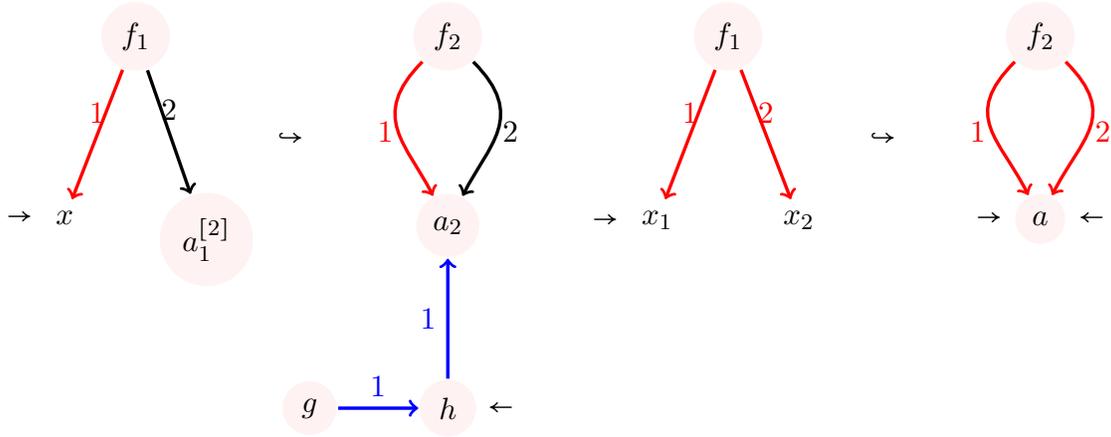
\begin{figure}[t]
\resizebox*{1.\linewidth}{!}{
\begin{tikzpicture}[
Droundnode/.style={circle, draw=blue!50, fill=blue!5, very thick, minimum size=4mm},
namenode/.style={},
Dproundnode/.style={circle, fill=red!5, very thick, minimum size=4mm},
namenode/.style={}
]
%
\node[Dproundnode]    (f1)    []  {$f_1$};
\node (flf) [left=3mm of f1] {};
\node (frf) [right=3mm of f1] {};
\node (fbf) [below=7mm of f1] {};
\node (x) [below=19mm of flf, very thick] {$x$};
\node (rootx) [left=0mm of x, very thick] {$\ra$};
\node[Dproundnode]    (a1)    [below=18mm of frf, very thick] {$a_1^{[2]}$};
\draw[->, very thick, red] (f1) to node [above] {1} (x);
\draw[->, very thick] (f1) to node [above] {2} (a1);

\node (emb) [right=15mm of fbf] {${\huge \hookrightarrow}$};

\node[Dproundnode]    (f2)    [right=30mm of f1, very thick]    {$f_2$};       
\node[Dproundnode]    (a2)    [below=15mm of f2, very thick] {$a_2$};
\draw[->, very thick, red] (f2) .. controls + (-1,-1) and + (-0.5,1) .. (a2);
\draw[->, very thick] (f2) .. controls + (1,-1) and + (0.5,1) .. (a2);
\node    (fn7)    [below=6mm of f2, thick, red] {};
\node    (n7)    [left=4mm of fn7, thick, red] {$1$};
\node    (n8)    [right=4mm of fn7, thick] {$2$};

\node[Dproundnode]    (h)    [below=15mm of a2, thick] {$h$};
\node (rooth) [right=0mm of h] {$\la$};
\draw[->, very thick, blue] (h) to node [left] {1} (a2);
\node[Dproundnode]   (g)   [left=10mm of h] {$g$};
\draw[->, very thick, blue] (g) to node [above] {1} (h);

\node[Dproundnode]    (nf1)    [right=6.5cm of f1]  {$f_1$};
\node (nflf) [left=3mm of nf1] {};
\node (nfrf) [right=3mm of nf1] {};
\node (nfbf) [below=7mm of nf1] {};
\node (nx1) [below=19mm of nflf, very thick] {$x_1$};
\node (rootnx) [left=0mm of nx1, very thick] {$\ra$};
\node (nx2)    [below=19mm of nfrf, very thick] {$x_2$};
\draw[->, very thick, red] (nf1) to node [above] {1} (nx1);
\draw[->, very thick, red] (nf1) to node [above] {2} (nx2);

\node (nemb) [right=15mm of nfbf] {${\huge \hookrightarrow}$};

\node[Dproundnode]    (nf2)    [right=30mm of nf1, very thick]    {$f_2$};       
\node[Dproundnode]    (na2)    [below=15mm of nf2, very thick] {$a$};
\node (root1na2) [left=0mm of na2, very thick] {$\ra$};
\node (root2na2) [right=0mm of na2, very thick] {$\la$};
\draw[->, very thick, red] (nf2) .. controls + (-1,-1) and + (-0.5,1) .. (na2);
\draw[->, very thick, red] (nf2) .. controls + (1,-1) and + (0.5,1) .. (na2);
\node    (fn9)    [below=6mm of nf2, thick, red] {};
\node    (n9)    [left=4mm of fn9, thick, red] {$1$};
\node    (n10)    [right=4mm of fn9, thick, red] {$2$};

\end{tikzpicture}
}
\caption{Additional examples of monomorphisms.\label{f:monos1}}
\end{figure}

\begin{exa}[Isomorphism]
\label{ex:iso}
We now show that the two drags $D=f(x,y^{[1]})$ and $D'=f(x,z^{[1]})$, sharing the sprout labeled $x$, where $y$ and $z$ are different, are isomorphic. Using the map $\omicron:D\to D'$ such that $\omicron(f_1)=f_2$, $\omicron(x)=x$ and $\omicron(y)=z$, $\omicron$ being the identity for the shared vertex $x$, implying that
$\omicron_X(\edgen(f_1,1,x))=\edgen(f_2,1,x)$ and $\omicron_X(\edgen(f_1,2,x))=\edgen(f_2,2,z)$, we get $\omicron(D)=D'$. The map $\omicron$ is a monomorphism that is bijective, hence is an isomorphism. Were the vertex $z$ in $v$ replaced by a vertex $w$ labeled $y$, both drags would be equimorphic, regardless of whether $w$ is shared with the vertex $y$ in $u$.
\qed\end{exa}

\subsection{The Category of Drags}
We have already pointed out the strong preservation condition of equimorphic drags by morphisms: 
The weaker condition that identically labeled sprouts be related by isomorphism, as suggested in~\cite{DBLP:journals/tcs/DershowitzJ19}, would not ensure the following key properties:

\begin{lem}
\label{l:closureprop}
Morphisms, monomorphisms, injections, isomorphisms, and equimorphisms are closed under composition.
\end{lem}

\begin{proof}
Consider $\omicron : D\to D'$ and $\omicron' : D'\to D"$, and let 
$\omicron"=\omicron'\circ\omicron : D\to D"$ be their composition.

Equimorphisms compose since bijections do, and preservation of variable labels and edges is transitive.

Morphisms compose because restrictions to internal vertices and isolated sprouts compose, preservation of indegrees does because addition is associative, 
edge maps compose because vertex maps do, and preservation of equimorphic subdrags is satisfied because equimorphisms compose.

Monomorphisms compose because injectivity on internal vertices and isolated sprouts is preserved by composition.

Being monomorphisms whose vertex map is the identity, injections compose.

Being monomorphisms whose vertex map is bijective, isomorphisms compose.
\end{proof}

Classically, graphs and their morphisms form a category but have no roots at their internal vertices nor variables at their leaves that can be redirected to the roots. 
Their presence in drags and the associated preservation properties that morphisms must satisfy make nontrivial two properties that are usually obtained naturally.

Drag morphisms are indeed defined on equivalence classes of drags modulo isomorphisms because, given any input drag, an isomorphism preserves its roots at all vertices and maps its equimorphic subdrags to equimorphic subdrags of the output drag. 

Monomorphisms on graphs are just injective morphisms between their sets of vertices and between their sets of edges. Here, monomorphisms are more complex maps, being injective on internal vertices only, as are substitutions on terms, which actually implies that $\omicron_X$ is injective on all edges thanks to the natural number component of an edge. These properties imply that our monomorphisms are indeed monomorphisms in the categorical sense, that is, are the left-cancellative morphisms:

\begin{lem}[Categoricity]
Let $D',D"$ be drags. A morphism $\kappa$ from $D'$ to $D"$ is a monomorphism iff for it is left cancellative, that is,
for all drags $D$ and morphisms $\omicron, \omicron':D\ra D'$ such that 
\renewcommand{\theequation}{*}
\begin{equation}
\kappa\circ \omicron=\kappa\circ\omicron'
\end{equation}
it is the case that $\omicron=\omicron'$.

\end{lem}

\begin{proof}
Claim: Let $I$ and $I'$ be the sets of internal vertices and isolated sprouts of $D$ and $D'$, respectively. Property (*) implies $\kappa_{I'}\circ \omicron_I=\kappa_{I'}\circ\omicron'_I$. Since categoricity reduces to injectivity for functions on sets, it follows that $\omicron_I=\omicron'_I$ iff $\kappa_{I'}$ is injective. 

Assume now that $\kappa$ is a monomorphism, hence that $\kappa_{I'}$ is injective and $\omicron_I=\omicron'_I$ by the above Claim. Therefore, $\omicron_X=\omicron'_X$ by definition of an edge map, which implies in turn  that $\omicron_V=\omicron'_V$ for non-isolated sprouts. Altogether, we get $\omicron=\omicron'$.

Conversely, assume that $\omicron=\omicron'$ for all morphisms $\kappa$ satisfying (*). By the same token as previously, the restriction $\kappa_{I'}$ of the morphism $\kappa$ is therefore injective, hence $\kappa$ is a monomorphism.
\end{proof}

The proof shows that injectivity of monomorphisms for internal vertices and isolated sprouts ensures categoricity. The definition of morphisms for sprouts is therefore unimportant, provided transitivity is satisfied. The precise definition we gave has therefore a single objective: to ensure that matching behaves as desired. We will show in Section \ref{s:matching} that this is indeed the case.

We can now conclude:

\begin{thm}
\label{t:cat}
Drags equipped with their morphisms, monomorphisms, and isomorphisms form a category, of which the category of terms is a particular case.
\end{thm}

\subsection{Notations}  We end this section by introducing notations for the various kinds of morphisms and related equalities that are relevant for drags. 

We write $\omicron : D\to D'$ for a morphism $\omicron$; $\omicron: D\hookrightarrow D'$ if $\omicron$ is a monomorphism;
$D\simeq_{\omicron} D'$ when $\omicron$ is an isomorphism; $D\equiv_\omicron D'$ if $\omicron$ is an equimorphism; and $D=D'$ when $\omicron$ is identity. 
The subscript $\omicron$ will often be omitted. 
Monomorphisms may be abbreviated "monos". 

We may need more precise notations $D \simeq_\omicron^\sigma D'$ in case of an isomorphism $\omicron$, where $\sigma$ is a bijection between the variables of corresponding sprouts. $D'$ is sometimes called a \emph{renaming} of $D$. The renaming \emph{preserves variables} if $\sigma$ is the identity, in which case $D$ and $D'$ are equimorphic, implying that $D\simeq_\omicron^{id} D'$ and $D\equiv_\omicron D'$ are indeed the same. The notation $D=D'$ (not to be confused with definitional equality) is reserved for the case where $D$ and $D'$ are the very same drag. 

Two drags $D$ and $D'$ are \emph{disjoint} if they share no vertices nor variables. \emph{Renaming apart} two drags $D$ and $D'$ amounts to renaming bijectively the shared vertices and variables of $D$ so that the result $D"$ is isomorphic to $D$ while $D'$ and $D"$ are disjoint. 


\section{Drag Operations}
\label{ss:ops}
There are two main operations on drags, \emph{sum} and \emph{product}, that are reminiscent of similar operations used in the literature, although in restricted settings compared to here. 
The third operation, \emph{wiring}, is in some sense more fundamental than product, which amounts to a particular case of wiring a sum. Sum and product operate on \emph{compatible} drags.

\subsection{Compatibility}
\label{ss:cd}

Operations on drags that are not disjoint requires an assumption:

\begin{defi}[Compatible drags]
\label{d:compatible}
Two drags $D=\langle V, R, L, X, S\rangle$ and $D'=\langle V', R', L', X', S'\rangle$ are \emph{compatible} if 
\begin{itemize}
\item $V\cap V'$ is closed under $X$ and $X'$;
\item $L$ and $L'$ coincide on $V\cap V'$;
\item $D$ and $D'$ have the same indegree at each shared vertex $v$, at least equal to the total number of shared and non-shared edges heading at $v$.
\end{itemize}
\end{defi}

\noindent
In words, compatible drags that share a vertex must also share the whole subdrag generated by that vertex. 
Also, they must have enough roots at their shared vertices so as to have the same indegree after replacing the roots of each graph at a shared vertex by the incoming non-shared edges of the other graph at that vertex.

\begin{exa}
\label{ex:comp}
The following  two graphs $D,D'$ are compatible:
\begin{align*}
D=&(\{v_1, v_2, v_3\}, (R(v_1)=R(v_3)=0, R(v_2)=1),(L(v_1)=f, L(v_2)=g, L(v_3)=a),\\
&(X(v_1)=v_2, X(v_2)=v_3, X(v_3)=v_2), \varnothing),\\
D'=&(\{v'_1,v_2,v_3\}, (R(v'_1)=R(v_3)=0, R(v'_2)=1),(L(v'_1)=h, L(v_2)=g, L(v_3)=a),\\
&(X(v'_1)=v_2, X(v_2)=v_3, X(v_3)=v_2), \varnothing).
\end{align*}
On the other hand, the  following graphs are not compatible (closure condition violated):
\begin{align*}
D=&(\{v_1, v_2, v_3\}, (R(v_1)=R(v_3)=0, R(v_2)=1),(L(v_1)=f, L(v_2)=g, L(v_3)=a), \\
& (X(v_1)=v_2, X(v_2)=v_3, X(v_3)=v_1), \varnothing),\\
D'=&(\{v'_1,v_2,v_3\}, (R(v'_1)=R(v_3)=0, R(v_2)=1),(L(v'_1)=h, L(v_2)=g, L(v_3)=a),\\
&(X(v'_1)=v_2, X(v_2)=v_3, X(v_3)=v_2), \varnothing).
\end{align*}
The  following two graphs are also incompatible (lack of roots at $v_2$):
\begin{align*}
D=&(\{v_1, v_2, v_3\}, (R(v_1)=R(v_2)=R(v_3)=0),(L(v_1)=f, L(v_2)=g, L(v_3)=a), \\
&(X(v_1)=v_2, X(v_2)=v_3, X(v_3)=v_1), \varnothing),\\
D'=&(\{v'_1,v_2,v_3\}, (R(v'_1)=R(v_2)=R(v_3)=0),(L(v'_1)=h, L(v_2)=g, L(v_3)=a),\\
&(X(v'_1)=v_2, X(v_2)=v_3, X(v_3)=v_2), \varnothing).
\end{align*}
\end{exa}

One difficulty is that compatibility is not preserved by equimorphisms: $D=f(a_1)$ is compatible with itself, and $D$ is equimorphic with $D'=f(a_2)$, but $D$ and $D'$ are incompatible since $f$ is shared in $D$ and $D'$ but its successors $a_1$ in $D$ and $a_2$ in $D'$ are not.

\subsection{Sum}
\label{ss:sum}

Our first operation on compatible drags is both very simple and familiar when both drags are disjoint: 
it consists of placing two drags side by side to form a new drag. 
\begin{defi}[Sum]
\label{d:union}
The \emph{parallel composition} or \emph{sum} $D\oplus D'$ of two
compatible drags $D=\langle V, R, L, X, S\rangle$ and $D'=\langle V', R', L', X', S'\rangle$ is the drag $\langle V\cup V', R", L\cup L', X\cup X', S\cup S' \rangle$, where, 
$\forall x\in V\cup V'\st R"(v)=
 R(v)- |\{\edgen(u,i,v)\in X'\setminus X\}|= R'(v)- |\{\edgen(u,i,v)\in X\setminus X'\}|$.
\end{defi}

Note that the equality statement when defining the number of roots at all vertices in the union of two drags follows from the definition of compatibility. 
Note also that $R"(v)=R(v)$ if $v\in V\setminus V'$ and $R"(v)=R'(v)$ if $v\in V'\setminus V$, implying that the union of disjoint drags is their juxtaposition.

The following straightforward properties of parallel composition are important:

\begin{lem}
\label{l:ism}
Given two compatible drags $D=\langle V, R, L, X, S\rangle$ and $ D'=\langle V', R', L', X', S'\rangle$,
\begin{enumerate}
    \item $D\oplus D'$ preserves indegrees of $D$ and $D'$ at all vertices;
    \item the injections from $V$ and $V'$ to $V\cup V'$ are monomorphisms from $D$ and $D'$ to $D\oplus D'$.
\end{enumerate}
\end{lem}

\begin{proof}
Preservation of indegrees at all vertices follows from the definitions of compatibility and sum. 
Being identities, these injections are  injective morphisms that protect equimorphic subdrags and indegrees, hence are monomorphisms.
\end{proof}


Parallel composition allows to define the intersection of two compatible drags:

\begin{defi}[Intersection]
\label{d:intersection}
The \emph{intersection} of two compatible drags $D$ and $D'$ with vertices $V$ and $V'$ respectively is the subdrag of $D\oplus D'$, denoted $D\cap D'$, generated by $V\cap V'$.
\end{defi}

Once more, we remark that indegrees are preserved at all vertices of the intersection.

\subsection{Wiring}\label{sss:w}

The purpose of wiring a drag $D$ is to add new edges to a drag by \emph{connecting} sprouts to roots. Informally, a set of wires will be a set of pairs made out of a sprout and a root, written as  $\wir{s}{r}$. 
Wiring $D$ will be the action of redirecting all edges $\edge(u,s)$ in $D$, including the roots of $s$, so that they become edges $\edge(u,r)$ or roots of $r$ in a new drag $D'$. Wiring may use a succession of wires like $\wir{s}{t}$ and $\wir{t}{r}$ that generate chains of wirings.

\begin{defi}[Wire, origin, target]
\label{d:wire}
Given a drag $D=\langle V, R, L, X, S\rangle$, a \emph{wire} is a pair $\wir{s}{r}$ of vertices of $D$, whose \emph{origin} $s$ is a sprout and \emph{target} $r$ a vertex different from $s$.
\end{defi}

\begin{defi}[Wiring chain]
Given a set of wires $W$ of a drag $D=\langle V, R, L, X, S\rangle$, we define $s >_W r$ for $s\in S$ and $r\in R$,
 if $\wir{s}{r}\in W$ or if
there is a vertex $t$ such that $\wir{s}{t}\in W$ and $t>_W r$.
\end{defi}
\noindent
In a wiring chain $s_0\leadsto s_1\leadsto\cdots\leadsto s_n\leadsto r$, all but possibly the final element $r$ are sprouts, and the final element $r$ is a root---and possibly also a sprout.
 
\begin{defi}[Well-behaved set of wires]
A finite set $W$ of wires of $D$ is \emph{well-behaved} if:
\begin{enumerate}
    \item functionality:     $\forall \wir{s}{r},\wir{s}{r'}\in W \st r=r'$; 
    \item injectivity:     $\forall r\in R\st\Sigma_{s >_W r}\, \pred(s)\leq R(r)$;
    \item well-foundedness: $W$ does not induce a cycle among the sprouts of $D$, that is, the restriction of $>_W$ to $S\times S$ is acyclic.
\end{enumerate}
The domain $\Dom{W}$ of $W$ is the set of sprouts that are origins of a wire.
\end{defi}

Condition (1) implies that $W$ is a partial function from $S$ to $V$. 
We will therefore be able to consider the \emph{restriction} of that function to a subset of its domain. 
If $V'\subseteq V$, we will say that $W$ \emph{restricts} to $V'$ if  $\forall\wir{s}{r}\in W \st  r\in V' $ if $ s\in V'$.

Condition (2) means that vertex $r$ is a root with a multiplicity  large enough so that rewiring edges from $s_1,\ldots s_n$ to $r$ does not require more roots of $r$ than are available, which is yet another manifestation of multi-injectivity. 

Condition (3) allows us to compare sprouts. 
Well-foundedness of this order aims at defining wiring by induction.

It also implies the following:
\begin{prop}
The relation $\geq_W$ is a partial order  for any well-behaved set of wires $W$.
\end{prop}

An empty set of wires is trivially well-behaved.  
A wire $\wir{s}{t}$, considered as a singleton set, is well-behaved iff it satisfies injectivity, that is, if the indegree of $s$ is no larger than the number of roots of $r$.

We now define the drag obtained by adding wires to an existing drag, starting with the case of a single wire $\wir{s}{r}$. The idea is that  sprout $s$ is removed, all edges ending up in $s$ (but not its roots) are moved to $r$ using its roots in the same number:

\begin{defi}[Elementary wiring]
\label{d:gwiring}
Given a drag $D=\langle V, R, L, X, S\rangle$ and a wire $\wir{s}{r}$, we define the drag $\gwiring{D}{\wir{s}{r}}$, after \emph{wiring}, as the drag $\langle V', R', L', X', S'\rangle$ such that:
\begin{enumerate}
\item
$V'=V\setminus s$;
\item
$R'= R \setminus (R(s) \cup R(r)) \cup r^{R(r)-\pred(s)}$;
\item $L'=L\upharpoonright V'$, the restriction of labels $L$ to  vertices in $V'$;
\item
$X'= X \setminus \{\edge(v, s)\st vXs\} \cup \{\edge(v,r): vXs\}$;
\item
$S'= S\setminus s$.
\end{enumerate}
\end{defi}

The condition that the origin of a wire is distinct from its target ensures that the sprout origin of the wire has disappeared from the resulting drag. 
The calculation of the new multiset of roots in item (2) expresses the fact that each edge redirected from the origin to the target of the wire consumes a root of the target.
In contrast with \cite{DBLP:journals/tcs/DershowitzJ19}, sprouts at the origin of a wire disappear with their roots, which are therefore lost if there were any.
In the particular case where $s:x$ and $r:y$ are both rootless isolated sprouts, wiring allows one to rename the sprout labeled $x$ by the sprout labeled $y$.

\begin{lem}
\label{l:pres}
Wiring preserves indegrees at all remaining vertices.
\end{lem}

To define wiring for an arbitrary set of well-behaved wires, we write a non-empty well-behaved set of wires $W$ as the union $\wir{s}{r}\cup W'$, where sprout $s$ is maximal in $>_W$. 
Well-foundedness of $>_W$ allows us to wire $\wir{s}{r}$ first, and then recur on $W'$, which can be easily shown well-behaved:

\begin{lem}
\label{l:wrec}
Let $W=\wir{s}{r}\cup W'$ be a well-behaved set of wires of $D$ such that $s$ is maximal in $>_W$. 
Then, $W'$ is a well-behaved set of wires of $D'=\gwiring{D}{\wir{s}{r}}$.
\end{lem}

\begin{proof}
By maximality assumption, $s$ does not occur in $W'$ which is therefore a set of wires of $D'$. 
Functionality and membership follow straightforwardly from well-behavedness of $W$, as well as well-foundedness, since it restricts to subsets. 
Injectivity holds at all vertices since $\Sigma_{t>_W r} \pred(t)=\Sigma_{t>_{W'} s} \pred(t) + \pred(s) $.
\end{proof}

It follows that all subsets of a well-behaved set of wires are well-behaved.

We can now define recursively wiring with an arbitrary well-behaved set $W$ of wires. Not only do we define the new drag $D_W$, but also trace the vertices of the original drag along the recursive computation, with 
or without 
their root multiplicity.

\begin{defi}[Wiring, resolution, natural injection]\label{d:wiring}
Given a well-behaved set of wires $W$ of a drag $D$, we define the drag $\wiring{D}{W}$, after wiring, by induction on the size of $W$, where for a non-empty set of wires $W$, we write $\wir{s}{r}\cup W'$ for $W=\{\wir{s}{r}\}\cup W'$,
\begin{align*}
\wiring{D}{\varnothing} &=D \\
\wiring{D}{\wir{s}{r}\cup W'} &=
\wiring{(\gwiring{D}{\wir{s}{r}})}{W'}.
\end{align*}
The \emph{resolution} 
$\targetw{t}{D}{W}$ of sprout $t$ of $D$ in the domain of $W$ is the (unique) root $r$ in $D$ that is the minimal one such that $t\geq_W r$, together with its root multiplicity in $D_W$ \emph{after} wiring.
(There is a unique minimum on account of functionality of well-behaved sets of wires.)
Also, the \emph{natural injection} of $D$ into $\wiring{D}{W}$ is
the map $\targets{t}{D}{W}$ which returns the same vertex as $\targetw{t}{D}{W}$ without its root multiplicity.
\end{defi}

Since $W'$ is a well-behaved set of wires for the drag $\gwiring{D}{W}$ by Lemma~\ref{l:wrec}, the recursive call $\wiring{(\gwiring{D}{\wir{s}{r}})}{W'}$ makes sense. 
The functions, resolution and natural injection, could also be defined by induction in the same way that the post-wiring drag was. 
Note that in case $t$ is an internal vertex, the vertex itself is not changed by wiring, but its multiplicity might be. 

Note also that a cycle is generated in a wired drag in case a sprout $s$ is accessible in the original drag from its resolution, implying that a loop (a cycle of length 1) can only be generated on an internal vertex. 

\begin{exa}
\label{ex:gc}
Figure~\ref{f:gc} displays two examples of wiring, illustrating the recursive calculation of the result. 
We start with a drag union of two drags, and a well-behaved set $W$ of two wires written underneath. 
The number of times a vertex is a root is indicated next to the root arrow's origin; the number 1 is often omitted. 
Labels can serve here as vertex names since no two vertices have the same label. The middle drag is the result of the first step of the calculation, with the remaining set of one wire written again underneath.

For the first calculation, both edges that ended up in sprout $x$, which has disappeared, are now redirected to sprout $y$, which has a single root left, since two roots have been utilized by redirecting the edges $\edge(f,x)$ and $\edge(h,x)$. 
The rightmost drag is the final result, there is no set of wires left. $3$ roots out of 4 in $h$ have been used by redirecting the 3 edges ending up in $y$ to $h$, while the root of $y$ is lost. 

Starting with 2 roots instead of $3$ for $y$ would yield the same result.
Note also that starting with the wire $\wir{y}{h}$ in which $y$ is not maximal would not make sense, since $\wir{x}{y}$ would not be a wire anymore after the first step of the calculation.
The tracing of $\targetw{x}{}{W}$ is indicated in red, $x$ moving to $y$ and then to $h$. 

The second calculation is similar. We note that it yields the same result.
This is no coincidence, as we shall see.
\qed\end{exa}

\begin{figure}[!t]
\setlength{\unitlength}{0.7cm} 

\hspace*{5mm}
\begin{picture}(17,11.5)(0,-8)   
\put(0,2.5){    
\begin{picture}(1.5,3.5)(0,0)
\put(0.05,1){\textcolor{black}{$\da$}}
\put(0.,0.5){\textcolor{black}{$f$}}
\put(0.15,0){\textcolor{black}{$\searrow$}}

\put(1.75,1.4){\textcolor{black}{\scriptsize 4}}
\put(1.75,1){\textcolor{black}{$\da$}}
\put(1.7,0.5){\textcolor{black}{$h$}}
\put(1.35,0){\textcolor{black}{$\swarrow$}}

\put(0.8,-0.3){\textcolor{red}{$x$}}
\put(0.8,0.1){\textcolor{red}{$\da$}}
\put(0.8,0.5){\textcolor{red}{\scriptsize 1}}


\put(3.75,1){\textcolor{black}{$\da$}}
\put(3.7,0.5){\textcolor{black}{$g$}}
\put(3.7,0.){$\da$}
\put(3.7,-0.4){$y$}
\put(3.1,0.25){\textcolor{black}{\scriptsize 3}}
\put(3.1,-0.1){\textcolor{black}{$\searrow$}}

\put(0.75,-1.5){$\{\wir{x}{y},\; \wir{y}{h}\}$}
\end{picture}
}

\put(6.5,3){\textcolor{black}{$=$}}

\put(5.5,2.5){
\begin{picture}(4,3.5)(0,0)
\put(-4,0){
\put(6.8,1){\textcolor{black}{$\da$}}
\put(6.75,0.5){\textcolor{black}{$f$}}
\put(7.1,0){\textcolor{black}{$\searrow$}}

\put(8.7,1.4){\textcolor{black}{\scriptsize 4}}
\put(8.7,1.){\textcolor{black}{$\da$}}
\put(8.7,0.55){\textcolor{black}{$h$}}
\put(8.2,0){\textcolor{black}{$\swarrow$}}

\put(7.75,1){\textcolor{black}{$\da$}}
\put(7.75,0.5){\textcolor{black}{$g$}}
\put(7.75,0){\textcolor{black}{$\da$}}

\put(7.75,-0.4){\textcolor{red}{$y$}}
\put(7.05,-0.4){\textcolor{red}{$\ra$}}
\put(6.75,-0.4){\textcolor{red}{\scriptsize 1}}
}
\put(3,-1.5){$\{\wir{y}{h}\}$}
\end{picture}
}

\put(12.4,3){\textcolor{black}{$=$}}

\put(7.1,2.5){
\begin{picture}(4,3.5)(0,0)
\put(7.1,1){\textcolor{black}{$\da$}}
\put(7.05,0.5){\textcolor{black}{$f$}}
\put(7.2,0){\textcolor{black}{$\searrow$}}

\put(8.5,1.){\textcolor{black}{$\da$}}
\put(8.5,0.55){\textcolor{black}{$g$}}
\put(8.1,0){\textcolor{black}{$\swarrow$}}

\put(7.75,-0.4){\textcolor{black}{$h$}}
\put(7.65,-0.85){\textcolor{black}{$\circlearrowleft$}}
\put(7.65,-0.85){\textcolor{black}{$\circlearrowleft$}}
\put(7.2,-0.4){\textcolor{black}{$\ra$}}
\put(6.9,-0.4){\textcolor{black}{\scriptsize 1}}
\end{picture}
}


\put(0,-2.5){    
\begin{picture}(1.5,3.5)(0,0)
\put(0.05,1){\textcolor{black}{$\da$}}
\put(0.,0.5){\textcolor{black}{$f$}}
\put(0.15,0){\textcolor{black}{$\searrow$}}

\put(1.75,1.4){\textcolor{black}{\scriptsize 4}}
\put(1.75,1){\textcolor{black}{$\da$}}
\put(1.7,0.5){\textcolor{black}{$h$}}
\put(1.35,0){\textcolor{black}{$\swarrow$}}

\put(0.8,-0.3){\textcolor{black}{$x$}}
\put(0.8,0.1){\black{$\da$}}
\put(0.8,0.5){\textcolor{black}{\scriptsize 1}}


\put(3.75,1){\textcolor{black}{$\da$}}
\put(3.7,0.5){\textcolor{black}{$g$}}
\put(3.7,0.){$\da$}
\put(3.7,-0.4){\textcolor{red}{$y$}}
\put(3.1,0.25){\textcolor{red}{\scriptsize 3}}
\put(3.1,-0.1){\textcolor{red}{$\searrow$}}

\put(0.75,-1.5){$\{\wir{y}{h},\; \wir{x}{h}\}$}
\end{picture}
}

\put(6.5,-2){\textcolor{black}{$=$}}

\put(5.5,-2.5){
\begin{picture}(4,3.5)(0,0)
\put(-4,0){
\put(6.8,1){\textcolor{black}{$\da$}}
\put(6.75,0.5){\textcolor{black}{$f$}}
\put(7.1,0){\textcolor{black}{$\searrow$}}

\put(9.6,0.55){\textcolor{black}{\scriptsize 3}}
\put(9.1,0.55){\textcolor{black}{$\la$}}
\put(8.7,0.55){\textcolor{black}{$h$}}
\put(8.2,0){\textcolor{black}{$\swarrow$}}

\put(8.7,1.05){\textcolor{black}{$\da$}}
\put(8.7,1.6){\textcolor{black}{$g$}}
\put(8.7,1.95){\textcolor{black}{$\da$}}

\put(7.75,-0.4){\textcolor{red}{$x$}}
\put(7.05,-0.4){\textcolor{red}{$\ra$}}
\put(6.75,-0.4){\textcolor{red}{\scriptsize 1}}
}
\put(3,-1.5){$\{\wir{x}{h}\}$}
\end{picture}
}

\put(12.4,-2){\textcolor{black}{$=$}}

\put(7.1,-2.5){
\begin{picture}(4,3.5)(0,0)
\put(7.1,1){\textcolor{black}{$\da$}}
\put(7.05,0.5){\textcolor{black}{$f$}}
\put(7.2,0){\textcolor{black}{$\searrow$}}

\put(8.5,1.){\textcolor{black}{$\da$}}
\put(8.5,0.55){\textcolor{black}{$g$}}
\put(8.1,0){\textcolor{black}{$\swarrow$}}

\put(7.75,-0.4){\textcolor{black}{$h$}}
\put(7.65,-0.85){\textcolor{black}{$\circlearrowleft$}}
\put(7.65,-0.85){\textcolor{black}{$\circlearrowleft$}}
\put(7.2,-0.4){\textcolor{black}{$\ra$}}
\put(6.9,-0.4){\textcolor{black}{\scriptsize 1}}
\end{picture}
}
\end{picture}

\vspace{-2.5cm}
\caption{Formation of cycles via composition, two versions.}\label{f:gc}
\end{figure}
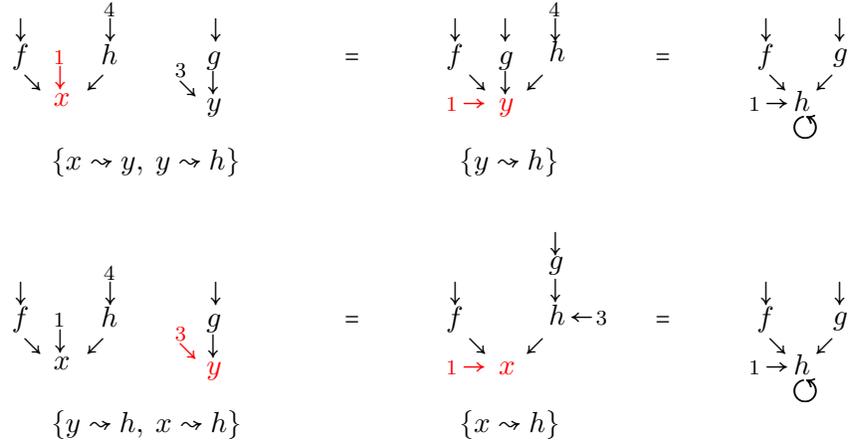

For the recursive calculation of Definition~\ref{d:wiring} to make sense, we need to show that the resulting drag does not depend upon the choice of a maximal wire $\wir{s}{r}$:

\begin{lem}
\label{l:wd}
Given a drag $D=\langle V, R, L, X, S\rangle$ and a well-behaved set of wires $W$, $\wiring{D}{W}$ and $\targetw{v}{D}{W}$, for each vertex $v\in V$, are well-defined.
\end{lem}

\begin{proof}
By induction on the size of $W$. 
We give the proof for the wiring definition.

If $W$ is empty, the result is straightforward. If it
contains a single maximal wire, the induction hypothesis applies.
Otherwise, let
$W=\wir{s}{r}\cup \wir{s'}{r'}\cup W'$ where  $\wir{s}{r}$ and $\wir{s'}{r'}$ are distinct wires, maximal for $>_W$. Let
$D'=\gwiring{(\gwiring{D}{\wir{s}{r}})}{\wir{s'}{r'}}$ and $D"=\gwiring{(\gwiring{D}{\wir{s'}{r'}})}{\wir{s}{r}}$. Note that $W'$ is well-behaved for both $D'$ and $D"$, even if $r=r'$.
By functionality, $s\neq s'$, and by maximality, $s'\neq r$ and $s\neq r'$; hence, redirecting the edges ending up in $s$ to $r$ and those ending up in $s'$ to $r'$ can be done in any order, showing that $D'=D"$.

Using Definition~\ref{d:wiring} twice for each calculation, 
 $\wiring{D}{W}=\wiring{D'}{W'}$ if $\wir{s}{r}$ is chosen first, and
  $\wiring{D}{W}=\wiring{D"}{W'}$ if $\wir{s'}{r'}$ is chosen first. 
  Since $D'=D"$, we conclude by the induction hypothesis that $\wiring{D}{W}$ is well-defined.
\end{proof}

Wiring has three important properties illustrated at Figure~\ref{f:gc}:

\begin{lem}
\label{l:edges}
\label{l:targetinv}
\label{l:idpw}
Let $D$ be a drag and $W$ a well-behaved set of wires of $D$. Then,
\begin{enumerate}
    \item $\wiring{D}{W}$ and $D$ have the same internal vertices and
    total number of edges---not counting roots as edges;
    \item all vertices $u$ of $\wiring{D}{W}$ have the same indegree in $D$ and $\wiring{D}{W}$;
    \item whenever $W'$ is a well-behaved set of wires such that $\forall t\in\Vertex{D} \st \targetw{t}{D}{W}=\targetw{t}{D}{W'}$, then $\wiring{D}{W}=\wiring{D}{W'}$.
\end{enumerate}
\end{lem}

The fact that wiring preserves indegrees, just like monomorphisms must, will turn out to be a key observation (Lemma~\ref{l:ipmconverse} below).

\begin{proof}
The third property follows from the other two, which hold because
internal vertices, total number of edges, and indegree of vertices are all preserved by an elementary wiring step---thanks to Lemma~\ref{l:pres} for indegree.
\end{proof}

\subsection{Coherence}\label{s:cohere}

We haven't yet required that different wires sharing the same label satisfy an assumption implying that the corresponding sprouts are related. More precisely, given a drag $D$ and a set of wires $W$, we want two sprouts $s:x$ and $t:x$ of $D$, with the same label $x$, to become equivalent after wiring. Ideally, we would like to check the property without computing $\wiring{D}{W}$, that is, read it on $D$ and $W$. In \cite{DBLP:journals/tcs/DershowitzJ19}, we required that $r=r'$ for any two wires $\wir{s:x}{r}$ and $\wir{s':x}{r'}$ belonging to $W$, a model called \emph{forced sharing}, and suggested that a more general equivalence should be drag isomorphism. 

It turns out, however, that drag isomorphism is too weak.
Take for example the drag $D=f(x,x,y,z,a,b)$ and the set of wires $W=\{\wir{x_1}{y},\wir{x_2}{z}, \wir{y}{a}, \wir{z}{b}\}$. The sprouts $x_1$ and $x_2$ are replaced by equivalent drags, but won't remain equivalent in $\wiring{D}{W}$ since $x_1$ will be eventually replaced by $a$ and $x_2$ by $b$. We could then expect that equimorphism, which is weaker than equality but stronger than isomorphism, could do.

And indeed, equimorphism in $D$ of the subdrags generated by $r$ and $r'$ is one possible answer, although not the most general one. 
There is however a difficulty: 
Equimorphism is not preserved along the wiring process, although it is finally restored. 
For an example, consider the drag $D$ made of three copies of $f(x)$, numbered 1,2,3. 
Let $W= \wir{x_3}{f_1}\cup W'$, with $W'=\{\wir{x_1}{f_2},\wir{x_2}{f_3}\}$, be a set of wires that satisfies this equivalence condition. 
Wiring the sole wire $\wir{x_3}{f_1}$ yields the drag $D'$ made of three subdrags, the first two are still the same as before while the last has become $f_3(f_1(x_1))$. 
Wires $W'$ no longer satisfy the property, since $x_1$ maps to $f_2$, which generates the subdrag $f_2(x_2)$, while $x_2$ maps to $f_3$, which generates now the subdrag $f_3(f_2(x_2))$, two non-isomorphic subdrags. 
On the other hand, applying all wires at once yields the drag which has three vertices $f_1, f_2, f_3$ and three edges $\edge(f_1,f_2), \edge(f_2,f_3), \edge(f_3,f_1)$. 
Obviously, the subdrags generated by $f_1, f_2, f_3$ are still equimorphic. 
This is the reason why we did not include coherence into the definition of a well-behaved set of wires: recurring on $W$ would have become impossible.
Wiring, hopefully, makes sense even when $W$ is not coherent, hence our choice. 
The question however remains, whether we need a property weaker that equimorphism. 
To illustrate the need, let us consider the drag $D=f(x,x,y,a)$ and the set of wires $W=\{\wir{x_1}{y},\wir{x_2}{a}, \wir{y}{a}\}$. 
Obviously, $y$ and $a$ do not generate equimorphic drags in $D$, but they do in $\wiring{D}{W}$. 
This tells us that the condition should be checked on the result of wiring $W$ in $D$.
We now give the formal definition of a coherent set of wires:

\begin{defi}[Coherence]
    Let $D$ be a drag and $W$ a well-behaved set of wires of $D$. We say that $W$ is \emph{coherent} (\emph{strongly coherent}, respectively) if it satisfies the two following properties:
    \begin{enumerate}
        \item 
        Fullness: All sprouts with the same label are the origin of a wire as soon as one is.
        \item
         Equimorphism: For any two wires $\wir{s:x}{r}$ and $\wir{s':x}{r'}$ of $W$, $\targets{r}{D}{W}$ and $\targets{r'}{D}{W}$ ($r$ and $r'$, respectively)
           generate equimorphic subdrags of $\wiring{D}{W}$ (of $D$, respectively).
          
    \end{enumerate}
\end{defi}

Note that forced sharing is a simple, particular case of strong coherence.

Like for the sum operation, wiring relates to the existence of certain morphisms between the wired drag and the original drag that will play a key r\^ole when it comes to rewriting drags. 

\begin{lem}
\label{l:injw}
\label{l:iwm}
Let $D$ be a drag having no isolated sprout and $W$ a well-behaved, coherent set of wires of $D$. 
Then the \emph{natural injection} from $D$ to $\wiring{D}{W}$ defines a monomorphism. 
\end{lem}

\begin{proof}
The natural injection $\iota$ from $D$ to $\wiring{D}{W}$ is a map $\iota$ that is the identity on the internal vertices of $D$, hence preserves their labels, is injective and indegree preserving by Lemma~\ref{l:idpw} (2), and maps every sprout to its resolution;

Since every vertex of $D_W$ is a vertex of $D$, the set of entering edges is empty. Edges of $D$ of the form $\edgen(u,i,s_i)$, where $s_i$ is a sprout such that $v'=\omicron(s_i)$, create the edge $\edgen(\omicron(u),i,v')$ of $D_W$. This edge takes the place of a root at $v$ in $D$ if the inverse image of $\iota$ contains an internal vertex or is a sprout $s=\omicron(s)$ shared between $D$ and $D_W$, a root that has therefore disappeared in $D_W$, hence ensuring root preservation. The case where $v'$ is a sprout whose inverse image is a set of sprouts all different from $s$ is simply impossible here since a resolution vertex must be a vertex of $D$.
We are left with showing that $\iota$ preserves equimorphic subdrags of $D_W$, which follows from coherence of $W$.
\end{proof}

Coherence can actually be checked without requiring the full computation of $\wiring{D}{W}$: any property implying coherence and preserved by wiring would do. This is of course the case of forced sharing, but we can do better.

\begin{lem}
Given a drag $D$, a well-behaved set of wires $W$ of $D$ is coherent iff 
$W=W'\cup W"$ for some coherent $W'$ and strongly coherent $W"$.
\end{lem}

In words, coherence is achieved in the drag resulting from the whole computation as soon as, for every variable $x$, the various wires $\wir{s:x}{r}$ generate equimorphic subdrags in the drag computed so far.
Note that only  nonlinear variables of $D$ need be checked for strong coherence.

\begin{proof}
The only-if direction follows directly from the definition of coherence by taking $W"=\varnothing$. The converse is by induction on the size of $W"$, the base case being obtained for $W"=\varnothing$, which yields coherence of $W=W'$ by assumption.

If $W"$ is non-empty, let $x$ be a variable maximal in $>_{W"}$ and $W"=W^x\cup W"'$, where $W^x=\{\wir{s_i:x}{r_i}\}$ is the set of wires in $W"$ whose origins are labeled $x$. The origins of wires in $W"'$ are therefore labeled by variables other than $x$. We will first show that (i) $W'\cup W^x$ is coherent, then that (ii) $W"'$ is strongly coherent with respect to $D_{W'\cup W^x}$, before  concluding that $W$ is coherent via the induction hypothesis. 

\noindent (i) Since $x$ is maximal, the subdrags generated by the $r_i$'s are identical in $D_{W'}$ and $(D_{W'})_{W^x}$, establishing property (i).

\noindent (ii) Since replacing sprouts labeled $x$ in equimorphic subdrags of $D_{W'}$ by equimorphic subdrags of $D_{W'}$ yields equimorphic subdrags of $(D_{W'})_{W^x}$, property (ii) follows.
\end{proof}

\subsection{Product}
While wiring operates on a single drag, product operates on a pair of drags via a connecting device we call a \emph{switchboard}:

\begin{defi}[Switchboard]
\label{d:switchboard}
Given two disjoint drags $D, D'$, a \emph{switchboard} $\xi$ for $D,D'$ is a pair  of partial functions $\langle \xi_D:\spr{D}\ra\roots{D'},\, \xi_{D'}:\spr{D'}\ra\roots{D}\rangle$, called \emph{switchboard components}, such that  $\xi_D\cup\xi_{D'}$ is a coherent well-behaved set of wires for $D\oplus D'$.
We also say that $\ext(D',\xi)$ is an \emph{extension} of $D$ and $D'$ its \emph{context extension}.
Switchboard $\xi$ is \emph{one-way} if either one of $\xi_D$ and $\xi_{D'}$ has an empty domain.
\end{defi}

Drags $D$ and $D'$ being disjoint, $\xi_D\cup\xi_{D'}$ is well defined. Therefore, $\xi$ can be identified with $\xi_D\cup\xi_{D'}$.
Note further that $\xi_D$ and $\xi_{D'}$ need not be true injective functions as in~\cite{DBLP:journals/tcs/DershowitzJ19}, where roots were lists with repetitions: 
Injectivity has been adapted to multisets of roots in our definition of a well-behaved set of wires.

Composition can now be defined as a wiring operation on $D\oplus D'$:

\begin{defi}[Composition]
\label{d:composition}
Let $D = \langle V, R, L, X, S\rangle$ and $D' = \langle V',
  R', L', X', S'\rangle$ be disjoint drags,
and $\xi$ a switchboard for $D, D'$.
Their \emph{cyclic composition} or \emph{product} is the drag $D \otimes_\xi D'=\wiring{(D\oplus D')}{\xi}$.
\end{defi}

\begin{exa}
\label{ex:gcsuite}
In Example~\ref{ex:gc}, the first drag is the union of two drags that share no vertices, and the set of wires is just their switchboard. 
The result is therefore the composition of these two drags with respect to that switchboard. 
\qed\end{exa}

Lemma~\ref{l:edges}(1) applies to composition: 
The total number of edges of $D\otimes_\xi D'$ is the sum of the number of edges of $D$ and $D'$, a property already noted in~\cite{DBLP:journals/tcs/DershowitzJ19}. 
The indegree is preserved at all vertices of $D\otimes_\xi D'$ since indegrees are preserved by the sum of disjoint drags and by wiring. Finally, the injection of $D$ into $D\otimes_\xi D'$ is a monomorphism by Lemmas \ref{l:ism}(2) and \ref{l:injw}(2).

When restricted to one-way switchboards, cyclic composition is indeed a substitution, and is dubbed "sequential" in~\cite{Gadducci07} and one-way in \cite{DBLP:journals/tcs/DershowitzJ19}. Many other names coexist that target other classes of graphs and particular cases of composition. 

A direct definition of a switchboard and of composition of two disjoint drags connected by a switchboard was given in~\cite{DBLP:journals/tcs/DershowitzJ19}. 
Our definition here assumes 
(a) that roots are multisets, instead of the lists there, 
(b) that resolution vertices have enough roots for transferring the edges of its corresponding origins, instead of edges and roots there, and 
(c) that coherence is ensured via equimorphism instead of sharing. 
These two models therefore have  different behaviors.

Apart from these differences, both definitions are similar:
our goal now is to give a definition of composition via wiring in the style of the direct definition used in~\cite{DBLP:journals/tcs/DershowitzJ19}.
In the context of composition of two drags $D,D'$ with respect to switchboard $\xi$, we define:

\begin{defi}[Resolution]
\label{d:target}
Let $D=\langle V, R, L, X, S\rangle$ and $D'=\langle V', R', L', X', S'\rangle$ be disjoint drags, and $\xi$ a switchboard for $D, D'$. 
The \emph{resolution} of a sprout $s\in S\cup S'$ is the vertex
$\target{s}=\targets{s}{D\oplus D'}{\xi}$, viewing the switchboard $\xi$ as the set of wires $W$ in Definition~\ref{d:wiring}.
\end{defi}

The direct definition of composition has now become a simple property of the wiring definition:

\begin{lem}[Composition]
\label{l:prod}
Let $D = \langle V, R, L, X, S\rangle$ and $D' = \langle V',
  R', L', X', S'\rangle$ be compatible drags, and $\xi$ be a switchboard for $D, D'$. Then  $D \otimes_\xi D' = \langle V", R", L", X", S"\rangle$,  
where
\begin{enumerate}
\item
$V"  =  (V \cup V')\setminus \Dom{\xi}$;
\item
$S"  =  (S \cup S') \setminus \Dom{\xi}$;
\item
$R"(v)= \begin{cases}
R(v) -\Sigma_{\target{w}=v\wedge w\neq v}\pred(w,D)  & \mbox{if }  v \in  R\setminus \Dom{\xi}\\
R'(v) -\Sigma_{\target{w}=v\wedge w\neq v}\pred(w,D') & \mbox{if } v \in  R'\setminus \Dom{\xi} ;
\end{cases}$
\item
$L"(v)  =  \begin{cases}
L(v)  & \mbox{if }  v \in  V\cap V"\\
L'(v) & \mbox{if } v \in  V'\cap V" ;
\end{cases}$
\item
$X"(v)  =  \begin{cases} 
\target{X(v)}  
& \mbox{if }  v \in  V\setminus S\\
\target{X'(v)} & \mbox{if } v \in  V'\setminus S' .
\end{cases}$
\end{enumerate}
\end{lem}

The calculation of the multiset of roots of the composition is based on the preservation of indegrees by wiring, making it very simple: 
Edges accumulate along the computation until the very end, at which point they must be compensated for by
the roots of the resolution.



\begin{proof}
By Definition~\ref{d:composition}, $D\otimes_\xi D'=\wiring{(D\oplus D')}{\xi}$. We therefore show that $\wiring{(D\oplus D')}{\xi}$ is the drag obtained from $D\oplus D'$ by (i) removing from the set of vertices all sprouts in $\Dom{\xi}$, (ii) redirecting all edges ending up in a removed sprout to its associated resolution, and (iii) keeping the indegree unchanged for all vertices in $(V \cup V')\setminus \Dom{\xi}$,
which determines their number of roots as stated. The proof is by induction on the size of $\xi$ considered as a set of wires. 
There are two cases:

-- In the empty case, $D\otimes_{\varnothing}D'=D\oplus D'$. Properties (i), (ii), and (iii) hold trivially.

-- Otherwise, by Lemma~\ref{l:wd}, we can choose any wire
$\wir{s}{r}$ such that 
$\xi=W\cup\wir{s}{r}$, where $s$ is maximal in $\xi$; hence
$D\otimes_\xi D'= \wiring{((\gwiring{D\oplus D')}{\wir{s}{r}})}{W}$.
We then conclude by the induction hypothesis, since $W$ has one wire fewer than does $\xi$, that $D\otimes_\xi D'$ is obtained from
 $\gwiring{(D\oplus D')}{\wir{s}{r}}$ as claimed. Since $s$ is maximal, it follows from the definitions of resolutions $\targets{x}{D}{\xi}$ and $\target{x}$ that $\wiring{(D\oplus D')}{\xi}$ is obtained from $D\oplus D'$ as claimed.
\end{proof}

Not all vertices of the composition $D\otimes_\xi D'$ may be reached via $\xi$ from the sprouts of one of them. We denote by $\xi(D)$ the subdrag of $D'$ generated by the vertices in $\xi_D(\Dom{\xi_D})$, which contains all vertices of $D'$ reached from sprouts of $D$.

There are important particular cases of switchboards that impose additional, hence stronger, coherence conditions:
\begin{itemize}
\item\emph{Equality switchboard}:
    $\forall s\neq t\in  S\cup S'$ such that $L(s) = L(t)$, we have $\xi(s)=\xi(t)$.
\item\emph{Inequality switchboard}:
$\forall s\neq t\in  S\cup S'$ such that $L(s) = L(t)$, we have $\rest{(D\oplus D')}{\xi(s)}$ and
$\rest{(D\oplus D')}{\xi(t)}$ have no vertex in common.
\end{itemize}

Composition is said to \emph{force sharing} if $\xi$ is a switchboard with equality, and to \emph{force cloning} if $\xi$ is a switchboard with inequality. Term rewriting is based on cloning switchboards while dag rewriting is based on equality switchboards. Our notion of composition can potentially do both, and can also achieve partial sharing by having  $\rest{(D\oplus D')}{\xi(s)}$ and
$\rest{(D\oplus D')}{\xi(t)}$ sharing subdrags, in particular sprouts, when possible. These question will be studied in more detail in Section~\ref{ss:drvtr}.

\paragraph*{\bf Notations:} Given two drags $D$ and $D'$ and a switchboard $\xi= \{\wir{s_i}{r_i}\}_i$ for $(D,D')$, we will sometimes need to restrict the switchboard $\xi$ to some subsets of vertices $V$ of $\Vertex{D}$ and $V'$ of $\Vertex{D'}$. 
We will therefore denote by $\xi_{V\ra V'}$ the restriction of switchboard $\xi$ to the set of wires $\{\wir{s_i}{r_i}\st s_i\in V, r_i\in V'\}_i$. 
We will use $\xi_{V,V'}$ for the switchboard $\xi_{V\ra V'}\cup \xi_{V'\ra V}$, $\xi_{V\ra}$ for the switchboard $\xi_{V\ra \Vertex{D'}}$, and $\xi_{\ra V'}$ for the switchboard $\xi_{\Vertex{D}\ra V'}$.

\section{Decomposition of Drags}
\label{s:decd}

We now investigate to what extent a drag can be expressed in terms of simpler drags by means of sum and product so as to obtain a \emph{drag expression}:

\begin{defi}[Drag expression]
By a \emph{drag expression}, we mean any expression built from a given set of \emph{drag components} $\{D_i\}_i$ such that $D_i$ and $D_j$ share no vertex nor variable if $i\neq j$, by means of sum and product of drags. Drags occurring in a drag  expression are its \emph{drag components}. A drag $D$ is a \emph{trivial} drag expression whose drag $D$ is its only drag component.
\end{defi}

Note that product alone would suffice to build drag expressions, writing a sum $D\oplus D'$ as the product $D\otimes_{\varnothing} D'$.

An initial, straightforward answer is that we can decompose a drag according to some subset of its internal vertices by using the corresponding subdrag and associated context:


\begin{lem}[Reconstruction]
\label{l:rec}
Given a drag $D$ and a subset $W$ of its vertices, then
$D=\rest{D}{W}\otimes_\xi \cont{D}{W}$ for some switchboard $\xi$.
\end{lem}

\begin{proof}
Using the notations of Definition~\ref{d:subd}, it suffices to define $\xi$, which maps every new sprout $s_v$ of the restriction (of the context, respectively) to the vertex $v$ of the context
(of the restriction, respectively). The equality claim then follows easily using preservation of indegrees by wiring.
\end{proof}

We will indeed use the restriction of $D$ to some carefully-chosen single internal vertex, give a specific definition for that case, and treat some special cases separately. 

First, we need to define what are "atomic" drags, the kind we like to have in a full decomposition, and the non-atomic ones that we want to eliminate:

\begin{defi}[Connected drag]
\begin{enumerate}
\item
A \emph{connected} drag is any drag $\langle V,R,L,X,S\rangle$ whose set of vertices is generated by successor and equal labeling of sprouts, that is, any subset $W$ of $V$ closed under these two operations must be $V$ itself:
$$\forall W\subseteq V (\forall u\forall v \in V (uXv \st u\in W \mbox{ iff } v\in W) \mbox{ and } \forall s:x\forall t:x \in V (s\in W \mbox{ iff } t\in W)) \Rightarrow W=V$$
\item
A \emph{flat} drag is a connected drag with no non-trivial path between internal vertices.
\item
An \emph{atomic vertex} is a drag with only a single internal vertex 
and any number of roots and different sprouts as successors, all with different labels.
\item
A nonempty set of pairwise distinct sprouts is \emph{atomic} if they all share the same variable, each with any number of roots.
\item
An \emph{atomic drag} is an atomic vertex or an atomic set of sprouts.
\end{enumerate}
\end{defi}

For example, the drag with two internal vertices sharing one sprout is flat, while the drag made of a single loop on an internal vertex is not. 
Note also that the drag made of two drags $f(s:x)$ and $g(t:x)$ is connected by our definition, as if $s$ and $t$ where the same shared sprout.

A major property of non-flat connected drags is that they all have non-necessarily distinct internal vertices $u,v$ such that $u$ is a predecessor of $v$.

Atomic drags are of course flat, but there are also flat non-atomic drags. On the other hand, non-connected drags can always be written as a sum of connected drags. We therefore consider the decomposition of non-flat connected drags first, before  addressing the case of flat non-atomic connected drags.

In what follows, we give a sequence of four transformation rules in the form of as many lemmas. 
These transformations take as input a drag expression $E$ containing a drag component $D$ that is not yet atomic, but instead: (i) comprises pairwise distinct connected components; or (ii) is
a non-flat connected component with an edge $\edge(u,v)$ between two distinct internal vertices; or (iii) is a non-flat connected component with a loop $\edge(u,u)$ on some internal vertex $u$; or else (v) is a flat connected component with sprouts $s_i:x$ that are either shared or sharing the variable label $x$, or both.
Clearly, any 
non-atomic drag belongs to one of these four categories. 
Therefore, applying these transformations repeatedly to a not-yet atomic drag component $D$ of $E$ will eventually transform $E$ into a drag expression $E'$  all of whose components are atomic drags. This is so because the drag expressions $E'$ obtained from $E$ are simpler than $E$, in some well-founded order, hence implying that any sequence of transformation is finite. 

Before specifying the transformation rules, we define the order used to compare drag expressions:

\begin{defi}
To a given drag $D$, we associate the triple $\langle \mathit{\#I}, \mathit{In}, M \rangle$, where
\begin{enumerate}
\item $\mathit{\#I}$ is the number of its internal vertices plus the edges between them;
\item $\mathit{In}$ is the multiset of its sprouts' indegrees;
\item $M$ is the multiset counting, for each variable $x\in\Var{D}$, the number of its sprouts that are labeled $x$.
\end{enumerate}
Denoting by $\gtnat$ the usual order on natural numbers, triples---hence drags---are compared in the well-founded order ${\gg} = (\gtnat, \gtnat^{mul}, \gtnat^{mul})^{lex}$, where $^{lex}$ and $^{mul}$ denote lexicographic and multiset extensions of an order, respectively. 
Drag expressions, interpreted as the multiset of interpretations of their drag components, are compared in the well-founded order $\gg^{mul}$.
\end{defi}

\begin{lem}
\label{l:decomp}
Let $D$ be a drag made of pairwise distinct connected drags $D_1, \ldots, D_n$, with $n>1$. Then, $D=D_1\oplus \cdots \oplus D_n$
is a drag expression such that $\forall i \st D \gg D_i$.
\end{lem}

\begin{proof}
The only difficulty is the ordering statement. 
If there are two or more $D_i$'s with internal vertices, the result is clear. 
If all $D_i$'s are made of sprouts, their first component is 0, as for $D$, but the second has fewer 0's than in $D$, hence decreases strictly. 
If, say, $D_1$ has at least one internal vertex but all other $D_i$'s do not, then $D$ has strictly more internal vertices than the $D_i$'s, while $D_1$ has the same number of them as does $D$. 
But its multiset of sprout indegrees must have decreased strictly.
\end{proof}

\begin{lem}[Decomposition]
\label{l:decomp1}
Given a non-flat connected drag $D=\langle V,R,L,X,S\rangle$, let $v$ be an internal vertex of $D$ of indegree $p$, labeled $f$ of arity $n$, having at least one predecessor $u\neq v$, and whose successors in $D$ are the vertices $v_1,\ldots, v_n$.
Then---denoting $v$ by $f$:
$$D=D_f \otimes_\xi D' \mbox{ with } \xi=\{\wir{s}{v},\wir{s_1}{w_1},\dots, \wir{s_n}{w_n}\}$$
is a drag expression such that $D\gg D_f$ and $D\gg D'$, where:
\begin{enumerate}
\item
$D_f=f^{[p]}(s_1:x_1,\ldots,s_n:x_n)$;
\item
$D'=\langle V',R',L',X',S'\rangle$;
\item
$V'= (V\setminus v)\cup s$, where $s$ is a fresh sprout;
\item
$\forall u\in V\setminus (v\cup \{v_i\}) \st R'(u)=R(u)$;  $\forall u\in \{v_i\}_i \st R'(u)=R(u)+1$; $R'(s)= 0$;
\item
$\forall u\in V\setminus v \st L'(u)=L(u)$; $L'(s)=x$, where $x$ is a fresh variable;
\item
$\forall u,w\in V\setminus s \st uX'w $ iff  $uXw$; $\forall u\in V\setminus s \st uX's$ iff $uXv$;
\item
$S'=S\cup s$;
\item
$w_i=s$  if  $v_i=v$, and otherwise it is $v_i$.
\end{enumerate}
\end{lem}

\begin{proof}
Note that $D$ is a non-flat connected drag, since $\edge(u,v)$.
The switchboard is designed so that it is trivially coherent and well-behaved, and $D$ is reconstructed from two drag components sharing no vertex or variable. 
We do not and need not assert that $D'$ itself is a connected drag; it may not be if the subdrag generated by some $v_i$ has no shared vertex with its associated context drag.

The equality claim follows from preservation of indegrees by composition.

Drag $D'$ is smaller than $D$ since an internal vertex has been removed, and the edges between the remaining internal vertices are those of $D$. 

For the last claim,  notice that the drag $f^{[p]}(s_1:x_1,\ldots,s_n:x_n)$ has a single internal vertex $v$ and no edge $\edge(v,v)$.
\end{proof}

\begin{exa}
Consider a drag $D$, reduced to two internal vertices labeled $g$ and $a$, of arities 1 and 0, respectively, plus a single edge $\edge(g,a)$. 
We get $D=a^{[1]}\otimes_{s\mapsto a} g(s^{[0]})$.
\qed\end{exa}

\begin{lem}[Loop decomposition]
\label{l:decompl}
Let $D$ be a drag with an internal vertex $v$
of indegree $p$, labeled $f$ of arity $n$, whose successors in $D$ are the vertices $v_1,\ldots, v_n$, with $v_i=v$ for some $i$. Then,
$$D= D' \otimes_\xi t^{[1]}:y \mbox{ with } \xi=\{\wir{s}{t},\wir{t}{v}\}$$ 
is a drag expression such that  $D\gg D'$ and $D\gg t^{[1]}$, where:
\begin{enumerate}
\item
$D'=\langle V',R',L',X',S'\rangle$;
\item
$V'= V\cup s$, where $s$ is a fresh sprout;
\item
$\forall u\in V\setminus v \st R'(u)=R(u)$;  $R'(v)=R(v)+1$; $R'(s)= 0$;
\item
$\forall u\in V \st L'(u)=L(u)$; $L'(s)=x$, where $x$ is a fresh variable;
\item
$\forall u,w\in V \st uX'w $ iff $ uXw$ except for the edge $\edgen(v,i,v)$ of $D$ replaced by $\edgen(v,i,s)$ in $D'$;
\item
$S'=S\cup s$.
\end{enumerate}
\end{lem}

This lemma allows us to eliminate one loop at a time. 
When there are several loops on the internal vertex $v$, they  have to be eliminated one by one. 
We could of course give a more general statement eliminating them all at once, at the price of a slightly more complicated statement and proof.

\begin{proof}
Similar to the proof of Lemma~\ref{l:decomp1}, except for the ordering statements.

Here $D'$ has the same total number of internal vertices but strictly fewer edges between them since some edge $\edge(v,v)$ of $D$ has been replaced by an edge $\edge(v,s)$ in $D'$, which does not count because $s$ is a sprout. As for the other drag, it has no internal vertices, while $D$ has at least one.
\end{proof}

\begin{exa}
Consider a drag $D$ reduced to a single internal vertex $v$ labeled $f$ of arity 1, and a single looping edge $\edge(f,f)$. 
Then, we have $D=f^{[1]}(s:x)\otimes_{s\mapsto t, t\mapsto f} t^{[1]}$.
\qed\end{exa}

We next consider the case of non-atomic connected flat drags. 
They  are made of one or more internal vertices connected via their
successor sprouts, some of them being shared, or sharing the same variable label, or both. 
We will eliminate at once all connections related to a given variable label. 
If there are several variable labels involved in the connections, they will have to be eliminated one by one.

\begin{lem}[Unsharing decomposition]
\label{l:decomp2}
Let $D$ be a connected flat drag whose nonempty set $\{s_i:x\}^{n}_{i=1}$ ($n\geq 1$) of distinct sprouts sharing variable label $x$, having $q_i$ predecessors and indegree $p_i\geq q_i$, respectively, contains at least two sprouts, or one shared sprout, or both.
Then $$D=D'\otimes_\xi \left(s_1^{[p_1]}\ \cdots \ s_n^{[p_n]}\right)$$
is a drag expression such that $D\gg D'$ and $D\gg s_1^{[p_1]} \ldots s_n^{[p_n]}$, where:
\begin{enumerate}
\item
$D'$ is obtained from $D$ by replacing each vertex $s_i$ by fresh rootless sprouts $t_{i,1}:y_{i,1},~ \dots,~ t_{i,{q_i}}:y_{i,{q_i}}$ and every edge $\edge(v,s_i)$, if any, by edges $\edge(v,t_{i,j})$, with $1 \leq j \leq q_i$;
\item
$\xi=\{\wir{t_{i,k}}{s_i}\}_i$.
\end{enumerate}
\end{lem}
Note that the drag $s_1^{[p_1]} \cdots s_n^{[p_n]}$, all of whose  sprouts share the same variable, is an atomic variable drag that cannot be decomposed any further without violating our notion of drag expression, since all these sprouts share label $x$.

\begin{proof}
In case $p_i=0$, the switchboard $\xi$ maps a rootless sprout $t_{i_1}$ to a rootless sprout $s_i$. 
Since all sprouts labeled $x$, whether shared or nor, are renamed with the appropriate number of fresh sprouts, the resultant product is a drag expression. 
The equality statement is again a straightforward consequence of indegree preservation. 

Finally, $D'$ has the same number of internal vertices as $D$. If $n>1$, the number of sprouts labeled $x$ decreases strictly while
the multiset of sprout indegrees does not increase. Or else $p_1>1$ and
the  multiset of sprouts indegrees decreases strictly. We are left with the case of a flat drag with no internal vertices but several sprouts labeled $x$. In that case, $D$ and $D'$ have the same first two components, but the third decreases strictly. The other ordering statement is straightforward.
\end{proof}

Here are two examples of flat non-atomic drags decomposed into atomic drags:
\begin{align*}
f(s^{[0]}:x,t^{[0]}:x) &=f(s'^{[0]}:y_1,t'^{[0]}:y_2)\otimes_{\wir{s'}{s},\wir{t'}{t}} (s^{[1]}:x\quad  t^{[1]}:x) 
\\
f(s^{[0]}:x,s^{[0]}:x)
&=f(s':y_1,t':y_2)\otimes_{\wir{s'}{s},\wir{t'}{s}} s^{[2]} \\& \qquad\qquad\qquad\qquad\mbox{---assuming now that $s^{[0]}$ is shared.}
\end{align*}
Note the difficulty in representing the drag of the second example faithfully by means of the expression $f(s^{[0]}:x,s^{[0]}:x)$. We have to make explicit that $s^{[0]}$ is shared. On the other hand, our product-based representation is faithful; sharing is attained by calculating the drag expression. Note that this representation avoids the kind of unary binding notation usual in programming languages: The annotated product operation is a binder for its arguments.
 
The decomposition properties can be used to decompose a given drag $D$ into atomic drags, so as to obtain a drag expression built from atomic drags by means of sum and product. 
Initially, we are given a drag $D$, considered as a (trivial) drag expression. We get the following:

\begin{thm}
For every drag $D$, there exists a drag expression made of atomic drags whose evaluation yields $D$.
\end{thm}

\begin{proof}
By its definition, the order on drag expressions is monotonic: replacing in $E$ a drag component $D$ by a drag expression $D'$  all of whose components are strictly smaller than $D$ yields a strictly smaller drag expression $E'$.

Since every non-atomic dag is either made of several distinct connected components, or is a non-flat connected component, or is a flat connected component that is not yet atomic, all cases have been taken care of. Therefore, by applying Lemmas \ref{l:decomp}, \ref{l:decomp1}, \ref{l:decompl}, and \ref{l:decomp2} for as long as possible---termination being ensured by the well-founded order on drag expressions, we get a drag decomposition into atomic components.
\end{proof}

A drag can therefore be defined by induction from atomic pieces glued together by appropriate compositions, which gives rise to an (non-unique) algebraic notation for drags. 

\begin{exa}

Let $f$ be of arity 2 and $h$ of arity 1. The connected drag $D$
with vertices $f_1, f_2$ labeled $f$, vertices $h_1, h_2, h_3$ labeled $h$, and edges
$\edgen(f_1, 1, f_2), \edgen(f_1, 2, h_2), \edgen(f_2, 1, h_1), \edgen(f_2,2, h_2), \edgen(h_1,1,h_1), \edgen(h_2, 1, h_3), \edgen(h_3,1,f_2)$ has as its drag expression
\[
f_2^{[2]}(x_1,x_2) \otimes_{\wir{x_1}{h_1}, \wir{x_2}{h_2}, \wir{x_3}{f_2}}
\left(f_1(x_3, h_2^{[1]}(h_3(x_3))) \oplus h_1^{[1]}(\textsc{self})\right)
\]
where $x_3$ denotes a shared sprout,
obtained by applying successively Lemma~\ref{l:decomp1} to the vertex $f_2$ and then Lemma~\ref{l:decomp} to the resultant drag. We continue now with the right-hand side argument of $\otimes$. Applying first
Lemma~\ref{l:decomp1} at vertex $h_2$,  we get:
\[
(h_2^{[2]}(x_4) \otimes_{\wir{x_4}{h_3},\wir{x_5}{h_2}}
  (f_1(x_3,x_5) \quad h_3^{[1]}(x_3))) \oplus h_1^{[1]}(\textsc{self})
\]
Then applying Lemma~\ref{l:decomp2} to the flat sub-expression sharing sprout $s_3$, which is the right-hand side argument of $\otimes$, yields the drag expression:
\[
(h_2^{[2]}(x_4) \otimes_{\wir{x_4}{h_3}\wir{x_5}{h_2}}
  ((f_1(x_6,x_5) \otimes_{\wir{x_6}{x_3},\wir{x_7}{x_3}} h_3^{[1]}(x_7)))) \oplus h_1^{[1]}(\textsc{self}))
\]
And now applying Lemma~\ref{l:decompl} to the rightmost subdrag, being a loop, gives
\[
(h_2^{[2]}(x_4) \otimes_{\wir{x_4}{h_3},\wir{x_5}{h_2}}
  ((f_1(x_6,x_5) \otimes_{\wir{x_6}{x_3},\wir{x_7}{x_3}} h_3^{[1]}(x_7))))
  \oplus (h_1^{[2]}(x_6) 
  \otimes_{\wir{x_6}{x_8}\wir{x_8}{h_1}} x_8^{[1]}
\]
a decomposition with no non-atomic components left. 
Putting the pieces together, we get the following drag decomposition of the starting drag:
\[
\begin{array}{lll}
f_2^{[2]}(x_1,x_2) &\otimes_{\wir{x_1}{h_1}, \wir{x_2}{h_2}, \wir{x_3}{f_2}}\\
&\qquad \big(h_2^{[2]}(x_4) \otimes_{\wir{x_4}{h_3},\wir{x_5}{h_2}}
\quad((f_1(x_6,x_5) \otimes_{\wir{x_6}{x_3},\wir{x_7}{x_3}}\quad h_3^{[1]}(x_7)))\big)\\
&\qquad\oplus \\
&\quad\qquad\; (h_1^{[2]}(x_6) \otimes_{\wir{x_6}{x_8}\wir{x_8}{h_1}} \quad x_8^{[1]})\big)
\end{array}
\]
The following decomposition of the same drag
\[
\begin{array}{llll}
f_1(x_1,x_2) &\otimes_{\wir{x_1}{f_2}, \wir{x_2}{h_2}}\\
&\big(f_2^{[2]}(x_3,x_4) &\otimes_{\wir{x_3}{h_1},\wir{x_4}{h_2},\wir{x_6}{f_2}} \\ 
&\qquad\big(h_2^{[2]}(x_5) \otimes_{\wir{x_5}{h_3}} h_3^{[1]}(x_6) \big) 
&\oplus\ \  \big( h_1^{[1]}(x_7) \otimes_{\wir{x_7}{x_8},\wir{x_8}{h_1}} (x_8)^{[1]}\big) \big)
\end{array}
\]
cannot be obtained by using our lemmas, since we are missing an analog of Lemma~\ref{l:decomp1} applying to a vertex without ancestors of a non-flat drag. This additional decomposition lemma can be surmised by the reader;
it would be needed in case we were interested in all possible decompositions of a given drag.
Note that we end up here with a decomposition using again 8 fresh sprouts. 
This is no surprise: one sprout is needed for each incoming edge in the original drag, plus one for each loop, since a loop decomposition requires two wires.
The reader can verify that all switchboards used in these decompositions are well behaved and that the result is indeed the drag $D$ in both cases.
\qed\end{exa}

\section{Algebra of Drags}
\label{s:alg}
Once equipped with sum and product, drags
enjoy a very rich algebraic structure which is known to be suitable for expressing distributed computations~\cite{DBLP:journals/iandc/MeseguerM90,DBLP:conf/birthday/MeseguerMS97}.

\begin{lem}
\label{l:sumac}
Drag sum is associative, commutative, idempotent, and has the empty drag as an identity element. 
\end{lem}

Noting that a drag is compatible with itself, all these properties are straightforward. We proceed with product:

\begin{lem}
\label{l:prodac}
Drag product is associative, commutative, and has an identity element, the empty drag. Other identities are isolated sprouts provided they belong to the domain of the switchboard, and do not belong to its image. 
\end{lem}

\begin{proof}
Commutativity is straightforward here, because of the root structure as a multiset. The empty drag is again an identity for any drag $D$, since $\xi$ must be empty, and therefore
$D\otimes_\varnothing \varnothing= D\oplus \varnothing = D$. 
Let now $s^{[n]}$ be a drag reduced to a sprout which is not a vertex of $D$, and $\wir{s}{v}$ be a wire for $(D,s)$. By assumption, $s$ does not belong to the image of $\xi_D$.
Then, a straightforward calculation shows that \smash{$D \otimes_{\wir{s}{v}}s^{[n]}=D$}.

We are left with associativity. Consider the drag
$(C\otimes_\xi D)\otimes_\zeta E$. 
We prove first that $\xi\cup\zeta$ is a well-behaved set of wires for $C\oplus D\oplus E$. By the definition of a switchboard,
$\xi$ and $\zeta$ are well-behaved sets of wires for $C\oplus D$ and $(C\otimes_\xi D)\oplus E$, respectively. Since $\zeta$ maps remaining sprouts of $C\oplus D$ after wiring with $\xi$ to roots of $E$, and sprouts of $E$ to remaining roots of 
$C\oplus D$ after wiring with $\xi$, $\xi\oplus\zeta$ is well-defined and satisfies functionality, coherence and injectivity. It is also well-founded, since chains of sprouts for $\xi\cup\zeta$ are either chains of sprouts for $\xi$ or for $\zeta$ which both satisfy well-foundedness.

We  construct now two new sets of wires, $\zeta'$ for $D\oplus E$, and $\xi'$ for $C\oplus(D\otimes_{\zeta'} E)$
such that $\xi'\cup\zeta'=\xi\cup\zeta$, showing that
$\xi'\cup\zeta'$ is a well-behaved set of wires for $C\oplus D\oplus E$, which implies that $\xi'$ and $\zeta'$ are well-behaved sets of wires for $C\oplus(D\otimes_{\zeta'} E)$ and $D\oplus E$, respectively. We now classify each wire $\wir{s}{r}\in\xi\cup\zeta$ as a wire of $\xi'$ or $\zeta'$:
\begin{enumerate}
    \item $s\in\spr{C} \st  \wir{s}{r}\in\xi'$;
    \item $r\in\roots{C} \st  \wir{s}{r}\in\xi'$;
    \item $s\in\spr{D}\mbox{ and } r\in\roots{E} \st  \wir{s}{r}\in\zeta'$;
    \item $s\in\spr{E}\mbox{ and } r\in\roots{D} \st  \wir{s}{r}\in\zeta'$.
\end{enumerate}
The equality between both obtained drags now follows from routine calculations.
\end{proof}

Finally, we consider the distributivity law, which will be used later when investigating the categorical structure of drags.
Omitting switchboards,   this is an identity of the form $C\otimes (D\oplus E)=(C\otimes D)\oplus (C\otimes E)$. 
There are two obstacles: 
The first is that  the sums $D\oplus E$ and $(C\otimes D) \oplus (C\otimes E)$ must make sense, which requires that compatibility of drags $D$ and $E$ is preserved by their product with $C$. This is not true in general, but will require that all wires whose source sprout is in $C$ satisfy some \emph{safety} condition.
The second obstacle is that wirings between $D$ and $E$ going through $C$ in $C\otimes (D\oplus E)$ cannot be reproduced, in general, in $(C\otimes D)\oplus (C\otimes E)$, unless wiring again the result, which is not expected from a distributivity law whose r\^ole is to transform a product into a sum. 
We will therefore need to strengthen the safety condition for that purpose. 

\begin{defi}
\label{d:safe}
Let $D,E$ be compatible drags with shared subdrag $G$, $C$ a drag disjoint from $D\oplus E$, $\xi$ a switchboard for $(D\oplus E,C)$, $H$ the subgraph of $C$ generated by all vertices $v\in\Vertex{C}$ such that $\wir{s}{v}\in\xi_{G}$, $C'$ the context drag of $H$ in $C$, and $D', E'$ the context drags of $G$ in $D,E$, respectively. We say that switchboard $\xi$ is \emph{safe} if 
\begin{enumerate}
\item
$\forall \wir{t}{v}\in \xi_H \st v\in\Vertex{G}$;
\item
$\forall\wir{s}{u}\in \xi_D \,\forall\wir{t}{u}\in \xi_E \mbox{ such that } s,t\not\in G \st v\in\Vertex{H}$
\item 
$\forall\wir{s}{t}, \wir{t}{v}\in\xi \mbox{ such that } s\in \spr{D'}\mbox{($s\in\spr{E'}$, respectively)} \;\st v\notin\Vertex{E'} \mbox{($v\notin\Vertex{D'}$, respectively)}$.
\end{enumerate}
\end{defi}

We can now show distributivity under the safety assumption:
\begin{lem}[Distributivity]
\label{l:distr}
Let $D,E$ be compatible drags with shared subdrag $G$, $C$ a drag disjoint from $D\oplus E$, and $\xi$ a safe switchboard for $(D\oplus E,C)$. Then,
drags $D\otimes_\xi C$ and $E\otimes_\xi C$ are compatible with shared subdrag  $G\otimes_\xi C$, and
\[(D\oplus E) \otimes_\xi C = (D\otimes_\xi C)\oplus (E\otimes_\xi C).\]
\end{lem}
\begin{proof}
The definition of $H$ and safety condition (1) ensure that $G\otimes_\xi H$ is shared by drags $D\otimes_\xi C$ and $E\otimes_\xi C$, while safety condition (2) ensures that no other vertex is shared, hence $G\otimes_\xi H$ is exactly their shared subgraph, implying compatibility.

Safety condition (3) ensures that no edge is created between $C'$ and $D'$ 
by the product $(D\oplus E)\otimes_\xi C$ unless it is an edge in $G$ created by the product $G\otimes_\xi H$. Therefore, all edges of $(D\oplus E)\otimes_\xi C$ will be edges of $(D\otimes_\xi C)\oplus (E\otimes_\xi C)$, the converse being trivial.
\end{proof}
Since $\xi$ must map sprouts of $H$ to vertices of $G$ by condition (1) of Definition~\ref{d:safe}, then condition (3) is trivially true in case $C=H$. This remark will later allow us to use distributivity in order to characterize pushout objects by means of product and sum of drags in the proof of the important Lemma~\ref{l:pushout}.


\section{Sharing Equivalence}
\label{ss:se}

Next, we define and study the equivalence on drags defined by sharing subdrags, which will play an important r\^ole for defining rewriting. We will share subdrags that are equimorphic up to their roots, which do not play a r\^ole here.

\begin{defi}[Sharing]
\label{d:eqshare}
A drag is \emph{maximally shared} if no two distinct bare subdrags are equimorphic.
\end{defi}

Note that a maximally shared drag must be linear.

First, we define the \emph{maximally shared form} $\nf{D}$ of a drag $D$ by iterating the following \emph{sharing} transformation as long as necessary:
\begin{enumerate}
\item 
Assume $E_1, \ldots, E_n, F$ are all pairwise distinct maximally shared subdrags of $D$ whose bare versions are equimorphic to $\bare{F}$, called the \emph{class} of $F$ in $D$, and let $C_F$ be the \emph{context} of $\overline{F}=E_1\oplus\cdots\oplus E_n\oplus F$.  By Lemma~\ref{l:rec}, $D=C_F\otimes_\xi \overline{F}$ for some $\xi$. The class of $F$ is said to be \emph{trivial} if it consists of the single drag $F$ only ($n=0$), and \emph{nontrivial} otherwise ($n>0$).
\end{enumerate}

\begin{claim}
Given a drag $D$, assume that the bare version of some drag in a class $\overline{E}$ is equimorphic to the bare version of some strict subdrag of a drag in a class $\overline{F}$. Then, all drags in the class $\overline{F}$ contain at least one subdrag belonging to the class $\overline{E}$. This shows that the well-founded subdrag order lifts to classes of drags. As a consequence, some classes are minimal in this order.
\end{claim}

\begin{enumerate}
\setcounter{enumi}{1}
\item
Assuming $D$ is not maximally shared, there exists at least one minimal nontrivial class $\overline{F}$ in $D$. 
Let, therefore, $\omicron_i:E_i\hookrightarrow F$ be the equimorphism from $E_i$ to $F$ and $\omicron=\bigcup_i\omicron_i$. We define $\xi'= \omicron\circ \xi$ and
$D'=C_F\otimes_{\xi'} F$. In words, $D'$ is obtained from $D$ by
replacing any edge of $D$ from an internal vertex $u$ of $C$ to a vertex $v$ of some $E_i$ by an edge from $u$ to $\omicron_i(v)$, and transfer any root of a vertex $v$ of some $E_j$ to $\omicron_j(v)$, resulting in the drag $D'$, which contains a unique element of the class of $F$, $F$ itself. Note that this step does not create any new class of drags whose bare versions are equimorphic.
The choice of $F$ in its class implies that the maximally shared form of $D$ will be defined up to equimorphism.
\end{enumerate}


\begin{lem}
\label{l:msf}
The maximally shared form $\nf{D}$ of a drag $D$ exists and is unique up to equimorphism.
\end{lem}

\begin{proof}
A sharing step strictly decreases  the number of nontrivial classes of the drag $D$. It follows that it terminates, and therefore maximally shared forms exist. We prove uniqueness by 
showing that any two different sharing steps commute, and then conclude by the Diamond Lemma~\cite{DBLP:journals/jsyml/Hindley69}.

Let $\overline{G},\overline{H}$ be two different minimal classes 
of equimorphic drags of $D$ that can be shared each in turn. 
By minimality of both classes, we can consider the context $C$ of $\overline{G}\oplus \overline{H}$ such that
    $D= C\otimes_\xi (\overline{G}\oplus \overline{H})= C\otimes_\xi (\overline{G}\otimes_\varnothing \overline{H})$. 
    Using associativity and commutativity of product, we get
    $D=(C\otimes_{\xi_{C,\overline{H}}}\overline{H})\otimes_{\xi_{C,\overline{G}}}\overline{G}= (C\otimes_{\xi_{C,\overline{G}}}\overline{G})\otimes_{\xi_{C,\overline{H}}}\overline{H}$, showing that  $(C\otimes_{\xi_{C,\overline{H}}}\overline{H})$ and 
    $(C\otimes_{\xi_{C,\overline{G}}}\overline{G})$ are the contexts in $D$ of $\overline{G}$ and $\overline{H}$, respectively.
    Let now \smash{$E= C\otimes_{\xi_{C,\overline{H}}}\overline{H})\otimes_{\xi_{C,G}} G$} and $F= (C\otimes_{\xi_{C,\overline{G}}}\overline{G})\otimes_{\xi_{C,H}} H$, obtained from $D$ by sharing classes $\overline{G}$ and $\overline{H}$, respectively. Using associativity and commutativity again, we can now share the classes \smash{$\overline{G}$} and \smash{$\overline{H}$} in $E$ and $F$, respectively. We get $E'= (C\otimes_{\xi_{E,H}} H)\otimes_{\xi_{C,G}} G$ and $F'= (C\otimes_{\xi_{C,G}} G)\otimes_{\xi_{C,H}} H$. Using associativity and commutativity again, we get $E'=(C\otimes_{\xi} (G\otimes_\varnothing H)= F'$.
\end{proof}

\begin{exa}
Figure~\ref{f:muld} shows two examples of drags that have the same maximally shared form. For the left drag, the two subdrags reduced to a vertex labeled $a$ are equimorphic. The maximally shared form is obtained in one step. For the right drag, two steps will be needed, as shown on the figure.
\qed\end{exa}

\begin{figure}[!t]
\setlength{\unitlength}{0.7cm} 
\hspace*{8mm}
\begin{picture}(17,4)(0,0)   
\thicklines
\put(-1,0){}
\put(-1,3){%
\begin{picture}(4,3.25)(0,0)
\put(-0.5,-0.4){$\Bigg($}
\put(0.75,0.5){\textcolor{black}{$f$}}
\put(0.95,0.){\textcolor{black}{$\searrow$}}
\put(0.45,0.){\textcolor{black}{$\swarrow$}}
\put(0.05,0.0){$\da$}
\put(0.05,-0.4){\textcolor{black}{$a$}}
\put(1.45,-0.4){\textcolor{black}{$k$}}
\put(1.4,-0.95){$\da$}
\put(1.6,-0.95){$\da$}
\put(1.55,-1.4){$a$}
\put(1.85,-0,4){$\Bigg)\quad\;=$}
\put(2.35,0.3){\vector(0,-1){1}}  
\end{picture}%
}

\put(1.75,3){
\begin{picture}(4,3.25)(0,0)
\put(1.45,0.5){\textcolor{black}{$f$}}
\put(1.45,0.){\textcolor{black}{$\da$}}
\put(1.45,-0.5){\textcolor{black}{$k$}}
\put(1.4,-0.95){$\da$}
\put(1.6,-0.95){$\da$}
\put(1.5,-1.4){$a$}
\put(1,-1.4){$\ra$}
\qbezier(1.4,0.35)(1,-0.4)(1.4,-1)
\put(1.4,-1){\vector(1,-3){0}}
\end{picture}
}

\put(7.5,3){
\begin{picture}(4,3.25)(0,0)
\put(-0.5,-0.4){$\Bigg($}
\put(0.75,0.5){\textcolor{black}{$f$}}
\put(0.95,0.){\textcolor{black}{$\searrow$}}
\put(0.45,0.){\textcolor{black}{$\swarrow$}}
\put(0.05,-0.4){\textcolor{black}{$a$}}
\put(0.5,-0.4){$\longleftarrow$}
\put(1.45,-0.4){\textcolor{black}{$k$}}
\put(1.5,-1){$\da$}
\put(1.45,-1.45){$a$}
\put(1.5,0.95){$\da$}
\put(1.45,0.5){$a$}
\put(1.85,-0,4){$\Bigg)\quad\;=$}
\put(2.35,0.3){\vector(0,-1){1}}  
\end{picture}
}

\put(11.75,3){
\begin{picture}(4,3.25)(0,0)
\put(-0.5,-0.4){$\Bigg($}
\put(0.75,0.5){\textcolor{black}{$f$}}
\put(0.95,0.){\textcolor{black}{$\searrow$}}
\put(0.45,0.){\textcolor{black}{$\swarrow$}}
\put(0.05,-0.4){\textcolor{black}{$a$}}
\put(0.5,-0.4){$\longleftarrow$}
\put(1.45,-0.4){\textcolor{black}{$k$}}
\put(1.4,-1){$\da$}
\put(1.45,-1.45){$a$}
\put(1,-1.4){$\ra$}
\put(1.85,-0,4){$\Bigg)\quad\;=$}
\put(2.35,0.3){\vector(0,-1){1}}  
\end{picture}
}

\put(14.5,3){
\begin{picture}(4,3.25)(0,0)
\put(1.45,0.5){\textcolor{black}{$f$}}
\put(1.45,0.){\textcolor{black}{$\da$}}
\put(1.45,-0.5){\textcolor{black}{$k$}}
\put(1.4,-0.95){$\da$}
\put(1.6,-0.95){$\da$}
\put(1.5,-1.4){$a$}
\put(1,-1.4){$\ra$}
\qbezier(1.4,0.35)(1,-0.4)(1.4,-1)
\put(1.4,-1){\vector(1,-3){0}}
\end{picture}
}
\end{picture}

\caption{Maximally shared form of a drag.}\label{f:muld}
\end{figure}
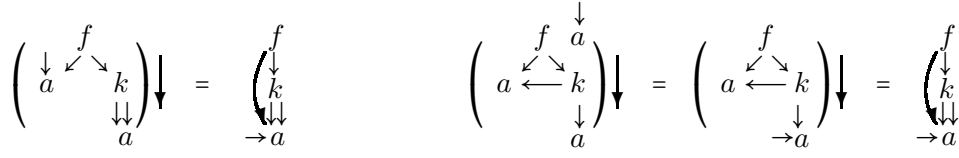

\begin{defi}[Sharing-equivalence]
Two drags $D,D'$ are \emph{sharing-equivalent}, written $D\she D'$, if they have equimorphic maximally shared forms.
\end{defi}

Simple consequences are that equimorphic drags are sharing-equivalent; and that drags that are sharing-equivalent have the same sets of variables and multiset of roots, since both are preserved by computing maximally shared forms.

\begin{lem}
\label{l:mnf}
Given a drag $D$, there exists a morphism $\omicron: D\to \nf{D}$, such that all vertices of $D$ sent to the same vertex by $\omicron$ generate subdrags of $D$ that are sharing-equivalent.
\end{lem}

\begin{proof}
Straightforward, noticing that each class of equimorphic drags in a drag has a representative (the one chosen by sharing) in normal form.
\end{proof}

We now show that sharing-equivalence is closed by the operations on drags that we have defined. First, obviously,

\begin{lem}
\label{l:closuresum}
Sharing-equivalence is closed under parallel composition.
\end{lem}

\begin{defi}[Sharing-equivalence of wires]
\label{d:eqshareext}
 Let $D,D'$ be two drags that are sharing-equivalent. Two well-behaved set of wires $W=\{\wir{s_i}{r_i}\}_i$ of $D$ and $W'=\{\wir{s'_j}{r'_j}\}_j$ of $D'$
are \emph{sharing-equivalent} if 
\begin{itemize}
    \item the sets of variables labeling the sprouts $\{s_i\}_i$ and $\{s'_j\}_j$ are identical;
    \item
    for all sprouts $s_i,s'_j$ sharing the same label, $\subd{D}{r_i}$ and $\subd{D'}{r'_j}$ are sharing-equivalent.
    \end{itemize}
 Two rewriting extensions $\ext(E,\xi)$ of $D$ and $\ext(E',\xi')$ of $D'$ are \emph{sharing-equivalent} if so are $E$ and $E'$, and $\xi$ and $\xi'$.
\end{defi}

The definition of sharing-equivalence for sets of wires makes sense since the same variables label the respective sprouts of sharing-equivalent drags. It makes sense for extensions by Lemma~\ref{l:closuresum}, since $\xi$ and $\xi'$ are well-behaved sets of wires of $D\oplus E$ and $D'\oplus E'$ by definition.
Sharing-equivalence is thus an equivalence on drags, sets of wires, and extensions.

\begin{lem}
\label{l:eqsharewires}
Given two sharing-equivalent well-behaved sets of wires $\xi,\zeta$ of a drag $D$, $\wiring{D}{\xi}$ and $\wiring{D}{\zeta}$ are sharing-equivalent.
\end{lem}

\begin{proof}
By induction on the number of wires in $\xi$. 
If $\xi$ is empty, then $\zeta$ must be empty too since they are sharing-equivalent, and the result holds in that case. 
Otherwise, neither $\xi$ nor $\zeta$ can be empty. Let $\wir{s}{r}\in\xi$ be a maximal wire. 
By the definition of sharing-equivalence for sets of wires, there must exist a wire $\wir{s'}{r'}\in\zeta$ such that $\subd{D}{r}$ and $\subd{D}{r'}$ are sharing-equivalent, implying that neither are sprouts or both are in which case their label is the same. By coherence of a well-behaved set of wires, we can always choose $s=s'$. It follows that $\wir{s}{r'}$ must be maximal in $\zeta$.
The drags $\gwiring{D}{\wir{s}{r}}$ and $\gwiring{D}{\wir{s}{r'}}$ are clearly sharing-equivalent.
Since $\xi\setminus\wir{s}{r}$ and $\zeta\setminus\wir{s}{r'}$ are sharing-equivalent, well-behaved sets of wires, we conclude by the induction hypothesis.
\end{proof}

\begin{lem}
\label{l:nfcom}
Sharing commutes with wiring: $\nf{\wiring{D}{\xi}}=\nf{\wiring{(\nf{D}}{\xi})}$.
\end{lem}

The set of wires $\xi$ for $\nf{D}$ should of course be understood as the restriction of $\xi$ to the sprouts of $D$ which are still vertices of $\nf{D}$.

\begin{proof}
By induction on the size of $\xi$. If $\xi$ is empty, the result is clear. Otherwise, let $\xi  = \zeta\cup\wir{s}{r}$, where $\wir{s}{r}$ is maximal.
By coherence of $\xi$, we can choose $\wir{s}{r}$ such that $s$ is still  a sprout  of $\nf{D}$. By Lemma~\ref{l:mnf}, $r$ is  mapped to a vertex $\omicron(r)$ of $D$ such that $r$ and $\omicron(r)$ generate sharing-equivalent subdrags.
Now, $\wiring{D}{\xi}=\wiring{(\gwiring{D}{\wir{s}{r}})}{\zeta}$, and 
$\wiring{\nf{D}}{\xi}=\wiring{(\gwiring{\nf{D}}{\wir{s}{\omicron(r)}})}{\zeta}$, and by the previous remark, 
$\gwiring{D}{\wir{s}{r}}$ and $\gwiring{\nf{D}}{\wir{s}{r}}$ are sharing-equivalent.
We then conclude by the induction hypothesis. 
\end{proof}

\begin{lem}
\label{l:closurewire}
Sharing-equivalence is closed under wiring.
\end{lem}

\begin{proof}
We are now given two sharing-equivalent drags $D,D'$ and two sharing equivalent, well-behaved sets of wires $\xi,\zeta$ for $D,D'$, respectively. Now,
\begin{align*}
\nf{\wiring{D}{\xi}} &= \nf{(\wiring{\nf{D}}{\xi})} & \mbox{(by Lemma~\ref{l:nfcom})}\\ &=\nf{(\wiring{\nf{D'}}{\xi})}
& \mbox{(because $D$ and $D'$ are sharing-equivalent)}\\
&=\nf{(\wiring{\nf{D'}}{\zeta})} & \mbox{(by Lemma~\ref{l:eqsharewires})}\\ &=\nf{(\wiring{D'}{\zeta})},
\end{align*}
showing that
$\wiring{D}{\xi}$ and $\wiring{D'}{\zeta}$ are sharing-equivalent.
\end{proof}

Using now Lemmas~\ref{l:closuresum} and~\ref{l:closurewire}, we get:

\begin{lem}
\label{l:closureprod}
Given two sharing-equivalent drags $D,D'$, let $\ext(E,\xi)$ and $\ext(E',\xi')$ be two sharing equivalent extensions of $D$ and $D'$, respectively. Then,
$D\otimes_\xi E$ and $D'\otimes_{\xi'} E'$ are sharing-equivalent.
\end{lem}

The following important result summarizes the closure properties of sharing-equivalence:

\begin{thm}
\label{t:sec}
Sharing-equivalence is closed under parallel and cyclic composition.
\end{thm}

\section{Rewriting}
\label{s:rew}
Rewriting is often used as a method to decide congruences, or to describe syntactic transformations, the underlying congruence being implicit. The idea is that a congruence is an equivalence that is closed under composition with respect to extensions. This is the case for drags just like it is for terms, 
composition taking here the place of both context application and substitution.

As usual, rewriting is a precongruence.
Symmetry is eschewed, so as to allow the unidirectional use of rewriting to decide whether two given drags are equivalent in the congruence generated by a given set of drag equations, thereby potentially reducing the nondeterminism involved in proof search.

\subsection{Congruences}
Term congruences are defined as equivalences closed under context application and substitution. Drag congruences should generalize term congruences by being closed under sum and product. The latter will require a precise correspondence between the roots of congruent terms.  Besides, they should also allow for sharing, hence contain sharing equivalence. 

\begin{defi}[Root-map]
\label{d:rm}
 Given drags $D,D'$, a \emph{root-map} from $D$ to $D'$
 is a multi-injective multi-map $\eta :\,\roots{D} \ra \roots{D'}$.
 \end{defi}
As noted in Section \ref{s:drags}, $\eta$ is multi-equijective from its domain multiset to its image multiset.
  When needed, root-maps can be written as lists with repetitions $\{u_1\mapsto v_1, \ldots, u_n\mapsto v_n\}$.

\begin{defi}[Rewriting extension]
\label{d:re}
Given a drag $D$ and a linear drag $C$ disjoint from $D$, $\ext(C,\xi)$ is a \emph{rewriting extension} of $D$ if $\xi_D$ is total and $\xi_C$ surjective.
\end{defi}

Imposing that $\xi_C$ is surjective is not a restriction, since it is always possible to extend $C$ by isolated sprouts mapped by $\xi_C$ to the remaining roots of $D$. Imposing totality is no restriction either since we are indeed interested in congruences on closed drags.

\begin{defi}[Compatible extensions]
\label{d:te}
 Given compatible drags $D,D'$, rewriting extensions $\ext(C,\xi)$ of $D$ and $\ext(C',\xi')$ of $D'$ are \emph{compatible} with  root-map $\eta : D\ra D'$, if
 \begin{enumerate}
 \item $C$ and $C'$ are compatible, equimorphic drags: $C\equiv_{\omicron} C'$;
 \item for all sprouts $t:x$ of $D$ and $t':x$ of $D'$, the subdrags
        $\rest{C}{\xi(t)}$ and $\rest{C'}{\xi'(t')}$ are equimorphic;
  \item \smash{$\forall s\in\Dom{\xi_C} \st\, \xi'_C(\omicron(s))= \eta(\xi_C(s))$}.
\end{enumerate}
 \end{defi}

 Note that any two sprouts of $D\oplus D'$ sharing their variable label must be mapped by $\xi_D\cup\xi'_{D'}$ to vertices generating equimorphic subdrags. Totality is the key to ban dangling edges when rewriting.  It is also important to notice that $\xi_D$ need not being surjective; hence the context $C$ may contain rooted subdrags needed for $\xi'_{D'}$ that remain useless for $\xi_D$. This is in particular the case if there are more sprouts labeled by a given variable $x$ in $D'$ than in $D$.

\begin{defi}
A congruence over the set of drags is an equivalence over compatible drags $U,V$ equipped with a \emph{root-map} $\eta$, written $U\equiv^\eta V$ or simply $U\equiv V$ when the root-map is not needed, satisfying the following closure properties:
\begin{enumerate}
\item 
Under sharing equivalence:
Let $U\equiv^\eta V$ and $V\she V'$, that is $V\equiv_\omicron V'$.
Then, $U\equiv^{\omicron\circ \eta} V'$.
\item 
Under sum:
Let $U\equiv^\eta V$, $W$ compatible with both $U$ and $V$, and $\mu$ a root-map from $W$ to $W$ whose domain and image are disjoint from $\roots{U,V}$, then $U\oplus W\equiv^{\eta\cup\mu} V\oplus W$.
\item 
Under product:
Let $U\equiv^\eta V$, 
$\ext(C,\xi)$ and $\ext(C',\xi')$ extensions of $U,V$, respectively, compatible with the root-map $\eta$, and $\mu$ a restriction of the equimorphism from $C$ to $C'$ whose domain and image are disjoint from $\Ima{\xi_U}$ and $\Ima{\xi'_{U'}}$, respectively. 
Then  $U\otimes_\xi C \equiv^\mu V\otimes_{\xi'} C'$.
\end{enumerate}
\end{defi}

Note that we use the same symbol $\equiv$ for both equimorphisms and congruences, the former with a lower index and the latter with an upper index.

By wiring different sprouts, possibly but not necessarily labeled the same, to the same vertex, case (3) incorporates some limited amount of sharing. Only limited, since internal vertices generating equimorphic subgraph cannot become shared when computing products.

\begin{exa}
Let drags $D=g^{[1]}(f_1^{[1]}(x_1,x_2))\equiv^{\{g\mapsto k, f_1\mapsto f_2\}} k^{[1]}(f_2^{[1]}(x_3,x_3))=D'$, and equimorphic context drags $C=h_1^{[2]}(z_1)\oplus z_1'^{[1]}$ and $C'=h_2^{[2]}(z_2)\oplus z_2'^{[1]}$, where $x_1, x_2, x_3$ are different vertices labeled $x$, while $z_1, z_2$ are labeled $z$, and $z_1', z_2'$ are labeled $z'$.
Then,
\begin{align*}
&[y=f(y',y'), y'=h(y)]\,g(y) = 
g^{[1]}(f_1^{[1]}(x_1,x_2))\otimes_{\{\wir{x_1}{h_1},\wir{x_2}{h_1},\wir{z_1}{f_1},\wir{z_2}{k}\}} C\\
&\equiv^\varnothing\\
&\smash{[y=f(y',y'), y'=h(y)]\,k(y) =
k^{[1]}(f_2^{[1]}(x_3,x_3)) \otimes_{\{\wir{x_3}{h_2},\wir{z_3}{f_2},\wir{z_4}{g}\}} C'}.
\end{align*}
Note that $x_1$ and $x_2$ are both mapped to $h_1$, while $x_3$ is mapped to $h_2$, $h_1$ and $h_2$ generating equimorphic subdrags of the result. 
Note also that $z_1$ is mapped to $f_1$ while $z_2$ is mapped to $f_2$ as dictated by the root-map. Similarly, $z'_1$ and $z'_2$ are mapped to  $g$ and $k$, respectively.
\end{exa}

We are of course interested in congruences generated by a set of equations:

\begin{defi}
\label{l:eq}
An equation is a pair of compatible drags $U,V$ together with a root-map $\eta$, denoted $U=^\eta V$. The equality over drags generated by a set $\cE$ of equations, denoted $=_\cE$, is the least congruence that contains all pairs in $\cE$. 
\end{defi}

\begin{exa}
Assuming $f$ is binary, we take here for $\cE$ the associativity equation with its trivial root-map: $f_1\mapsto f_3: f_1^{[1]}(f_2(x,y),z)=f_3^{[1]}(x, f_4(y,z))$. Then, closure under product with the compatible extensions $\ext(x'^{[1]}\oplus b^{[1]} \oplus c^{[1]}, \{\wir{x}{x'}, \wir{y}{b},\wir{z}{c},\wir{x'}{f_1}\})$ and $\ext(x''^{[1]}\oplus b^{[1]} \oplus c^{[1]}, \{\wir{x}{x''}, \wir{y}{b},\wir{z}{c},\wir{x''}{f_3}\})$ renders equal the two ground drags
$[z=f(f(z,b),c)]z$ and  $[z'=f(z',f(b,c))]z'$.
\end{exa}

\subsection{Rewrite rules}
A rewrite rule serves to replace some drag pattern $L$ by some other drag pattern $R$ in a given drag $D$ that contains $L$ in a context defined by an extension $\ext(E,\xi)$ of $L$. That is, $D=E\otimes_\xi L$.

First, it is important to ensure that all roots and sprouts of $L$ disappear in this composition, yielding the previously encountered notion of a rewriting extension.
Next, it is equally important to ensure that replacing $L$ by $R$ is possible; this is what root-maps are for. Finally, the extensions used for the left-hand and right-hand sides of a rewrite rule must be compatible, so that sprouts with identical labels are mapped to vertices that generate equimorphic subdrags.

\begin{defi}[Patterns]
A drag all of whose vertices are accessible is called a  \emph{right-pattern}. 
It is called a \emph{left-pattern}, or simply \emph{pattern}, if also all of its sprouts have predecessors.
\end{defi}

\begin{defi}[Rewrite rules]\label{d:rules}
A \emph{drag rewrite rule} is a triple written $\eta: L\ra R$ (alternatively, $L\ra_\eta R$) made of two compatible drags $L$ and $R$, such that $L$ is a pattern, $R$ is a right pattern, and $\eta$ is a multi-equijective root-map from $\roots{L}$ to $\roots{R}$ (hence $|\roots{R}|=|\roots{L}|$).

A rule is \emph{stringent} if $\Var{L}\subseteq \Var{R}$.

A \emph{renaming} of a rule $\eta:L\ra R$ is a rule $\eta':L'\ra R'$ such that (i) $L'\oplus R'$ is a renaming of $L\oplus R$, that is,
$L\oplus R \simeq_\omicron^\sigma L'\oplus R'$,  and (ii)
$\forall r\in\roots{L} \st \eta'(\sigma(r)) = \sigma(\eta(r))$.
\end{defi}

Since $L$ and $R$ are compatible drags, their sum $L\oplus R$ is defined and can indeed be seen as the rule whose multiset of roots is partitioned into a left multiset $\roots{L}$ and a right multiset $\roots{R}$, so that $\eta$ is a multi-equijective map from the first to the second.

The case of term rewriting rules is then quite simple: since $L$ and $R$ have a single root each, there is a unique possible map $\eta$. In~\cite{DBLP:journals/tcs/DershowitzJ19}, $L$ and $R$ have ordered lists of roots of the same length, making $\eta$ unique again. In both these cases, there is no need for $\eta$ to be explicit. Having multisets of roots forces us to specify $\eta$.

\subsection{Rewriting relations}

We first consider "relational" rewriting. 
Given drags $D,D'$, rewriting $D$ to $D'$ using rule $\eta:L\ra R$ sharing no variable with $D,D'$ involves the following steps:
\begin{enumerate}
    \item 
    \textit{Match}: Find a variable-preserving renaming $\eta':L'\ra R$ of $\eta:L\ra R$, and rewriting extensions $\ext(E,\xi)$ of $L'$ 
    and $\ext(E',\xi')$ of $R'$, such that $D=E\otimes_\xi L'$ and $D'=E'\otimes_{\xi'} R'$.
    \item    
    \textit{Verify}: Establish that these two extensions are \emph{compatible}.
    \end{enumerate}

In the following, we will usually confuse the rule $\eta:L\ra R$ with its variable preserving renaming $\eta':L'\ra R'$.

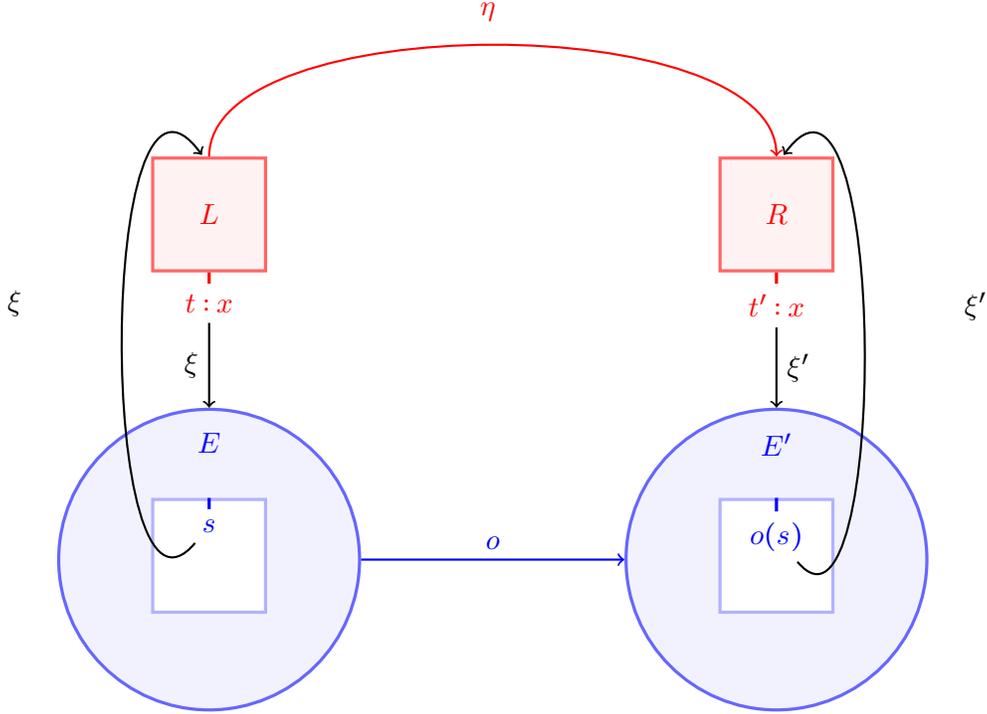
\begin{figure}
\hspace*{-6mm}
\begin{tikzpicture}[
squarenode/.style={rectangle, draw=red!60, fill=red!5, very thick, minimum size=15mm},
squarenodewhite/.style={rectangle, draw=blue!30, fill=white, very thick, minimum size=15mm},
roundnode/.style={circle, draw=blue!60, fill=blue!5, very thick, minimum size=40mm},
namenode/.style={},
x=0.5cm,y=0.6cm  
]

\node[squarenode]    (vertexL)                             {\red{$L$}};
\node[squarenode]    (vertexR)       [right=6cm of vertexL] {\red{$R$}};
\node[roundnode]     (vertexE)       [below=18mm of vertexL]{};
\node[namenode]      (vertexEname)   [below=20mm of vertexL]{\blue{$E$}};
\node[squarenodewhite]   (vertexLP)      [below=3cm of vertexL]{};
\node[roundnode]     (vertexEP)      [below=18mm of vertexR]{};
\node[namenode]      (vertexEPname)  [below=20mm of vertexR]{\blue{$E'$}};
\node[squarenodewhite]   (vertexRP)      [below=3cm of vertexR]{};

\node (s)          [above=-6mm of vertexLP, thick, blue] {$s$};
\node (os)         [above=-8.5mm of vertexRP, thick, blue] {$\omicron(s)$};
\node (tx)          [below=1.5mm of vertexL, thick, red] {$t:x$};
\node (tpx)         [below=1.5mm of vertexR, thick, red] {$t':x$};
\node (Ltop)          [above=-2.5mm of vertexL] {};
\node (Rtop)          [above=-2.5mm of vertexR] {};

\node (ca)     [left=10mm of s]{};
\node (cb)     [left=10mm of vertexE]{};
\node (cc)     [left=10mm of vertexL]{};

\node (middle)     [right=28mm of vertexL]{};
\node (eta)        [above=23mm of middle]  {\red{$\eta$}};
\node (xiE)        [left=19mm of tx]  {$\xi$};
\node (xiEp)       [right=18.5mm of tpx]  {$\xi'$};
\draw[->, thick, red] (vertexL) .. controls + (up:2.75cm) and + (up:2.75cm) ..  (vertexR);

\draw[->, thick, blue] (vertexE) to node [above] {\blue{$\omicron$}} (vertexEP);

\draw[-,very thick, red] (tx) to (vertexL);
\draw[-,very thick, red] (tpx) to (vertexR);
\draw[-,very thick, blue] (s) to (vertexLP);
\draw[-,very thick, blue] (os) to (vertexRP);

\draw[->, thick] (tx) to node [left] {$\xi$}  (vertexE);
\draw[->, thick] (tpx) to node [right] {$\xi'$}  (vertexEP);
\draw[->, thick] (s) .. controls + (-3,-3) and + (-3,+3.5) .. (Ltop);
\draw[->, thick] (os) .. controls + (+3,-3) and + (+3,+3.5) ..  (Rtop);

\end{tikzpicture}
\caption{Compatible rewrite extension of a rewrite rule.}\label{f:rr}
\end{figure}

The rewriting switchboards map sprouts of $E,E'$ to roots of $L',R'$, respectively, and these mappings must fit with $\eta'$, hence condition (2). Note that taking $E$ and $E'$ isomorphic would suffice: imposing that their sprouts are labeled by the same variables is possible because we distinguish both switchboards $\xi$ for the left-hand side and $\xi'$ for the right-hand side. Condition (3) expresses the property that the restrictions of $\xi$ and $\xi'$, to the sprouts of $L'$ and $R'$, respectively, could be made into a single well-behaved set of wires, which will become important later.

One may worry about the complexity of deciding a rewrite step. The matching step 1 in the "list of roots" model has a worst case complexity which is bounded by the product of the sizes of $D\oplus D'$ and $L'\oplus R'$ \cite{JO22a}. It should not be very different for the "multiset of roots" model since the main drag invariant checked for that purpose is indegree preservation, for which only the number of roots matters. Checking equimorphism at step (3) has indeed a worst case complexity bounded by the sum of the sizes of $D$ and $D'$, since the two vertices $\xi(t)$ and $\xi'(t')$ are given, both drags are directed, and the successors of both vertices are lists of the same length. So, the overall complexity is determined by that of the matching process, which has a quadratic upper bound in the "list model".

We can now define \emph{relational rewriting}:

\begin{defi}[Rewriting]
\label{d:rewrel}
A drag $D$ is in a \emph{rewriting relation} with drag $D'$---using the (renamed) rewrite rule $\eta:L\ra R$ sharing no variable with $D$,
if there exist two compatible rewriting extensions $\ext(E,\xi)$ of $L$ and $\ext(E',\xi')$ of $R$ such that $D=E\otimes_\xi L$ and $D'=E'\otimes_{\xi'} R$.
\end{defi}

Since the extension drags $E$ and $E'$ must be equimorphic, it is tempting to take them to be identical. This is of course impossible in case $D$ and $D'$ do not share vertices, in which case only the sprouts of $E$ and $E'$ can be shared, which simplifies condition (2) already to $\xi'(s)=\eta(\xi(s))$. Usually, however, only $D$ is given, and $D'$ is defined by the rewriting process, in which case the same extension $E$ can do for both, and rewriting can be defined as a computation mechanism.

\begin{defi}[Strong compatibility]
Given a rule $\eta:L\ra R$, two rewriting extensions $\ext(E,\xi)$ and $\ext(E',\xi')$ of $L$ and $R$, respectively, are \emph{strongly compatible} if
\begin{enumerate}
        \item $E=E'$;
        \item for all sprout $s$ of $E$, $\xi'(s)=\eta(\xi(s))$;
    \item for all sprouts $s:x$ of $L$ and $t:x$ of $R$, the subdrags
    $\rest{E}{\xi(s)}$ and $\rest{E}{\xi'(t)}$ are equimorphic.
 \end{enumerate}
\end{defi}

Using the same drag $E$ for both extensions does not mean that all of $E$ is used for $L$ and all of $E$ is used for $E'$.
It may be that there are more sprouts labeled $x$ (for some $x$) in $R$ than in $L$; mapping these sprouts to different subdrags of $E$ requires more of them for $R$ than for $L$. 
So, the same rewrite could be achieved with two more economic drags $E$ and $E'$, that would not be equimorphic anymore, but which enjoy a more complex relationship.

\begin{defi}[Functional rewriting]
\label{d:rew}
A drag $D$ \emph{rewrites} to a drag $D'$ using the stringent rewrite rule $\eta:L\ra R$ sharing no variable with $D$,
if there exist two strongly compatible rewriting extensions $\ext(E,\xi)$ of $L$ and $\ext(E',\xi')$ of $R$ such that
$D=E\otimes_\xi L$ and $D'=E'\otimes_{\xi'} R$.

We say that drag $D$ \emph{rewrites} to drag $D'$ with rewrite system $\cR$, all of whose rules are stringent, denoted  $D\dlrps{}{\cR} D'$, if $D$ rewrites to $D'$ using some rule $\eta\in\cR$.
\end{defi}

\begin{exa}
\label{ex:rew}
Consider the two rewriting examples of Figure~\ref{f:rew}.
The input drags to be rewritten and the resulting output drags are in black. 
The rule $h(x,x)\ra k(x,x)$ is in red. 
Its various sprouts, all labeled $x$, are not shared; hence, we can use our naming conventions.
The extension drags are in blue, and the switchboards in black.

In the upper example, the right-hand side switchboard $\xi'$ coincides roughly with the left-hand side switchboard $\xi$. Note that the left- and right-hand sides of the rule have a single root, the vertices $h$ and $k$, respectively. Note also that these switchboards do satisfy strong compatibility.

In the lower example, the switchboard $\xi'$ differs from $\xi$ in that both left-hand side $x$'s are mapped to $a_2$ (there is no other choice) while the two right-hand sides $x$'s are mapped to $a_1$ and $a_2$, respectively, which yields a quite different result. Note that $a_1$ and $a_2$ generate equimorphic subdrags reduced to a single vertex labeled $a$ having two roots. Here, the right-hand side switchboard does not force sharing, that's why we can map $x_3$ and $x_4$ to different vertices.

Using the rule $h(x,x)\ra k(x,x)$ in which $x$ is shared in the left-hand side would yield exactly the same result with the same switchboards. In this example, the switchboard forces sharing on the left-hand side, even if the rule does not.
\qed\end{exa}

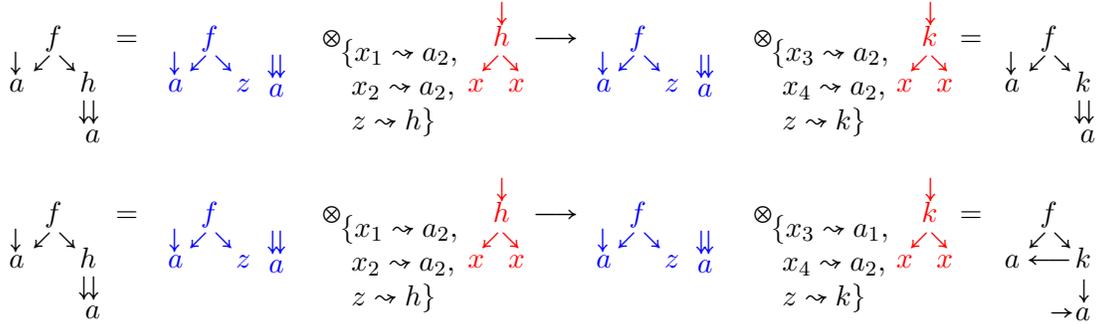
\begin{figure}[!t]
\setlength{\unitlength}{0.67cm} 
\hspace*{5mm}
\begin{picture}(17,9)(0,0)   
\thicklines

\put(0,2){
\put(-2.9,5){
\begin{picture}(4,3.25)(0,0)
\put(0.75,0.5){\textcolor{black}{$f$}}
\put(0.95,0.){\textcolor{black}{$\searrow$}}
\put(0.45,0.){\textcolor{black}{$\swarrow$}}
\put(0.05,0.0){$\da$}
\put(0.05,-0.4){\textcolor{black}{$a$}}
\put(1.45,-0.4){\textcolor{black}{$h$}}
\put(1.4,-0.95){$\da$}
\put(1.6,-0.95){$\da$}
\put(1.55,-1.4){$a$}
\end{picture}
}

\put(-0.55,5.5){=}

\put(0.2,5){
\begin{picture}(4,3.5)(0,0)
\put(0.75,0.5){\textcolor{blue}{$f$}}
\put(0.95,0.){\textcolor{blue}{$\searrow$}}
\put(0.45,0.){\textcolor{blue}{$\swarrow$}}
\put(0.1,0.0){\blue{$\da$}}
\put(0.1,-0.4){\textcolor{blue}{$a$}}
\put(1.45,-0.4){\textcolor{blue}{$z$}}

\put(2.05,-0.05){\blue{$\da$}}
\put(2.2,-0.05){\blue{$\da$}}
\put(2.1,-0.5){\blue{$a$}}
\end{picture}
}

\put(3.35,5.5){
$\otimes_{\hspace{-2mm}
\begin{array}{l}
\{\wir{x_1}{a_2},\\
\;\,\wir{x_2}{a_2},\\
\;\,\wir{z}{h}\}
\end{array}}$
}

\put(1.45,5){
\begin{picture}(4,3.5)(0,0)
\put(5.35,1){\textcolor{red}{$\da$}}
\put(5.3,0.5){\textcolor{red}{$h$}}
\put(5.,0.){\textcolor{red}{$\swarrow$}}
\put(4.8,-0.4){\textcolor{red}{$x$}}
\put(5.4,0.){\textcolor{red}{$\searrow$}}
\put(5.6,-0.4){\textcolor{red}{$x$}}
\end{picture}
}

\put(7.75,5.5){$\dlrps{}{}$}

\put(8.7,5){
\begin{picture}(4,3.5)(0,0)
\put(0.75,0.5){\textcolor{blue}{$f$}}
\put(0.95,0.){\textcolor{blue}{$\searrow$}}
\put(0.45,0.){\textcolor{blue}{$\swarrow$}}
\put(0.1,0.0){\blue{$\da$}}
\put(0.1,-0.4){\textcolor{blue}{$a$}}
\put(1.45,-0.4){\textcolor{blue}{$z$}}

\put(2.05,-0.05){\blue{$\da$}}
\put(2.2,-0.05){\blue{$\da$}}
\put(2.1,-0.5){\blue{$a$}}
\end{picture}
}

\put(11.9,5.5){
$\otimes_{\hspace{-2mm}
\begin{array}{l}
\{\wir{x_3}{a_2},\\
\;\,\wir{x_4}{a_2},\\
\;\,\wir{z}{k}\}
\end{array}}$
}

\put(9.95,5){
\begin{picture}(4,3.5)(0,0)
\put(5.35,1){\textcolor{red}{$\da$}}
\put(5.3,0.5){\textcolor{red}{$k$}}
\put(5.,0.){\textcolor{red}{$\swarrow$}}
\put(4.8,-0.4){\textcolor{red}{$x$}}
\put(5.4,0.){\textcolor{red}{$\searrow$}}
\put(5.6,-0.4){\textcolor{red}{$x$}}
\end{picture}
}

\put(16.2,5.5){=}

\put(16.85,5){
\begin{picture}(4,3.5)(0,0)
\put(0.75,0.5){\textcolor{black}{$f$}}
\put(0.95,0.){\textcolor{black}{$\searrow$}}
\put(0.45,0.){\textcolor{black}{$\swarrow$}}
\put(0.05,0.0){$\da$}
\put(0.05,-0.4){\textcolor{black}{$a$}}
\put(1.45,-0.4){\textcolor{black}{$k$}}
\put(1.4,-0.95){$\da$}
\put(1.6,-0.95){$\da$}
\put(1.55,-1.4){$a$}
\end{picture}
}
}

\put(0,-1.5){
\put(-2.9,5){
\begin{picture}(4,3.25)(0,0)
\put(0.75,0.5){\textcolor{black}{$f$}}
\put(0.95,0.){\textcolor{black}{$\searrow$}}
\put(0.45,0.){\textcolor{black}{$\swarrow$}}
\put(0.05,0.0){$\da$}
\put(0.05,-0.4){\textcolor{black}{$a$}}
\put(1.45,-0.4){\textcolor{black}{$h$}}
\put(1.4,-0.95){$\da$}
\put(1.6,-0.95){$\da$}
\put(1.55,-1.4){$a$}
\end{picture}
}

\put(-0.55,5.5){=}

\put(0.2,5){
\begin{picture}(4,3.5)(0,0)
\put(0.75,0.5){\textcolor{blue}{$f$}}
\put(0.95,0.){\textcolor{blue}{$\searrow$}}
\put(0.45,0.){\textcolor{blue}{$\swarrow$}}
\put(0.1,0.0){\blue{$\da$}}
\put(0.1,-0.4){\textcolor{blue}{$a$}}
\put(1.45,-0.4){\textcolor{blue}{$z$}}

\put(2.05,-0.05){\blue{$\da$}}
\put(2.2,-0.05){\blue{$\da$}}
\put(2.1,-0.5){\blue{$a$}}
\end{picture}
}

\put(3.35,5.5){
$\otimes_{\hspace{-2mm}
\begin{array}{l}
\{\wir{x_1}{a_2},\\
\;\,\wir{x_2}{a_2},\\
\;\,\wir{z}{h}\}
\end{array}}$
}

\put(1.45,5){
\begin{picture}(4,3.5)(0,0)
\put(5.35,1){\textcolor{red}{$\da$}}
\put(5.3,0.5){\textcolor{red}{$h$}}
\put(5.,0.){\textcolor{red}{$\swarrow$}}
\put(4.8,-0.4){\textcolor{red}{$x$}}
\put(5.4,0.){\textcolor{red}{$\searrow$}}
\put(5.6,-0.4){\textcolor{red}{$x$}}
\end{picture}
}

\put(7.75,5.5){$\dlrps{}{}$}

\put(8.7,5){
\begin{picture}(4,3.5)(0,0)
\put(0.75,0.5){\textcolor{blue}{$f$}}
\put(0.95,0.){\textcolor{blue}{$\searrow$}}
\put(0.45,0.){\textcolor{blue}{$\swarrow$}}
\put(0.1,0.0){\blue{$\da$}}
\put(0.1,-0.4){\textcolor{blue}{$a$}}
\put(1.45,-0.4){\textcolor{blue}{$z$}}

\put(2.05,-0.05){\blue{$\da$}}
\put(2.2,-0.05){\blue{$\da$}}
\put(2.1,-0.5){\blue{$a$}}
\end{picture}
}

\put(11.9,5.5){
$\otimes_{\hspace{-2mm}
\begin{array}{l}
\{\wir{x_3}{a_1},\\
\;\,\wir{x_4}{a_2},\\
\;\,\wir{z}{k}\}
\end{array}}$
}

\put(9.95,5){
\begin{picture}(4,3.5)(0,0)
\put(5.35,1){\textcolor{red}{$\da$}}
\put(5.3,0.5){\textcolor{red}{$k$}}
\put(5.,0.){\textcolor{red}{$\swarrow$}}
\put(4.8,-0.4){\textcolor{red}{$x$}}
\put(5.4,0.){\textcolor{red}{$\searrow$}}
\put(5.6,-0.4){\textcolor{red}{$x$}}
\end{picture}
}

\put(16.2,5.5){=}

\put(16.85,5){
\begin{picture}(4,3.5)(0,0)
\put(0.75,0.5){\textcolor{black}{$f$}}
\put(0.95,0.){\textcolor{black}{$\searrow$}}
\put(0.45,0.){\textcolor{black}{$\swarrow$}}
\put(0.05,-0.4){\textcolor{black}{$a$}}
\put(0.5,-0.4){$\longleftarrow$}
\put(1.45,-0.4){\textcolor{black}{$k$}}
\put(1.5,-1){$\da$}
\put(1.45,-1.45){$a$}
\put(0.95,-1.45){$\ra$}

\end{picture}
}
}

\end{picture}

\vspace{-1cm}
\caption{Rewriting example.}\label{f:rew}
\end{figure}

So, the determination of a switchboard obeys precise rules, but leaves also some room for choosing the right-hand side switchboard among sometimes several possibilities. As a consequence, the result of a rewrite step is not entirely determined by matching, as it seems to be the case for terms.  
The very same definition is indeed used for rewriting a term with a term rewrite rule $L \ra R$: the two contexts $E$ and $E'$ are identical, condition (2) is implicit, and condition (3) is usually ensured by taking for each occurrence of $x$ in $R$ a disjoint copy of $\rest{E}{\xi(s)}$ for $\rest{E}{\xi'(t)}$, at least when $R$ is non-linear (otherwise it is possible to use $\rest{E}{\xi(s)}$ for $\rest{E}{\xi'(t)}$). So, uniqueness of the result of a term rewriting step is not true in general, the result being unique up to equimorphism only. This is hidden by the assumption that terms are defined up to equimorphism, an assumption which pops up when terms are considered as particular drags.

In the general case of drag rewriting, however, the two resulting drags obtained in Example~\ref{ex:rew} are not even equimorphic: they are sharing-equivalent.

Given now a rule $L\ra R$ such that $\ext(E,\xi)$ and $\ext(E,\zeta)$ are two equimorphic extensions for $R$, the result of rewriting with that rule and these extensions will not depend upon which extension is used, up to sharing-equivalence:

\begin{thm}
\label{t:unirew}
Let $D\dlrps{}{L\ra R} G$ and $D\dlrps{}{L\ra R} G'$, using compatible rewriting extensions. Then, $G$ and $G'$ are sharing-equivalent.
\end{thm}

\begin{proof}
Since compatible rewriting extensions are sharing-equivalent, $R$ and $R'$ are sharing-equivalent by Theorem \ref{t:sec}.
\end{proof}

This result applies to rewriting as a computing device, but also to rewriting as a relation. Note that we have not required that variables of the right-hand side $R $ of the rule $L\ra R$ all occur in $L$. The choice of $\zeta(y)$, if $y$ is such a variable, is given by matching $G$ with respect to $R$, when using rewriting as a relation. When using rewriting as a computing device, every choice of $\zeta(y)$ will give a specific $G$, but any two strongly compatible choices of $\zeta(y)$ will give sharing-equivalent $G$'s.

In practice, the choice of a particular extension for the right-hand side of a rewrite rule is important. What should be pointed out here is that relational rewriting, and to some extent functional rewriting as well, leave us entire liberty of making the choice one likes. In practice, this choice can be expressed via a strategy being a parameter of the rewriting step. For example, one could choose a maximal sharing strategy, as is often the case with term rewriting implementations (which become then dag implementations.)

\begin{exa}
\label{ex:rone}
Figure~\ref{f:gcrew} illustrates rewriting with a rule whose left-hand side originates from the example of wiring presented in Figure~\ref{f:gc}, and right-hand side is just a variable with 3 roots. 
Leftmost is the input drag, and rightmost is the result. 
Both are in black.
The rule is written in red, the context in blue, the switchboard in black.
The reader is invited to verify the result of the rewrite step by actually doing the composition calculation.
\qed\end{exa}

\begin{figure}[!t]
\setlength{\unitlength}{0.7cm} 
\hspace*{12mm}
\begin{picture}(17,7)(0,0)   
\thicklines

\put(-2,5){
\begin{picture}(4,3.5)(0,0)
\put(0.1,1){\textcolor{black}{$\da$}}
\put(0.05,0.5){\textcolor{black}{$f$}}
\put(0.2,0){\textcolor{black}{$\searrow$}}

\put(1.5,1.){\textcolor{black}{$\da$}}
\put(1.5,0.55){\textcolor{black}{$g$}}
\put(1.1,0){\textcolor{black}{$\swarrow$}}

\put(0.75,-0.4){\textcolor{black}{$h$}}
\put(0.65,-0.85){\textcolor{black}{$\circlearrowleft$}}
\put(0.65,-0.85){\textcolor{black}{$\circlearrowleft$}}
\put(0.2,-0.4){\textcolor{black}{$\ra$}}
\end{picture}
}

\put(0.6,5.5){=}

\put(1.3,5){
\begin{picture}(4,3.5)(0,0)
\put(0.05,1){\textcolor{blue}{$\da$}}
\put(0.,0.5){\textcolor{blue}{$f$}}
\put(0.15,0){\textcolor{blue}{$\searrow$}}

\put(1.75,1.4){\textcolor{blue}{\scriptsize 4}}
\put(1.75,1){\textcolor{blue}{$\da$}}
\put(1.7,0.5){\textcolor{blue}{$h$}}
\put(1.35,0){\textcolor{blue}{$\swarrow$}}

\put(0.8,-0.3){\textcolor{blue}{$x$}}
\put(0.8,0.1){\blue{$\da$}}
\end{picture}
}

\put(3.75,5.5){
$\otimes_{\hspace{-2mm}
\begin{array}{l}
\{\wir{x}{y},\\
\;\,\wir{y}{h}\}
\end{array}}$
}

\put(1.6,5){
\begin{picture}(4,3.5)(0,0)
\put(5.35,1){\textcolor{red}{$\da$}}
\put(5.3,0.5){\textcolor{red}{$g$}}
\put(5.3,0.){\textcolor{red}{$\da$}}
\put(5.3,-0.4){\textcolor{red}{$y$}}
\put(4.7,0.25){\textcolor{red}{\scriptsize 2}}
\put(4.7,-0.1){\textcolor{red}{$\searrow$}}
\end{picture}
}

\put(7.85,5.5){$\dlrps{}{}$}

\put(9,5){
\begin{picture}(4,3.5)(0,0)
\put(0.05,1){\textcolor{blue}{$\da$}}
\put(0.,0.5){\textcolor{blue}{$f$}}
\put(0.15,0){\textcolor{blue}{$\searrow$}}

\put(1.75,1.4){\textcolor{blue}{\scriptsize 4}}
\put(1.75,1){\textcolor{blue}{$\da$}}
\put(1.7,0.5){\textcolor{blue}{$h$}}
\put(1.35,0){\textcolor{blue}{$\swarrow$}}

\put(0.8,-0.3){\textcolor{blue}{$x$}}
\put(0.8,0.1){\blue{$\da$}}
\end{picture}
}

\put(11.4,5.5){
$\otimes_{\hspace{-2mm}
\begin{array}{l}
\{\wir{x}{y},\\
\;\,\wir{y}{h}\}
\end{array}}$
}

\put(8.5,5){
\begin{picture}(4,3.5)(0,0)
\put(5.3,1.3){\textcolor{black}{\scriptsize 3}}
\put(5.3,0.9){\textcolor{red}{$\da$}}
\put(5.3,0.5){\textcolor{red}{$y$}}
\end{picture}
}

\put(14.9,5.5){=}

\put(15.5,5){
\begin{picture}(4,3.5)(0,0)
\put(0.8,1){\textcolor{black}{$\da$}}
\put(0.75,0.5){\textcolor{black}{$f$}}
\put(0.75,0){\textcolor{black}{$\da$}}

\put(-0.1,-0.55){\textcolor{red}{\scriptsize 2}}
\put(0.2,-0.55){\textcolor{black}{$\ra$}}
\put(0.75,-0.55){\textcolor{black}{$h$}}
\put(0.65,-1){\textcolor{black}{$\circlearrowleft$}}
\put(0.65,-1){\textcolor{black}{$\circlearrowleft$}}
\end{picture}
}

\end{picture}

\vspace{-2.5cm}
\caption{Rewriting with the red rule.}\label{f:gcrew}
\end{figure}
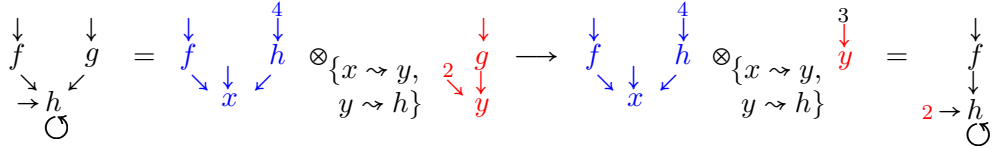

We now come back to the motivating example given in Section \ref{s:example}.

\begin{exa}
\label{ex:share}
In the fully shared subdrag scenario, the left- and right-hand sides are \emph{both} parts of the same drag, their sum. For the first rule of the example of Section~\ref{s:example},
we now have the following combined drag:
\[
\psset{unit=4mm}
\begin{tabular}{c@{\qquad}c}
\begin{pspicture}(-0.5,-2.5)(7,6)
\cnodeput*[fillstyle=solid,fillcolor=cyan]{0}(3,2.4){Aa}{$a$}
\cnodeput[]{0}(4.5,4.5){BFx}{$x$}
\cnodeput*{0}(7,3.4){CF_}{L}
\ncline[linecolor=black,linestyle=dotted]{|->}{CF_}{Aa}
\cnodeput[fillstyle=solid,linestyle=dashed,fillcolor=red]{0}(-2,-1){Ac}{$c$}
\cnodeput[fillstyle=solid,fillcolor=lightgray,linestyle=dashed,linecolor=black,dash=6pt 2pt]{0}(1.,0.5){Ao}{$o$}
\cnodeput*{0}(6,0.){CF_}{R}
\ncline[linecolor=black,linestyle=dashed]{|->}{CF_}{Ao}
\nccurve[linecolor=cadmiumgreen,linestyle=dashed,angleA=90,angleB=180]{->}{Ao}{BFx}
\ncline[linecolor=cadmiumgreen,linestyle=dashed]{<->}{Ac}{Ao}
\ncline[linecolor=red,linestyle=dotted]{->}{Aa}{BFx}
\end{pspicture}
\end{tabular}
\]
We are using dotted arrows for edges that are in the left side of the rule only, and dashed arrows when it's only on the right;
were an edge in both, we would leave it solid (like some of the roots in the second rule below).
Left roots are noted by an L; right ones by R.
Internal vertices are only on one side are likewise dotted or dashed.
So, this rule adds gray $o$ and red $c$ for each blue $a$, erasing the latter. 

The second rule looks now like this:
\[
\psset{unit=4mm}
\begin{tabular}{c@{\qquad}c}
\begin{pspicture}(2,-1)(5,9)
\cnodeput*{0}(2.4,5.9){L1}{L$_1$}
\cnodeput*{0}(4.5,2.9){L2}{L$_2$}
\cnodeput*{0}(8,7.7){LR}{L,R}
\cnodeput*[]{0}(8.5,0.5){R2}{R$_2$}
\cnodeput*[]{0}(3.5,0.5){R1}{R$_1$}
\cnodeput{0}(4.5,6.0){x}{$x$}
\cnodeput{0}(0,0){y}{$y$}
\cnodeput[fillstyle=solid,fillcolor=red]{0}(4.5,8){c}{$c$}
\cnodeput[linestyle=dotted,fillstyle=solid,fillcolor=lightgray]{0}(1.5,3.5){o}{$o$}
\cnodeput[linestyle=dashed,fillstyle=solid,fillcolor=yellow]{0}(6.5,1.5){e}{$e$}
\ncline[linecolor=black,linestyle=dotted]{|->}{L1}{o}
\ncline[linecolor=black,linestyle=dotted]{|->}{L2}{o}
\nccurve[linecolor=black,linestyle=dashed,angleA=220,angleB=270]{|->}{R2}{e}
\nccurve[linecolor=black,linestyle=dashed,angleA=40,angleB=200]{|->}{R1}{e}
\ncline[linecolor=black,linestyle=solid]{|->}{LR}{c}
\ncline[linecolor=cadmiumgreen,linestyle=dotted]{->}{o}{x}
\ncline[linecolor=cadmiumgreen,linestyle=dotted]{->}{o}{y}
\nccurve[linecolor=blue,linestyle=dashed,linewidth=1pt,angleA=180,angleB=90]{->}{c}{y}
\end{pspicture}
\end{tabular}
\]
The red $c$ vertex is shared; a yellow $e$ is introduced in place of the deleted gray $o$, with incoming edges L$_1$ and L$_2$ redirected to roots R$_1$ and R$_2$, respectively.
\end{exa}

We show with a last example that the drag format allows more sharing that available in usual term rewriting implementations.

\begin{exa}
\label{ex:rr}
Consider the rule described by the drag $f_1^{[1]}(f_2^{[1]}(x))$ with two roots, the left-hand side root pointing at the upper occurrence of $f$ (L${}=f_1$), and the right-hand side root pointing at the inner occurrence of $f$ (R${}=f_2$).
Rewriting with this rule amounts to eliminating an $f$, the outermost one, from any drag $D$ having two consecutive symbols $f$, for example $D=f(f(a))$. 
The entire drag $f(a)$ which results from the computation is therefore the very subterm $f(a)$ of $f(f(a))$, which goes beyond what is usually done in term rewriting implementations, where only $a$ would derive from the original term $f(f(a))$.

Consider now the rule given by the drag $f_1^{[1]}(f_2(x))\oplus f_3^{[1]}(x)$ made of the two terms $f(f(x))$ and $f(x)$ with no common subexpression, and two roots, $L=f_1$ and $R=f_3$. 
Rewriting the term $f(f(a))$ with that rule has a completely different effect: 
It still eliminates the topmost $f$, but it will now generate a new vertex labeled $f$, and possibly (but not necessarily, depending whether the two sprouts labeled $x$ are shared or not) a new vertex labeled $a$, resulting in a term $f(a)$ that may be an entire, or only partial, copy of the subterm $f(a)$ of the term $f(f(a))$.
\qed\end{exa}

We continue our investigation of drag rewriting with an important property to be used in the coming subsection.

\begin{lem}[Monotonicity]
\label{l:monr}
Given a drag rewrite rule $\eta:L\ra R$, let $\ext(C,\xi)$ and $\ext(C',\xi')$ be rewrite extensions of $L,R$, respectively, compatible with $\eta$, so that $U=C\otimes_\xi L \dlrps{}{} C'\otimes_{\xi'} R=V$. 
\begin{itemize}
\item 
Let $W$ be a drag compatible with both $U$ and $V$. 
Then,
$U \oplus W \dlrps{}{} V \oplus W$.
\item
Let $\ext(E,\zeta)$ and $\ext(E',\zeta)$ be two extensions of $U$ and $V$ compatible with a root-map from $U$ to $V$ induced by the equimorphism from $C$ to $C'$.
Then, 
$U'=E\otimes_\zeta U \dlrps{}{} E'\otimes_{\zeta'} V=V'$.
\end{itemize}
\end{lem}
\begin{proof}
The first statement is straightforward.
The second uses associativity of product (Lemma~\ref{l:prodac}) to define the extensions of $L$ and $R$: $U'= (E\otimes_\theta C) \otimes_\delta L$ for some switchboards $\theta$ and $\delta$, and $V'= (E'\otimes_{\theta'} C') \otimes_{\delta'} R$. 
It follows from the closure property of equimorphisms by product that $\ext(E\otimes_\theta C, \delta)$ and $\ext(E\otimes_\theta C, \delta)$ are indeed extensions of $L$ and $R$. 
These extensions are compatible because the root-map from $U$ to $V$, induced by the equimorphism from $C$ to $C'$, ensures that sprouts $s$ of $L$ and $t$ of $R$ sharing the same variable label $x$ are mapped by $\theta$ and $\delta$ to products of equimorphic subdrags $A,B$ of $E,C$, respectively, for $s$, and equimorphic subdrags $A',B'$ of $E',C'$, respectively, for $t$. 
An arbitrary root-map would not ensure that the subdrags $B$ and $B'$ of $C$ and $C'$ are equimorphic, but that their products $B\otimes_\zeta L$ and $B'\otimes_{\zeta'} R$ happen to be equimorphic, which would not imply that $B$ and $B'$ are themselves equimorphic.
\end{proof}

A final question arises: 
Given a drag $D$, a rule $L\ra R$, and a rewriting extension $\ext(E,\xi)$ for $L$, does there exist a rewriting extension $\ext(E,\xi')$ for $R$? 
In general, yes, but there is a particular case for which this is not the case. 
It may indeed be that the switchboard $\xi$, which is well-behaved for $L$, is not well-behaved for $R$. 
This happens in the following situation: $u$ is a rooted internal vertex of $L$, $s:x$ is a sprout of $L$ accessible from $u$, $s':x$ is a rooted sprout of $R$ such that $\eta(u)=s'$, $t:y$ is a sprout of $E$, $\xi(s)=t$, and $\xi(t)=u$. 
Switchboard $\xi$ is well-behaved with respect to $L$ because $u$ is an internal vertex. 
Now, $\xi'(s')= t$ and $\xi'(t)=\eta(u)=s'$; hence, $\xi'$ is not well-behaved. 
This is the only situation where this may arise, but this is why we always assumed the existence of one rewriting extension for $L$ and one for $R$.

Note finally that our definition of rewriting implies that a rewrite succeeds as soon as compatible rewrite extensions exist: dangling edges can't be created in our model. The situation would be different without indegree preservation of matching.  
On the other hand, indegree preservation implies that rules must be duplicated with various numbers of roots at all vertices so as to fit with all possible rewriting extensions. 
We shall come back to these questions when investigating the categorical properties of drags in Section \ref{s:matching}.

\subsection{Congruences joinable by rewriting}
Our goal now is to prove that congruent pairs are joinable when the congruence $\cE$ is generated by a confluent set of drag rewrite rules $\cR$, yet again generalizing  the  situation for terms. 
Furthermore, since a congruence contains sharing equivalence, rewriting should actually operate uniformly on sharing-equivalent drags. 
We will therefore assume two different properties of a rewrite system:

\begin{defi}[Church-Rosser]
A drag rewrite system $\cR$ is 
\begin{enumerate}
\item
\emph{confluent modulo sharing } if for all drags $U,V,W$ such that 
$V \drlpstr{\cR} U \dlrpstr{\cR} V$, there exist drags $W, W'$ such that
$V\dlrpstr{\cR} W \she \drlpstr{\cR} W'$;
\item
\emph{coherent modulo sharing} if, for all drags $U,V,W$ such that 
$V \she U \dlrpstr{\cR} V$, there exist $W$ such that
$V\dlrpstr{\cR} W \she V$.
\item \emph{Church-Rosser (modulo sharing)} if it is both confluent and coherent modulo sharing.
\end{enumerate}
\end{defi}

\begin{thm}
\label{t:rewcong}
Let $\cE$ be a set of equations and $\cR$ a Church-Rosser set of rules that satisfy the following two assumptions: 
\begin{enumerate}
\item $\forall U=^\eta V\in \cE\;\st
U \dlrpstr{\cR} \she \drlpstr{\cR} V$ (root-map $\eta$);
\item $\forall \eta: L\ra R\in \cR\;\st L=_\cE^\eta R$.
\end{enumerate}
Then, $U\;=_\cE\;V$ iff $U \dlrpstr{\cR} \she \drlpstr{\cR} V$.
\end{thm}

\begin{proof}
This will follow from the soundness and completeness lemmas (\ref{l:sound} and \ref{l:complete}) below.
\end{proof}

Hence, it is sound and complete to check a congruence by rewriting---possibly until a normal form is obtained if rewriting is confluent and coherent (with respect to sharing equivalence). 
The resulting terms---possibly normal forms---should be checked for equality under sharing equivalence, which is decidable. 

Note that equations need not be pairs of patterns; they can be built from arbitrary drags having the same number of roots, provided rules (built from patterns) exist that satisfy (1). 
When equations are built from patterns, assumption (1) is usually obtained by orienting the equations in $E$ into rules, and then completing the obtained set of rules until confluence and coherence are satisfied, providing  assumption (2).

\begin{lem}[Soundness]
\label{l:sound}
Assuming (2), $\dlrpstr{\cR} \;\subseteq\; =_\cE$.
\end{lem}

\begin{proof}
By induction on the length of the derivation $\dlrpstr{\cR}$. Since equimorphism is included in any congruence, it is enough to carry out the proof for a single step $U \dlrps{p}{\cR} V$, which is a consequence of assumption (2).
\end{proof}

\begin{lem}[Completeness]
\label{l:complete}
Assuming (1), let $U=_\cE^\eta V$. Then, $U\dlrpstr{\cR} \she \drlpstr{\cR} V$.
\end{lem}

\begin{proof}
By induction on the definition of the congruence $=_\cE$. If $U=^\eta V\in\cE$, the result follows by Assumption (1). Otherwise, $U\dlrpstr{\cR} \she \drlpstr{\cR} V$  by induction hypothesis, and let
\begin{enumerate}
\item 
$V \she W$. We then conclude by applying the coherence assumption..
\item 
$W$ compatible with both $U$ and $V$. Then $U\oplus W \dlrpstr{\cR} \she \drlpstr{\cR} V$, since rewriting is closed under sum with a compatible drag, as well as sharing-equivalence.
\item 
$\ext(C,\xi)$ and $\ext(C',\xi')$ rewriting extensions of $U$ and $V$ respectively, compatible with the root-map $\eta$. By repeated applications of Lemma~\ref{l:monr}, 
$C\otimes_\xi U \dlrpstr{} D\otimes_\zeta W$, where $ext(D,\zeta)$ is an extension of $W$ compatible with the extension $\ext(C,\xi)$ of $U$.
Likewise, $C'\otimes_{\xi'} V \dlrpstr{} D'\otimes_{\zeta'} W'$, where $ext(D',\zeta')$ is an extension of $W'$ compatible with the extension $\ext(C',\xi')$ of $V$.
We are therefore left showing that these extensions are compatible, which results from the transitivity property of compatibility. 
\item $V=_\cE^\mu W$. Then, $U\dlrpstr{\cR} \she V'\drlpstr{\cR} V \dlrpstr{\cR} V''\she \drlpstr{\cR} W$. By confluence assumption, $V'\dlrpstr{\cR} T\she T'\drlpstr{\cR} V''$ for some sharing-equivalent drags $T$ and $T'$.
We then conclude by using the coherence assumption twice. 
\qedhere
\end{enumerate}

\end{proof}

Ensuring coherence of $\cR$ with respect to sharing equivalence can be achieved by completing the set of rules until sharing equivalence is obtained. 
The problem indeed is that two distinct subdrags of a left-hand side $L$ of rule may become equimorphic when computing the product of $L$ with respect to some rewrite extension. Then, these two subdrags may become identical by applying a sharing equivalence step, and the rule won't be applicable anymore. For example, assume $L=f(g(x),g(y))\ra f(a)$, and $D=[z=g(a)]f(z)$. Then, rewriting $D$ would require using the rule
$L'=[z=g(x)]f(z)\ra f(a)$. Computing these extensions requires unifying the subdrags $g(x)$ and $g(y)$, a process very similar to Knuth and Bendix completion \cite{knuth70,DBLP:books/el/leeuwen90/DershowitzJ90}. We know that drag unification is solvable for drags whose roots are organized as lists with repetitions, and that any solvable unification problem has a most general solution, which is unique up to isomorphism. We assume here (there is little doubt about that conjecture) that this is the same for our current model of drags with multisets of roots. So, making a set of rules coherent with respect to sharing-equivalence should not be difficult. Whether it could be terminating is another question.

Theorem \ref{t:rewcong} does not state that these equalities, generated by the equations for the first and by the rules for the second, coincide. Indeed, they do not:

\begin{exa}
\label{ex:rewcong}
Let $\cE=\{f^{[1]}(x)=g^{[1]}(x), f^{[1]}(x)=h^{[1]}(x,x)\}$, and $\cR=\{f^{[1]}(x)\ra h^{[1]}(x,x), g^{[1]}(x)\ra h^{[1]}(x,x)\}$. Using the extension $\ext(a^{[1]},\{\wir{x}{a}\})$ for both left- and right-hand sides of the equation, we get
$f(a)=_\cE g(a)$. Using now the extensions
$\ext(a^{[1]}\oplus a^{[1]}, \{\wir{x_1}{a_1}\})$ and
$\ext(a^{[1]}\oplus a^{[1]}, \{\wir{x_2}{a_1}, \wir{x_3}{a_2}\})$ for the first rule, we get
$f(a) \dlrps{}{} h(a,a)$. Finally, using the extensions 
$\ext(a^{[1]}, \{\wir{x_1}{a}\})$ and 
$\ext(a^{[1]}, \{\wir{x_2}{a}, \wir{x_3}{a}\})$ for the second rule, we get
$g(a) \dlrps{}{} [z=a]h(z,z)$. Therefore:
$$f(a) \dlrps{}{} h(a,a) \;\she\; [x=a]h(x,x) \drlps{} g(a).$$
\end{exa}

Could we define congruences and rewriting so that they would coincide? The answer is yes, we believe. 
For the former, we would simply need to include sharing-equivalence in the definition of a congruence. 
For the latter, we would need for rewriting to be modulo sharing, a more complex task indeed, which is yet another way to achieve coherence. 
Since distinct subdrags of a left-hand side of rule $L$ could become equimorphic, hence shared, by computing the product $D$ of $L$ with an appropriate rewrite extension, the injection of $L$ into $D$ would no longer be be injective. 
This question would be worth exploring, but is beyond the scope of the present paper.
 
\section{Drag Rewriting versus Term Rewriting}
\label{ss:drvtr}
In this section, we consider term rewrite rules having possibly a root at their head, and (try to) apply them to terms, including terms that possibly involve sharing.

Consider a rule $f(x,x)\ra g(x,x)$ with no roots on either side. The sprouts are $\red{x_1},\red{x_2}$ in the left-hand side and $\red{x_3}, \red{x_4}$ in the right-hand side.

Let $t_1=f(a,a)$ with vertices
    $f$ (on top) and $a_1$ (shared).
    Let $t_2=f(a,a)$, with vertices $f, a_1, a_2$ 
    And let $t'_1=g(a,a)$, with the shared vertex $a_1$ being the same as above,
    and $t'_2=g(a,a)$  with vertices $a_2$ on the left and $a_1$ on the right.
    \begin{itemize}
    \item
    $t_1 \dlrps{}{} t'_1$ with
    extension drag $\blue{E}$ being the two-rooted vertex \blue{$a_1^{[2]}$} and switchboards being $\xi=\{\red{x_1}\mapsto \blue{a_1}, \red{x_2}\mapsto \blue{a_1}\}$
    and $\xi'=\{\red{x_3}\mapsto \blue{a_1}, \red{x_4}\mapsto \blue{a_1}\}$.
    \item
    $t_2 \dlrps{}{} t'_2$ with
    extension drag \blue{$E=a_1^{[1]}\oplus a_2^{[1]}$}, and switchboards $\xi=\{\red{x_1}\mapsto \blue{a_1}, \red{x_2}\mapsto \blue{a_2}\}$ and $\xi'=\{\red{x_3}\mapsto \blue{a_2}, \red{x_4}\mapsto \blue{a_1}\}$.
    \item 
    $t_1 \dlrps{}{} t'_2$ with extension drags (two are needed here) \blue{$E=a_1^{[2]}$} and \blue{$E'=a_1^{[1]}\oplus a_2^{[1]}$} and switchboards
    $\xi=\{\red{x_1}\mapsto \blue{a_1}, \red{x_2}\mapsto \blue{a_1}\}$ and $\xi'=\{\red{x_3}\mapsto \blue{a_2}, \red{x_4}\mapsto \blue{a_1}\}$.
    Note that the vertex $a_1$ does not have the same number of roots in $E$ and $E'$. Note also that $E'$ cannot be identical to $E$ since $t'_2$ has additional vertices that do not originate from the rewrite rule, they must therefore come from the extension drag. This rewrite is therefore relational, using the isomorphic copy (a "clone") $a_2$ of $a_1$.
      \item 
    $t_2 \dlrps{}{} t'_1$ . Take $E=a_1^{[1]}\oplus a_2^{[1]}$, $E'=a_1^{[2]}$, 
    $\xi=\{\red{x_1}\mapsto \blue{a_1}, \red{x_2}\mapsto \blue{a_2}\}$ and $\xi'=\{\red{x_3}\mapsto \blue{a_1}, \red{x_4}\mapsto \blue{a_1}\}$. Note that the choice of $E'$ allowed us to "garbage collect" the vertex $a_2^{[1]}$ of $E$ implicitly.     
    The other choice $E'=a_1^{[2]}\oplus a_2^{[1]}$ would yield as a result the expression $t'_1 \oplus a_2^{[1]}$, hence disabling "garbage collection".
    \item
    $t_1\dlrps{}{} a_3^{[1]}\oplus t'_1$, where $a_3$ is a fresh copy of $a$, that is, a \emph{clone} of $a$. Take $E=a_1^{[1]}\oplus a_2^{[1]}$, $E'=a_1^{[2]}\oplus a_3^{[1]}$, with $\xi$ and $\xi'$ as above. Here, we have achieved "cloning" and "garbage collection" at the same time.
\end{itemize}

Terminating computations---to which we are partial---%
are incompatible with "cloning". 
In general,
functional rewriting restricts context extensions so as to avoid "cloning" and allow one to attain termination.
Relational rewriting, on the other hand, is more lax regarding "cloning" and nontermination.

\newcommand{\head}[1]{\overline{#1}}
We  provide now a formal statement comparing term rewriting to drag rewriting. Remember that terms have positions identifying their various subterms, and that rewriting a term $s$ to a term $t$ with the rule $l\ra r$ is defined as $s|_p=l\sigma$ for some position $p$ in $s$ and substitution $\sigma$, and $t=s[r\sigma]_p$, denoting by $s|_p$ the subterm of $s$ at position $p$, and by $s[u]$ the term obtained by replacing $s|_p$ by $u$ at position $p$. 

Term rewriting does not really apply to terms, that is, drags whose all vertices have a single predecessor but one which has none, but to equimorphisms classes of terms (and not isomorphisms classes, variables can't be renamed in terms to be rewritten). This is so because term rewriting does not distinguish between two copies of the same term $t$, for example, two disjoint drags $f(a)$ and $f(a)$. So, term rewriting \emph{looks} functional, but it is not, the result of a rewrite step is determined up to equimorphism only. The situation is indeed the same with functional drag rewriting, hence we can compare both.

 Any strict subdrag of a term has a root at its head. We therefore denote by $\head{t}$, for a term $t$, the vertex at the head of the ground drag $t$, and by $\drag{t}$ the drag obtained by adding a root at vertex $\head{t}$ of $t$. This encoding has three simple properties: first,  a term and its encoding have exactly the same vertices, hence, in particular, the same labeled sprouts; second, positions in terms make therefore sense for their encoding; and third, taking subterms in terms commutes with their encoding. A substitution $\sigma=\{x\mapsto x\sigma : x\in\Var{l}\}$ is encoded as the drag $\Sigma_{x\in\Var{l}} \drag{x\sigma}$, but also as the switchboard \smash{$\xi_\sigma=\{\wir{v}{\head{\drag{x\sigma}}} : (v:x)\in\spr{\drag{l}}=\spr{l}\}$}.
 A term rewrite rule $l\ra r$ is encoded as the drag rewrite rule $\drag{l}\ra\drag{r}$. A set of term rewrite rules $R$ is denoted by $\drag{R}$ when encoded as drag rewrite rules.

\begin{lem}
\label{l:subtermdrag}
Let $t$ and $l$ be terms and $p$ a position of $t$. Then,
$t=t[l\sigma]_p$ for some substitution $\sigma$ of domain $\Var{l}$ iff 
$\drag{t}=(\drag{t[z]_p}\oplus \drag{\sigma})\otimes_{\xi} \drag{l}$ for some fresh sprout $s:z$ and switchboard \smash{$\xi=\{\wir{s}{\head{\drag{l}}}\}\cup \xi_\sigma$}.
\end{lem}

 Note first that the statement makes sense: $t$ being a term, $t[z]_p$ and $\sigma$ do not share any vertex, hence the sum of their encodings is defined, and since $\wir{z}{\head{l}}$ and $\sigma$ don't interfer, the union $\wir{s}{\head{l}}\cup \xi_\sigma$ is a (non-directed) switchboard.
 
\begin{proof}
First, the only if direction.
Since $t=t[l\sigma]_p$, we have $\drag{t}=\drag{t[z]_p}\otimes_{\wir{s}{\head{t|_p}}} \drag{t|_p}$ by Lemma~\ref{l:rec}. We now apply Lemma~\ref{l:rec} to $l\sigma$, which requires to linearize $l$ so as to become the context in $l\sigma$ of the subterms associated with the substitution $\sigma$. Denoting by $\textit{lin}(l)$ the linearized version of $l$, by $\sigma'$ the substitution satisfying $l\sigma=\textit{lin}(l)\sigma'$, we have
$\drag{t|_p}=\drag{l\sigma}=\drag{\textit{lin}(l)\sigma'}= \drag{\textit{lin}(l)}\otimes_{\sigma'} \drag{\sigma}= \drag{l}\otimes_{\sigma} \drag{\sigma}$, where $\sigma$ and $\sigma'$ are now considered as switchboards by replacing an elementary substitution
$x\mapsto u$ by the wire $\wir{s:x}{\head{u}}$.
It follows that $\drag{t}=\drag{t[z]_p}\otimes_{\wir{z}{\head{t|_p}}} \drag{t|_p}=\drag{t[z]_p}\otimes_{\wir{z}{\head{t|_p}}}(\drag{l}\otimes_{\xi_\sigma} \drag{\sigma})=\drag{t[z]_p}\otimes_{\wir{z}{\head{\drag{l}}}}(\drag{l}\otimes_{\xi_\sigma} \drag{\sigma})=(\drag{t[z]_p} \oplus \drag{\sigma}) \otimes_{{\wir{z}{\head{\drag{l}}}, \xi_\sigma}}\drag{l}$.

The converse amounts to apply the wires in $\xi_\sigma$ (which do not depend on $\wir{s}{\head{\drag{l}}}$, yielding $\drag{t}=\drag{t[z]_p}\otimes_{\wir{s}{\head{\drag{l}}}} \drag{l}$, hence the expected result by unencoding the terms on both sides of the equality, since the switchboard reduced to the single wire \smash{$\wir{s}{\head{\drag{l}}}$} is now directed, hence acts as a substitution.
\end{proof}

Assuming now that left-hand sides of term rewriting rules are not variables (hence are patterns), we have:

\begin{thm}
Let $R$ be a set of term rewriting rules and $s,t$ be two terms such that. 
$s\dlrps{p}{R}\, t$. Then, $\drag{s}\dlrps{}{\drag{R}} \drag{t}$.
\end{thm}

\begin{proof}
The only if part is a direct application of the rewriting definitions and of Lemma~\ref{l:subtermdrag} to $s$ and $t$. Checking functionality is straightforward. 
\end{proof}

The converse is not true in general, since drag rewriting is defined up to sharing equivalence instead of equimorphism: rewriting with rule $\drag{l}\ra \drag{r}$ may introduce sharing when $r$ is non-linear, hence cannot be mimicked by term rewriting with the rule $l\ra r$. Here is the precise statement:

\begin{lem}
Let $R$ be a set of term rewriting rules and $s,t$ be two terms such that. 
$\drag{s}\dlrps{}{\drag{R}} \drag{t}$. 
Then, $s\dlrps{p}{R}\, t'$ for some $t'$ sharing equivalent to $t$.
\end{lem}

\begin{proof}
Using this time the if direction of Lemma~\ref{l:subtermdrag}.
The reason for obtaining a term $t'$ sharing-equivalent to $t$ rather than $t$ itself, is the use of a right-hand side switchboard possibly different from the left-hand side one.
\end{proof}

Encoding term rewriting as drag rewriting appears to be slightly different from the encoding of term rewriting by jungle rewriting, as described in \cite{PlumpJungle}. The reason is that drag rewriting does not  integrate the use of sharing equivalence (actually, sharing normal forms) to the encoding process as they do, since we only consider here what could be called "plain dag rewriting". Encoding sharing equivalence would require the use of a slightly different notion of drag rewriting, modulo sharing equivalence, which is not within the scope of this paper.

We could of course now repeat a similar development for dags as for terms, defining a dag rewrite rule as a pair of dags $L\ra R$, whose encoding is obtained by having an arbitrary number of roots (possibly none) at the heads of both $L$ and $R$. Of course, there will be infinitely many rules, but using the rule schema \smash{$L^{[n]} \ra R^{[n]}$} instead will do, $n$ being calculated on the fly thanks to indegree preservation. Then, a dag $D$ to be rewritten can be encoded by the identity: there is no need to add roots at the head of $D$ since dag rewrite rules can have no root at their head.

\section{Categorical Interpretation of Drag Rewriting}
\label{s:matching}

Our definition of drag rewriting is concrete: 
Drags are specific concrete (multi-) graphs, and drag rewrite rules are pairs of drags whose roots are bijectively related, used to rewrite other drags. 
We found that drags and their morphisms form a category. 
So a natural question to ask is whether we could define drag rewriting in terms of the general DPO framework for graph transformation. 
We will see that this is in general impossible,  there   being several reasons for that failure.

\subsection{Rules}
In DPO graph transformation, rules are spans $L \stackrel{l}{\leftarrow} K \stackrel{r}{\to} R$, where $l,r$ are monomorphisms (sometimes a specific kind of monomorphisms), and $K$ is assumed to be the shared subgraph of $L$ and $R$.

Considering this definition, it seems reasonable to represent a drag rewriting rule $\eta: L\ra R$ as the DPO rule \smash{$L \stackrel{l}{\leftarrow} K \stackrel{r}{\to} R$}, where $K$ is the subgraph shared by $L$ and $R$, and $l,r$ are the natural injections $l: K \hookrightarrow L$ and $r: K \hookrightarrow R$. 
Unfortunately, this representation may be incorrect, as the following example shows. 
Suppose that $L$ is the drag $f \Longrightarrow a$ 
($a$ is a double successor of $f$) and $R$ is the drag 
$g_1 \Longrightarrow a \longleftarrow g_2$ (where vertex $a$ is shared with $L$). 
Then, the DPO representation of $K$ would be the drag $a$ with some roots, but at least one of the injections $l,r$ would fail to be a monomorphism, since they could not both preserve the indegree of $a$, which is 2 on one side and 3 on the other. This may restrict the choice
of a common subdrag by DPO to a non-maximal one---the empty one in this example.


\subsection{Matching}
In the DPO approach, the operation of matching a drag $D$ against a drag $L$ is defined by means of a (mono)morphism, while in drags this done by the definition of an extension $\ext(C,\xi)$ such that $D=L\otimes_\xi C$. We will see that both forms of matching are equivalent.

\begin{lem}
\label{l:ipm}
Given two disjoint drags $L,C$ and a switchboard $\xi$ for $L,C$, the natural injection from $L$ to $C\otimes_\xi L$ is a monomorphism.
\end{lem}

\begin{proof}
By Lemma~\ref{l:ism}, the natural injection from $L$ to $L\oplus C$ is a monomorphism. By Lemma~\ref{l:injw}, the natural injection from $L\oplus C$ to $L\otimes_\xi C$ is a monomorphism. We conclude by Lemma~\ref{l:closureprop}.
\end{proof}

Next, we consider the converse, namely, the existence of an extension when given a monomorphism from $L$ to $D$. First, we define the categorization of edges corresponding to the colours used in the examples of monomorphisms given earlier:

\begin{defi}[Inside/created/entering/outside edges]
\label{d:catmonos}
Given drags $L=\langle V, R, L, X, S\rangle$ and $D=\langle V', R', L', X', \varnothing\rangle$ with no vertex in common and a monomorphism $\omicron : V\to V'$, we characterize an edge $\edgen(u',i,v')$ of $D$ as follows:
\begin{enumerate}
    \item
    It is an \emph{inside} edge at $v'$ if $u'=\omicron(u)$ and $v'=\omicron(v)$ for some internal vertices $u,v$ of $L$ such that $\edgen(u,i,v)\in X$;
   \item 
    it is \emph{created} by the \emph{creating} edge $\edgen(u,i,s)\in X$ if $u'=\omicron(u)$ for some internal vertex $u$ of $L$, and $s$ is a sprout such that $\omicron(s)=v'$;
    \item
    it is an \emph{entering} edge at $v'$ if $v'=\omicron(v)$ for some internal vertex $v$ of $L$, and $\edgen(u',i,v')$ is not the image of an edge in $L$  by $\omicron_X$;
    \item
    it is an \emph{outside} edge at $v'$ if $v'$ is not the image of an internal vertex of $L$ by $\omicron$, and $\edgen(u',i,v')$ is not the image of an edge of $L$ by $\omicron_X$.
\end{enumerate}
\end{defi}

Inside and created edges of $D$ originate from some (unique) edge in $L$, while entering and outside edges don't originate from edges in $L$. Entering edges must have their head mapped from some internal vertex of $L$, while outside edges do not. This is therefore a categorization, which is illustrated in Example \ref{f:catmonos}.

\begin{lem}
\label{l:catmonos}
Let $L,D$ be drags satisfying the assumptions of Definition~\ref{d:catmonos}, and $\omicron$ be a monomorphism from $L$ to $D$. Then each edge of $D$ is precisely one of the above four kinds.
\end{lem}

\begin{proof}
We show first that these four kinds of edges are disjoint as a consequence of injectivity of $\omicron$. 
The first two kinds are disjoint because two edges $\edgen(u,i,v)$ and $\edgen(u,j,s)$ of $L$ such that $v\neq s$ are different, hence $i\neq j$.
The last two kinds are disjoint from the first two since edges in $D$ do not originate from edges in $L$. The fourth is disjoint from the third since they discriminate on $v'$ being the image of an internal vertex in $L$.

We show now that all edges in $D$ at considered. All edges in $D$ obtained from an edge $\edgen(u,i,v)$ in $L$ by $\omicron_X$ must have an internal vertex as tail vertex $u$, hence are either inside or created. If $\edgen(u',i,v')$ is not obtained from an edge in $L$, there are two cases: either $v'=\omicron(v)$ for some internal vertex $v$ of $L$ and it is an entering edge; or there is no such $v$ and it is an outside edge.
\end{proof}

\begin{figure}[t]

\begin{tikzpicture}[
Droundnode/.style={circle, draw=blue!50, fill=blue!5, very thick, minimum size=4mm},
namenode/.style={},
Dproundnode/.style={circle, fill=red!5, very thick, minimum size=4mm},
namenode/.style={}
]
%
\node[Dproundnode]    (f)    []  {$f$};
\node (g) [left=1cm of f] {$g$};
\node (x) [below=12mm of g, very thick] {$x$};
\node (a) [below=10.5mm of f, very thick] {$a$};
\draw[->, very thick, red] (g) to node [right] {} (x);
\draw[->, very thick] (f) to node [right] {} (a);

\node (femb) [right=15mm of f] {};
\node (emb) [below=1cm of femb] {${\huge \hookrightarrow}$};

\node[Dproundnode]    (f1)    [right=60mm of f, very thick]    {$f$}; 
\node  (g1)    [left=10mm of f1, very thick]    {$g$}; 
\node  (b)     [below=11mm of g1, very thick]    {$b$}; 
\node  (a1)    [below=10mm of f1, very thick]    {$a$}; 

\draw[->, very thick, blue] (b) to node [] {} (a1);
\draw[->, very thick, red] (g1) to node [left] {} (b);
\draw[->, very thick, black] (f1) to node [left] {} (a1);
\draw[->, very thick, green] (b) .. controls + (-0.8,-0.8) and + (1,-0.8) .. (b);

\node (m) [below=0mm of a] {};
\node (lm) [right=12mm of m] {$\omicron(x)=b$};

\end{tikzpicture}
\caption{An injection with all four kinds of edges:
black for inside, red for created, blue for entering and green for outside.}\label{f:catmonos}
\end{figure}
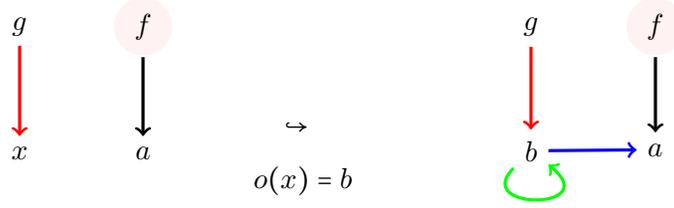

\begin{lem}
\label{l:ipmconverse}
Given a drag $D$, a drag $L$ with no rootless isolated sprout and no sprout in common with $D$, and an injection $\omicron: L \hookrightarrow D$, there exists a rewriting extension $\ext(C,\xi)$ of $L$ such that $D =C\otimes_\xi L$.
\end{lem}

\begin{proof}
Without loss of generality, sprouts of $D$, if any, will be considered internal vertices of $D$, since they will not belong to the domain of $\xi$.
In what follows, we successively {(a)} construct the context extension $C$, {(b)} then the switchboard $\xi$, and finally {(c)} verify that putting them together we have $D = C \otimes_\xi L$. These constructions rely on
Lemma~\ref{l:catmonos}, ensuring that each edge of $D$ belongs to exactly one kind.

\begin{trivlist}\itemindent 1ex
\item[(a)]
Construction of context $C$:
\begin{itemize}
\item Labeled internal vertices. Let $V$ be the vertices of $D$ and $I$ the internal vertices of $L$ (which are also internal vertices of $D$ by the assumption that $\omicron$ is an injection). 
Then, the set of internal vertices of $C$ will be  $W=V\setminus I$, each one equipped with its label in $D$. 
\item
Labeled sprouts. The set $S$ of sprouts of $C$ consists of  fresh sprouts $t_{u,i}:x_{u,i}$, bijectively associated with the pair $(u,i)$
for each  entering edge $\edgen(u,i,v)$ in $D$---with $u\in W$ and $v\in I$,
plus sprouts $t_{u,i}:y_{u,i}$ for each creating edge $\edgen(u,i,v)$ in $D$, such that $u,v\in I$, but $\edgen(u,i,v)$ is not an edge of $L$,  $\edgen(u,i,s)$ is an edge in $L$ and $\omicron(s) = v$. 
Notice that the edge  $\edgen(u,i,s)$, with $\omicron(s) = v$, must exist by the definition of morphism.
Notice also that, by definition, the variables of any two sprouts in $C$ must be different, implying that the context $C$ will be linear.
\item 
Edges. The set of edges of $C$ consists  of all outside edges $\edgen(u,i,v)$ in $D$, such that $u,v\in W$, plus edges  $\edgen(u,i,t_{u,i})$, for each  entering edge $\edgen(u,i,v)$ in $D$ with $u\in W$ and $v\in I$.
 These sprouts are not isolated. 
 \item Roots. Finally, each vertex $v$ in $W$ is equipped with roots so that $v$ has the same indegree in $C$ and $D$. 
 Sprouts $t_{u,i}$, $t_{u,i}:y_{u,i} \in C$, associated to the creating edge $\edgen(u,i,v)$ in $D$, are equipped with a single root, 
hence will be rooted, isolated sprouts.
\end{itemize}

\item[(b)]
Construction of switchboard $\xi$:
\begin{itemize}
\item 
For each creating edge $\edgen(u,i,s)$ of $L$, 
$\xi_L(s) = t_{u,i}:y_{u,i}$. 
\item
For any other  sprout $s \in L$, $\xi_L(s) = \omicron(s)$.
\item  For each sprout $t_{u,i}:x_{u,i} \in C$, associated to  entering edge $\edgen(u,i,v)$ in $D$, $\xi_C(t_{u,i}:x_{u,i}) = v$.
\item
For each sprout $t_{u,i}:y_{u,i} \in C$, associated to creating edge $\edgen(u,i,v)$ in $D$, $\xi_C(t_{u,i}:x_{u,i}) = v$.
\end{itemize}
We are left with showing that $\xi$ is a switchboard. 
The union $\xi_L \cup \xi_C$ is coherent and well-behaved: 
It is coherent since $\omicron$ is a morphism and $C$ is linear; by definition $\xi_L \cup \xi_C$ is functional and well-founded. 
Finally,  $\xi_L \cup \xi_C$ is injective, since $\omicron$ is a morphism and, hence, root preserving. 
Therefore, $\ext(C,\xi)$ is a rewriting extension.

\item[(c)]
Verification that $D =C\otimes_\xi L$:\\
Internal vertices of $D$ and $C\otimes_\xi L$ coincide, and so too their labeling, hence implying that they have the same number of edges. To show that their edges coincide, it is therefore enough to show that every edge of $D$ is an edge of $C\otimes_\xi L$, which we do by inspecting the four categories of an edge $\edgen(u,i,v)$ of $D$:
\begin{itemize}
\item Outside edges: $u,v \in W$. By construction of $C$, $\edgen(u,i,v)$ is an edge of $C$ and therefore of $C\otimes_\xi L$.
\item 
Inside edges: $u,v\in I$ and $\edgen(u,i,v)$ is an edge of $L$. Then $\edgen(\omicron(u),i,\omicron(v))$ is an edge in $C\otimes_\xi L$.
\item 
Created edges: $u,v\in I$, but $\edgen(u,i,v)$ is not an edge of $L$. Then, there must exist a creating edge $\edgen(u,i,s)$ in $L$.
By construction, $\xi_L(s)=t_{u,i}$ and $\xi_C(t_{u,i})=v$; hence $\edgen(u,i,v)$ is an edge in $C\otimes_\xi L$.
This case shows the need for bouncing from $L$ to $C$ and back from $C$ to $L$.
\item
Entering edge: $u\in W$, $v\in I$. By construction,
$C$ includes a sprout $t_{u,i}:x_{u,i}$ and  an edge $\edgen(u,i,t_{u,i})$, with $\xi_C(t_{u,i}) = v$; hence $\edgen(u,i,v)$ is an edge in $C\otimes_\xi L$.
\end{itemize}
It follows that all vertices have the same incoming edges in $D$ and $C\otimes_\xi L$, hence the same number of them. Since they have identical indegree by construction, they must have the same number of roots, which terminates verification and proof.
\qedhere
\end{trivlist}
\end{proof}

In case $L$ is a rooted isolated sprout, the constructed context $C$ is identical to $D$, and the switchboard $\xi$ contains the single wire $\wir{s}{\omicron(s)}$.
Then, all edges of $D$ are outside edges. Note that verification is straightforward in that case since isolated sprouts are then identities for product by Lemma~\ref{l:prodac}.

\begin{exa}
Let $D$ be a drag with two internal vertices labeled $h^{[1]}$ and $f^{[1]}$, both of arity 3, and edges $\edgen(h,1,f)$, $\edgen(h,2,f)$, $\edgen(h,3,h)$, and $\edgen(f,1,h)$, $\edgen(f,2,f)$, $\edgen(f,3,h)$. 
Now consider the drag
 $\smash{L=f^{[4]}(x_1^{[1]},x_2^{[1]},x_3)}$ with the injection $\omicron$ 
mapping its four vertices $f, x_1, x_2, x_3$ to $f,h,f,h$, respectively,
and \smash{$\omicron_R(x_1^{[1]})= \edgen(h,3,h)$}, \smash{$\omicron_R(f^{[3]})= \{\edgen(h,1,f),\edgen(h,2,f),\edgen(f,3,f)\}$}.

Let now $C$ be the drag $h^{[4]}(y_1,\textsc{self},y_3)\oplus z^{[2]}$. We construct the expected extension by processing all missing edges in turn:
\begin{enumerate}
    \item $\edgen(f,2,f)$: add $z^{[2]}$ to $C$, define $\xi_L(x_2)=z$ and $\xi_C(z)=f$; remove a root from $f$ and $\edgen(f,2,f)$ from $\omicron_R(f)$;
    \item $\edgen(f,1,h)$: define $\xi_L(x_1)=h$ and remove $\edgen(f,1,h)$ from $\omicron_R(h)$;
    \item $\edgen(f,3,h)$: define $\xi_L(x_3)=h$ and remove $\edgen(f,3,h)$ from $\omicron_R(h)$;
    \item $\edgen(h,1,f)$: define $\xi_C(y_1)=f$ and remove $\edgen(h,1,f)$ from $\omicron_R(f)$;
    \item $\edgen(h,2,f)$: define $\xi_C(y_2)=f$ and remove $\edgen(h,2,f)$ from $\omicron_R(f)$.
\end{enumerate}
We therefore obtain the extension:
$$(h^{[3]}(y_1,y_2,\textsc{self})\oplus z^{[2]},\;
\{\wir{x_1}{h},\wir{x_2}{z},\wir{z}{f}, \wir{x_3}{h},\wir{y_1}{f},\wir{y_2}{f},\wir{y_3}{x_1},\wir{z}{h}\}).$$
{We observe that mapping $x_2$ to $z$ and then $z$ to $f$ produces the edge $\edge(f,f)$, while other edges are produced more directly, as pointed out above, not being created edges of $L$ in $D$.}
\qed
\end{exa}

The previous lemmas imply that matching based on composition and matching based on the existence of an injection coincide:

\begin{thm}
\label{t:matchingsame}
Given a drag $D$, a drag $L$ with no rootless isolated sprout and no sprout in common with $D$, there exists an injection $\omicron: L \hookrightarrow D$ iff there exists a rewriting extension $\ext(C,\xi)$ of $L$ such that $D =C\otimes_\xi L$.
\end{thm}

We have not claimed uniqueness of the rewriting extension $\ext(C,\xi)$ when given $D$, $L$, and $\iota$, and there are indeed many rewriting extensions satisfying Lemma~\ref{l:ipmconverse}, hence Theorem \ref{t:matchingsame}. 
The relevance of uniqueness lies in the fact that given $D,L,\omicron$, the result of rewriting $D$ with $L\ra R$ at $\omicron$ would then be deterministic, as is usually expected from a functional rewriting mechanism.
Uniqueness could indeed be achieved by introducing the definition of a \emph{reduced extension}, imposing two requirements: 
The switchboard $\xi$ should not bounce between $L$ and $C$ more than necessary.
And isolated sprouts of the extension context $C$ should have a single root, hence cannot be targets for two different wires of the switchboard.
This question merits further investigation, and it would be interesting to have both a matching and an unification algorithm for this version of drags, in the style of \cite{JO22a}. 

\subsection{Rewriting}
\label{ss:catrew}
We now consider whether drag rewriting can be defined algebraically, in terms of a double pushout in \textbf{Drags}, the category of drags. A standard approach would be to show that the category of drags has some form of adhesivity. 
Unfortunately, drags are not even $\mathcal{M,N}$-adhesive \cite{DBLP:conf/gg/HabelP12}, as shown below, because we cannot always ensure the existence of pushouts when they are needed.
This is the second reason for failure of DPO.

We can see the problem with the following counterexample. Suppose that we have drags $D_0 = [v^{[1]}:a$], $D_1 = [v_1:f \longrightarrow v:a]$, and $D_2=[v_2:g \longrightarrow v:a$]. Let now $\omicron_1$ and $\omicron_2$ be the natural injections, mapping $D_0$ to $D_1$ and $D_2$, respectively. In particular $\omicron_1$ and $\omicron_2$ would map the root at $v$ to the edges $v_1 \longrightarrow v$ and $v_2 \longrightarrow v$, respectively.
Then, it is easy to see that there is no pushout of $\omicron_1$ and $\omicron_2$ in the category of drags. The obvious candidate would be the drag $D_3 = [v_1:f \longrightarrow v:a \longleftarrow v_2: g]$. The problem is that the natural injection $\omicron_1': D_1 \to D_3$, mapping vertices $v,v_1$ and edge $v_1 \longrightarrow v$ of $D_1$ to vertices $v,v_1$ and edge $v_1 \longrightarrow v$ in $D_3$, is not even a morphism, since the indegree of $v$ in $D_3$ is 2, while it is 1 in $D_1$, and the same happens with $\omicron_2': D_2 \to D_3$.

This counterexample shows a form of dangling edge problem in drags: 
If $\omicron_1$ maps $v:a$ in $D_0$ to $v:a$ in $D_1$, and the root on $v:a$ in $D_0$ to the edge $v_1 \longrightarrow v$ in $D_1$, there would not exist a pushout complement to $(\omicron_1,\omicron_1')$ since, after deleting $v_1:f$ from $D_3$, the edge $v_1 \longrightarrow v$ would be dangling. However, we may notice that, while in the context of DPO transformations,  dangling conditions ensure the existence of pushout complements but do not interfere with the existence of pushouts, this is no longer the case for drags.

Pushouts in \textbf{Drags}, when they exist, have the form of the diagram in Lemma~\ref{l:pushout}.
We use here the fact that if there exist a monomorphism
$\omicron_1: D_0 \ra D_1$, then $D_1$ is the product of $D_0$ by some rewriting extension $\ext(C_1,\xi_1)$ by Theorem \ref{s:matching}. 
Assume now that we are given two injections $\omicron_1: D_0 \ra D_1$ and $\omicron_1: D_0 \ra D_1$, so that $D_1$ and $D_2$ are compatible, $D_0$ being their shared subdrag. 
Then, the context subgraphs $C_1$ and $C_2$ must be disjoint, and disjoint from $D_0$. 
This explains the assumptions given below, which are exactly the assumptions needed for applying a rewrite step in the DPO model.

\begin{lem}
\label{l:pushout}
Given monomorphisms $\omicron_1: D_0 \to D_0 \otimes_{\xi_1} C_1$ and $\omicron_2: D_0 \to D_0 \otimes_{\xi_2} C_2$, such that $D_0$, $C_1$, and $C_2$ are pairwise disjoint, $\xi_1\cup\xi_2$ is a switchboard for $(D_0, C_1\oplus C_2)$,
and the pushout of $\omicron_1$ and $\omicron_2$ exists, then
the pushout object must be the drag $D_0\otimes_{\xi_1\cup\xi_2} (C_1\oplus C_2)$
and the monomorphisms from  $D_1$ and $D_2$ to the pushout object are the natural injections $\omicron'_1$ and $\omicron'_2$,
as depicted here:
\begin{figure}[h!]
$
    \xymatrix{
        D_0 \ar@{}[dr]_{}|{po} \ar@{^{(}->}[d]_{\omicron_2} \ar@{^{(}->}[r]^{\omicron_1} & D_1 = D_0\otimes_{\xi_1} C_1 \ar@{^{(}->}[d]_{\omicron'_1}  
        \\
        D_2 = D_0\otimes_{\xi_2} C_2 \ar@{^{(}->}[r]_{\omicron'_2} & D_3=D_0\otimes_{\xi_1\cup\xi_2} (C_1\oplus C_2) 
       }
$
\end{figure}
\end{lem}

\begin{proof}
Let $D'_3= D_0\otimes_{\xi_1\cup\xi_2} (C_1\oplus C_2)$, and $\omicron''_1, \omicron''_2$ be the natural injections of the vertices in $D_1, D_2$ to the vertices in $D'_3$. Under our assumptions,  $\xi_1\cup\xi_2$ is safe, hence 
$D'_3=(D_0\otimes_{\xi_1} C_1) \oplus (D_0\otimes_{\xi_2} C_2)$ by distributivity (Lemma~\ref{l:distr}).
Therefore, $\omicron''_1$ and $\omicron''_2$ are monomorphisms of $D_1, D_2$ to $D'_3$ by Lemma~\ref{l:ism}(2). 

Using our assumption of existence of a pushout, let $(D_3, \omicron'_1, \omicron'_2)$ be a pushout of $D_1$ and $D_2$ in \textbf{Drags}. 
By the universal property of pushouts, there exists a unique morphism $\omicron_3$ from $D_3$ to $D'_3$, such that $\omicron''_1 = \omicron'_1 \circ \omicron_3$ and $\omicron''_2 = \omicron'_2 \circ \omicron_3$, whose vertex map $\omicron_3:\mathcal{V}(D_3)\ra\mathcal{V}(D'_3)$ is the corresponding universal morphism in \textbf{Set}. 
By definition of $D'_3$, we know that $\mathcal{V}(D'_3)$ is also a pushout object of the diagram depicted in Figure~\ref{fig:pos}. 
Hence, $\omicron_3$ is an isomorphism in \textbf{Set} and, since it is a morphism in \textbf{Drags}, it is an isomorphism in \textbf{Drags}.
\end{proof}

\begin{figure}[t]
\begin{minipage}{0.48\linewidth}\hspace*{8mm}
$
    \xymatrix{
        \mathcal{V}(D_0) \ar@{}[dr]_{}|{po} \ar@{->}[d]_{\omicron_2} \ar@{->}[r]^{\omicron_1} & \mathcal{V}(D_1) \ar@{->}[d]_{\omicron'_1} \ar@/^/[ddr]^{\omicron''_1} 
        \\
        \mathcal{V}(D_2) \ar@{->}[r]^{\omicron_2'} \ar@/_/[drr]_{\omicron''_2}& \mathcal{V}(D_3) \ar@{.>}[dr]|{\omicron_3}&&\\
       && \mathcal{V}(D'_3) \\
       }
$
    \caption{Pushout for sets of vertices.}
    \label{fig:pos}
\end{minipage}
\begin{minipage}{0.48\linewidth}\hspace*{12mm}
$
    \xymatrix{
        D_0 \ar@{}[dr]_{}|{po} \ar@{^{(}->}[d]_{\omicron_2} \ar@{^{(}->}[r]^{\omicron_1} & D_1 \ar@{^{(}->}[d]_{\omicron'_1} \ar@/^/[ddr]^{\omicron''_1} 
        \\
        D_2 \ar@{^{(}->}[r]^{\omicron'_2} \ar@/_/[drr]_{\omicron''_2}& D_3 \ar@{.>}[dr]|{\omicron_3}&&\\
       && D'_3 \\
       }
$
    \caption{Pushout for drags.}
    \label{fig:podc}
\end{minipage}
\end{figure}

We now describe a necessary condition for the existence of pushouts of monomorphisms in \textbf{Drags}:

\begin{lem}
\label{l:pushcond}
Given monomorphisms $\omicron_1: D_0 \to D_1 =(D_0 \otimes_{\xi_1} C_1)$ and $\omicron_2: D_0 \to D_2 =(D_0 \otimes_{\xi_2} C_2)$, such that $C_1$ and $C_2$ are disjoint, then $\omicron_1$ and $\omicron_2$ have a pushout only if for every vertex $v \in V_0$:
$$R_0(v) \ge (pred(v,D_1) - pred(v,D_0)) + (pred(v,D_2) - pred(v,D_0)).$$
\end{lem}

\begin{proof}
First, $(pred(v,D_1) - pred(v,D_0))$ ($(pred(v,D_2) - pred(v,D_0))$, respectively) is the number of edges of head $v$ added to $D_0$ by the composition $(D_0 \otimes_{\xi_1} C_1)$ ($(D_0 \otimes_{\xi_1} C_2)$, respectively) to obtain $D_1$ ($D_2$, respectively).

By Lemma~\ref{l:pushout}, the pushout object of $\omicron_1$ and $\omicron_2$ is $D_0\otimes_{\xi_1+\xi_2} (C_1\oplus C_2)$. 
Suppose, contrary to the claim, that $R_0(v) < (pred(v,D_1) - pred(v,D_0)) + (pred(v,D_2) - pred(v,D_0))$, for some $v \in V_0$. 
The composition $D_0\otimes_{\xi_1+\xi_2} (C_1\oplus C_2)$ would need to add to $D_0$ $pred(v,D_1) - pred(v,D_0)) + (pred(v,D_2) - pred(v,D_0))$ edges whose target is $v$. 
But $v$ is lacking roots for that purpose, so this is impossible.
\end{proof}

This implies that \textbf{Drags} is not $\mathcal{M,N}$-adhesive, and hence is not $\mathcal{M}$-adhesive either. For being $\mathcal{M,N}$-adhesive, the category would need to have two designated classes of monomorphisms, $\mathcal{M,N}$, that, among other properties, should satisfy that \textbf{Drags} has all pushouts along 
$\mathcal{M}$-morphisms and also along $\mathcal{N}$-morphisms, respectively. 
However, the above lemma shows that there is no class of monomorphisms $\mathcal{M}$ such that, if $\omicron_1 \in \mathcal{M}$, then for any monomorphism $\omicron_2$ with the same domain as $\omicron_1$, $\omicron_1$ and $\omicron_2$ have a pushout, since the existence of this pushout depends on $\omicron_2$. And the same happens for monomorphisms in $\mathcal{N}$. Therefore, we conclude:

\begin{thm}
The category \textbf{Drags} is neither $\mathcal{M,N}$-adhesive nor $\mathcal{M}$-adhesive.
\end{thm}

\section{Discussion}
\label{s:discussion}
In the course of this study of graph rewriting, we have made a number of choices among alternatives, motivated by what seemed to us to be either more general and useful, or simpler and more convenient.

In Remark~\ref{rem:roots}, we explained why we have chosen to consider multisets of roots, rather than lists as in \cite{DBLP:journals/tcs/DershowitzJ19} or sets. This design choice impacted our definition of composition, for which we decided to maintain a strong invariant: the indegree of each individual vertex. We indeed tried an alternative, using root transfer as in \cite{DBLP:journals/tcs/DershowitzJ19}, which resulted in more complex technicalities that we were not able to resolve in a satisfactory manner. Root transfer considers roots as plugs waiting for a connection \emph{there}, while indegree preservation considers roots as wires waiting for a connection \emph{at the other end}.

Another issue is how to understand multiple instances of variables.
We have already explained in Section~\ref{s:cohere} why we require
equimorphism of the subdrags connected to different sprouts with the same label, rather that the weaker isomorphism suggested in \cite{DBLP:journals/tcs/DershowitzJ19} or the stronger identity relation used in \cite{DBLP:journals/tcs/DershowitzJ19}. We have also seen that this
gives us the very helpful Lemmas~\ref{l:closureprop} and \ref{l:injw}. 
The latter is extremely important in that it allowed us to relate, in Section \ref{s:matching}, two completely different notions of matching: the traditional one, matching as a monomorphism, and the new one, matching as a drag extension. 
This relationship is a major justification for the drag model. But is equimorphism the best possible answer?

In the remainder of this section, we hint at variations that may extend the capabilities of the drag model. 
First, we briefly discuss a more general definition of coherence of a set of wires. 
Second, and significantly, we consider how sharing can be improved by a definition of rules allowing their left- and right-hand sides to share subdrags. 
Next, we suggest dropping the fixed arity of labeled vertices, and consequently hint at a more flexible graph-rewriting model for which sprouts and roots are particular cases of a more general notion of connector.

\subsection{Coherent sets of wires}
There is a slight potential for generalization here,  replacing \emph{equimorphism} by \emph{sharing equivalence} in the definition of a coherent set of wires. 
This variant would require  changes in the definition of monomorphisms so as to preserve Lemma~\ref{l:injw}, a change that should not impact Theorem \ref{t:cat} since isomorphisms preserve sharing equivalence.
The computation of a product $L\otimes_\xi C$ can render two subterms of $L$ equimorphic, or even sharing equivalent, and they might then be shared in the resulting drag. 
As a consequence, matching would no longer be injective  on internal vertices. 
Note that part of the \dpo community uses  non-injective matching, although in  \cite{HabelMP98} it is shown that injective matching is more powerful.
We won't explore that path here, but it is worth mentioning as a potential area of future investigation.

\subsection{Varyadic labels}\label{sec:sym}
\label{ss:vl}
Drags were designed for generalizing terms and term rewriting. Accordingly, a vertex of a drag comes with a label equipped with a fixed arity that governs the number of successors of that vertex, a constraint that has not, however, been central to the theory of drags developed here.
Extending the model to deal with arbitrary graphs is possible by relaxing the fixed arity constraint, allowing for bounded or even unbounded arities. 

In the fixed-arity model, it is crucial that wiring does not change the number of outgoing edges at a vertex. 
It follows that only sprouts can be mapped to other vertices, implying that decomposing a drag into smaller pieces can only be done by cutting its edges. 

This limitation disappears with varyadic arities.
Thus an alternative would be to decompose drags by "slicing" apart vertices instead of cutting edges.
The incident edges (regardless of direction) of the vertex are then split between the slices.
Dissecting an internal vertex creates a new "snap",
along with the leftover "base" vertex.
Composition (or wiring) connects then snaps with vertices, which may also be snaps.
Decomposition of a drag into two is straightforward in this alternate model, as would be decomposition into single-edge atoms.

\section{Related Work}

There are several competing approaches to graph rewriting, as already sketched in the introduction. Here, we list and discuss some of the more closely related proposals.

\subsection{DPO}
Introduced in the early seventies,
the double-pushout (\dpo) approach~\cite{DBLP:conf/focs/EhrigPS73} is the best studied and most popular approach to graph transformation. 
There are many varieties of graphs that may be of interest in different contexts. 
For example, one may work with directed or undirected graphs; they may be typed or untyped; 
they may be labeled, unlabeled, or include attributes that represent values stored in vertices 
or edges; they may be graphical structures like Petri nets or state-transition diagrams, or they 
even may be drags. 
It should be clear that studying graph rewriting separately for each kind of graph is a waste of time since there is almost no difference between rewriting a directed or an undirected graph or any other kind of graphical structure. 
Using categorical constructions allows the \dpo approach to 
describe and study at one and the same time rewriting for all manner of structures that satisfy some given properties. 
The \dpo approach is currently defined for any category of objects that is \emph{adhesive} \cite{libro,LS06} (or \emph{$\mathcal{M}$-adhesive}~\cite{EhrigGHLO14,EhrigGHLO12}, or even \emph{$\mathcal{M}-\mathcal{N}$-adhesive}~\cite{DBLP:conf/gg/HabelP12}). 
This includes most graph categories as well as other graphical structures, and other categories of objects, like sets, bags, or algebraic specifications.

In Section \ref{s:matching}, we  showed that the \dpo construction applies to drags as well, and can be extended so as to be relational rather than functional.
The definition of a rule $L\ra R$ as a single drag with left- and right-hand sides roots L and R as introduced at Definition~\ref{d:rules} looks very much like a \dpo rule, the subgraphs accessible from both L and R serving as the interface. 
The only difference is that we do not force the rewriting extensions of $L$ and $R$ to be identical, but rather to be compatible, which makes sense for concrete graphs. 
A direct benefit of compatibility combined with a relational definition of rewriting is that this extended version of \dpo has a built-in ability for erasing and cloning subdrags. 
Furthermore, drag rules can be nonlinear, on the left or right, while DPO rules are essentially linear. 
We can therefore claim to have solved the problem of allowing nonlinear variables in drag rewrite rules, a problem the importance of which was stressed already in \cite{DBLP:conf/gg/Parisi-PresicceEM86}, and which has remained unsolved since then, despite numerous attempts. 
Non-linearity is permitted here by an appropriate definition of morphisms for drags with variables. 
We believe that similar ideas should scale to graphs whose vertices can have varyadic labels, along the lines suggested in Section~\ref{ss:vl}. 
Making sense of nonlinear rules (and morphisms) in the case of abstract graph categories might require additional nontrivial properties of morphisms yet to be elaborated.
This is likely to be nontrivial since we have seen that the category of drags does not have any of the adhesivity properties that have been elaborated along the years to make DPO work.

\subsection{Algebraic approaches}

Several other algebraic approaches, like Agree \cite{DBLP:conf/tagt/HabelMP98}, PBPO \cite{PBPO}, and PBPO$^+$ \cite{PBPO+}, have been defined to overcome some limitations of the DPO approach, such as the ability to erase or clone nodes. 
For us, cloning and erasing can be implemented by drag rewriting.

\subsection{Term graphs}
The literature on term graphs (graphs that represent terms having shared subterms) is extremely rich, as surveyed in \cite{Plump}.
Since the presence of cycles is a distinctive feature of drags---as studied in this paper---we compare them here with term graphs admitting cycles.

The work of Barendregt et al.~\cite{Barendregtetal} is essentially interested in showing that graph rewriting faithfully implements  term rewriting on finite and infinite trees. 
This turns out not always to be the case, in particular because of sharing that inhibits some term derivations, and because of infinite trees (obtained by unravelling graphs with a distinct root vertex). 
The first reason holds likewise for the drags of \cite{DBLP:journals/tcs/DershowitzJ19}, as pointed out there, due to forced sharing. 
It is shown in~\cite{Barendregtetal}, however, that faithfulness  obtains for "weakly regular" (left-linear, non-overlapping) rewrite systems. 
The proof uses complete developments, and that's where weak-regularity comes in by ensuring the absence of non-trivial critical pairs. 
Left-linearity is crucial. 
The authors argue that it could be adapted to non-left-linear rules, but would then require forced sharing to reflect non-left-linear rules, so, the problem would remain for non-linear systems. 
They also assert that a relaxed equality (instead of identity implied by forced sharing) should also be possible, but that this would require changing their definition of a graph with variables (which is inherently linear).

Example 3.8(v) in~\cite{Barendregtetal} demonstrates that  in their formalism rewriting the cyclic graph $G=v:I(v)$ with the rule $I(x) \ra x$ yields $G$ again.
In our setup, this is the case when there is no right-hand side extension (the switchboard is not well-behaved), so $G$ does not rewrite. 
The phenomenon is indeed the very same, since the computation of a product with a non-well-behaved switchboard would cycle if allowed; hence the computation of $G$ is non-terminating in both cases.

In the work of \cite{CorradiniGadducci99},
graphs are arbitrary, with a list of variables and a list without repetition of rooted vertices. Variables are ignored for defining morphisms. Composition is defined by gluing sprouts and roots (in one direction: switchboards are one-way only) of the same index, provided they are in equal number. This is a very restricted mechanism compared to ours. Feedback is the gluing of the last variable with the last root of a graph, provided they are in equal number. Rules are acyclic graphs with two roots, $l$ and $r$. The rule format is therefore more general than Barendregt's, but still quite limited with respect to drags. 

So, our model is much more expressive than other term graph formalisms: non-linear possibly cyclic left-hand sides, powerful composition operator, morphisms taking care of variables, whether linear or non-linear, absence of dangling edges, and adequacy of drag rewriting for non-linear term rewriting in this strictly more powerful formalism.

\subsection{Patches}

A framework similar to ours has been recently developed by Overbeek and Endrullis \cite{Endrullis}, but which is designed for graphs whose vertices have variable arity.
For composition, they employ "patches"---a device similar to our switchboard, which adds connecting edges between the two components.
Likewise, they have an analogue to roots cum sprouts (rather like the snaps suggested in Section~\ref{sec:sym}), which
allows one to constrain the permitted shapes of subgraphs around a match for a left-hand side, and also to specify how the subgraphs should be transformed.
Transformations include rearrangement, deletion, and duplication of edges.
PBPO$^+$ \cite{PBPO+}, which was developed in a categorical framework, can be seen as a conceptual successor to patches and has been proposed as a unifying notion. 

\subsection{Graphs with interfaces}

The idea of building graphs using some kind of composition operation, by gluing some selected nodes of the graphs involved, which are considered interfaces, is already quite old, going back to the work of Bauderon and Courcelle \cite{BC87}. 
In the context of DPO, graphs with interfaces and their transformation have been studied by Bonchi, Corradini, Gadducci, and their colleagues; see, among others, \cite{CG99,Gadducci07,DBLP:conf/esop/BonchiGKSZ17}. 
The main difference is that they only consider sequential composition;
they don't consider the possibility that sprouts of one drag are connected to roots of another \emph{and} vice versa. In our terms, this means that in this case composition is limited to \emph{one-way switchboards} ~\cite{DBLP:journals/tcs/DershowitzJ19}.

\subsection{String diagrams} 
String diagrams are a restricted graphical syntax for representing computational models used in various fields, including programming language semantics, circuit theory, and control theory. 
Mathematically, string diagrams are the terms of symmetric monoidal theories, which generalize algebraic theories in a way that  makes them suitable for expressing resource-sensitive systems in which variables cannot be copied or discarded at will. String diagrams enjoy a restricted composition operator in the sense that it is based on one-way switchboards (as in \cite{CorradiniGadducci99}).
Rewriting of string diagrams is defined as a specific instance of \dpo rewriting with interfaces (DPOI), called convex rewriting, for a category of labeled hypergraphs that correspond to string diagrams \cite{DBLP:journals/jacm/BonchiGKSZ22,DBLP:journals/corr/abs-2104-14686}. 

\section{Conclusion}

In this work, we have completely revamped the preliminary drag model of~\cite{DBLP:journals/tcs/DershowitzJ19}.
In particular, the arrangement of roots in this work is much more useful, variables provide much more flexible sharing, the notion of rewriting rule is much more general, and we end up with a much better algebra. 
Repeated (nonlinear) variables are used to restrict matches to equimorphic subgraphs.
Distinct drags may now share vertices, which is economical when rewriting.
The proposed model encompasses term rewriting as a special case, in stark contrast to the prior work.
Composition is facilitated by a new notion of step-by-step wiring of connections from sprouts to roots.
We have also developed a pleasing algebra of drags with sum and product.
These advances are supported by new, original notions of morphisms for drags, which allow us to precisely relate  drag rewriting based on composition with drag rewriting based on DPO, and even to slightly generalize DPO when applied to drags.
We observe that seemingly minor changes in the details of the formalism have had far-flung effects and have required significant effort to put all the pieces together in place. In this respect, indegree preservation happened to be the key property that made it possible.

The drag framework was conceived so as to apply to a specific category of graphs, namely drags, and to generalize the standard term rewriting and dag models to drags. As a consequence, drags are graphs equipped with specific vertices, called sprouts, labeled with variables, while the other, internal vertices are labeled by function symbols equipped with an arity that specifies the number of their outgoing edges. In addition, vertices are equipped with roots that provide them with the potential for creating new edges.

The major originality of the drag model is to base the matching of a given drag $D$ with respect to a left-hand side of rule $L$ on the existence of a pair made of a context drag $C$ and a switchboard $\xi$ such that $D$ is the product of $L$ and $C$ with respect to $\xi$. In this view, the switchboard $\xi$ maps a sprout $s$ of each drag to a rooted vertex $r$ of the other drag, provided $r$ has at least as many roots as the number of incoming edges and roots of $s$.
Computing the product amounts to redirecting to $r$ all edges incoming to $s$ and removing from $r$ an equal number of roots, an operation that leaves the indegree of $r$ unchanged. Rewriting amounts then to replacing $L$ by $R$, that is, to computing the new drag resulting from the product of the context $C$ with the right-hand side $R$ with respect to the switchboard $\xi$. This assumes the existence of an injective mapping from the roots of $L$ to the roots of $R$.

We have indeed succeeded, inasmuch as our new drag model appears to generalize the term and dag rewriting models very smoothly, something that  our former drag model~\cite{DBLP:journals/tcs/DershowitzJ19} could not do. Furthermore, it even generalizes the term and dag rewriting models when applied to terms and dags by having two new built-in capabilities: sharing and cloning. By this we mean that we are able to specify \emph{formally} at each rewrite step which subdrags should be shared and which should be duplicated.

The most widely accepted and used graph-rewriting model is \dpo. While \dpo was conceived so as to apply to various categories of graph structures, namely the adhesive categories, its expressivity is limited by the absence of variables, one consequence of which is the infeasibility of cloning. 
Furthermore, we have seen that drags lack some categorical properties that are essential for applying DPO, properties whose absence follows from the preservation of indegrees by monomorphisms. 
We could of course weaken that property so that indegrees could grow along monomorphisms, but the price to pay would be heavy: 
The two different views of matching, by the existence of a monomorphism as for DPO and by the existence of a rewriting extension of the left-hand side of rule, as for drags, would no longer coincide.

A natural question then follows: Can graphs be equipped with variables, and can these variables be used within the \dpo model?
This question was actually raised long ago~\cite{DBLP:conf/gg/Parisi-PresicceEM86}, but to date no satisfactory answer has been proffered~\cite{DBLP:conf/birthday/Hoffmann05}, despite several attempts, most notably that of~\cite{DBLP:journals/mscs/HabelP96}.
We have given here a general answer to that question for the category of drags, thanks to the drag's notion of a variable being a one-way channel, to the notion of switchboard---which allows one to compose graphs in a very general way, and to a notion of morphism (and monomorphism) for non-ground drags. 
Matching a left-hand side of rule can then be defined either via composition or via the existence of a monomorphism in the obtained category, while rewriting can be defined by replacement or by a double pushout. 
Moreover, a slight generalization of DPO is suggested that inherits the cloning and sharing capabilities of composition based drag rewriting.

Even more interesting are the following related questions: Can composition be defined for arbitrary graphical structures, or---more precisely---for arbitrary objects belonging to some adhesive category? 
Is adhesivity required for that purpose? 
Are variables needed for that purpose?


\bibliographystyle{alphaurl} 
\bibliography{main}

\end{document}